\documentclass[oneside,11pt,table]{book}

\newif\ifIntro\Introfalse
 \Introtrue
\newif\ifCPDS\CPDSfalse
 \CPDStrue
\newif\ifMSO\MSOfalse
 \MSOtrue
\newif\ifUnderapprox\Underapproxfalse
 \Underapproxtrue
\newif\ifPDL\PDLfalse
 \PDLtrue

\usepackage{hyperref}

\usepackage{graphicx}
\usepackage[pdflatex,recompilepics=false]{gastex}
\usepackage{qtree}
\usepackage{amsmath,amssymb}
\usepackage{xspace}
\usepackage{dsfont}
\usepackage{mathrsfs}
\usepackage[mathscr]{euscript}
\usepackage{bbold}
\usepackage{enumerate}
\usepackage{enumitem}
\usepackage{subfigure}
\usepackage[thmmarks,framed]{ntheorem}
\usepackage{framed}
\usepackage{xcolor}
\usepackage{tensor}
\usepackage{pifont}
\usepackage{caption}
\usepackage{subfigure}
\usepackage{booktabs}
\usepackage{stmaryrd}
\usepackage{array}
\usepackage[T1]{fontenc}
\usepackage{hhline}

\DeclareMathAlphabet{\mathpzc}{OT1}{pzc}{m}{it}
\usepackage{upgreek}

\usepackage{ifpdf}
\ifpdf
\usepackage{tikz}
\usetikzlibrary{shadows}
\RequirePackage[calcwidth]{titlesec}
\usepackage[textwidth=3cm]{todonotes}

\tikzstyle{thmbox} = [rectangle, rounded corners, draw=black,
  fill=lightgray, inner sep=4pt]

\tikzstyle{exbox} = [rectangle, rounded corners, 
  fill=excol, inner sep=4pt]

\fi

\definecolor{lightgray}{gray}{0.9}
\definecolor{seccol}{rgb}{0,0.8,0.3}

\colorlet{MyColorOne}{blue!40}

\newcommand{\lightercolor}[3]{
    \colorlet{#3}{#1!#2!white}
}

\newcommand{\darkercolor}[3]{
    \colorlet{#3}{#1!#2!black}
}

\lightercolor{MyColorOne}{17}{excol}
\darkercolor{MyColorOne}{50}{MyColorOneDark}

\newcounter{mycounter}
\numberwithin{mycounter}{chapter}

\theoremclass{Theorem}
\theoremstyle{break}

\theoremheaderfont{\normalfont\bfseries}
\theorembodyfont{\slshape}
\theoremstyle{plain}
\theoremseparator{.}

\newtheorem{theorem}[mycounter]{Theorem}
\newshadedtheorem{mytheorem}[mycounter]{Theorem}
\newshadedtheorem{mylemma}[mycounter]{Lemma}

\theorembodyfont{\upshape}

\newtheorem{lemma}[mycounter]{Lemma}
\newtheorem{proposition}[mycounter]{Proposition}
\newtheorem{corollary}[mycounter]{Corollary}

\theorembodyfont{\normalfont}
\newtheorem{remark}[mycounter]{Remark}

\theoremsymbol{\ensuremath{\diamondsuit}}
\newtheorem{example}[mycounter]{Example}
\newtheorem{exercise}[mycounter]{Exercise}

\theorembodyfont{\itshape}
\theoremsymbol{\ensuremath{\diamondsuit}}
\newtheorem{definition}[mycounter]{Definition}

\theoremheaderfont{\normalfont\bfseries}
\theorembodyfont{\upshape}
\theoremseparator{.}
\theoremsymbol{}
\newtheorem{claim}[mycounter]{Claim}

\theoremheaderfont{\itshape}
\theorembodyfont{\upshape}
\theoremstyle{nonumberplain}
\theoremseparator{.}
\theoremsymbol{\rule{1ex}{1ex}}
\newtheorem{proof}{Proof}

\usepackage[Sonny]{fncychap}
\ChNameVar{\Large\sf}
\ChNumVar{\Huge}
\ChTitleVar{\LARGE\sf}

\titleformat{\section}[hang]{\sffamily}
{\Large\thesection}{0pt}{~~\Large}[{\titlerule[0.5pt]}]

\titleformat{\subsection}[hang]{\sffamily\bf}
{\color{black}\large\thesubsection}{0pt}{~~\large}

\makeatletter

\renewcommand{\paragraph}{\@startsection{paragraph}{6}{\z@}{2ex}{-0.7em}{\bfseries\normalsize}}

\makeatother

\makeatletter
\newcommand{\xRightarrow}[2][]{\ext@arrow 0359\Rightarrowfill@{#1}{#2}}
\makeatother

\begin{document}

\title{{Non-Sequential Theory of Distributed Systems}\\[1ex]\small Lecture MPRI M2}
\author{Benedikt Bollig and Paul Gastin \\[1ex]
Universit{\'e} Paris-Saclay, CNRS, ENS Paris-Saclay \\[1ex]
Laboratoire M\'ethodes Formelles, 91190, Gif-sur-Yvette, France}
\date{\today}
\maketitle


\thispagestyle{empty}
\paragraph{Summary.}
These are the lecture notes of the graduate course
\emph{Non-Sequential Theory of Distributed Systems}
that is regularly given in the MPRI programme
(Parisian Master of Research in Computer Science):
\begin{center}
\url{https://wikimpri.dptinfo.ens-cachan.fr/doku.php?id=cours:c-2-8-1}
\end{center}

The lecture covers basic automata-theoretic concepts and logical formalisms for the modeling and verification of concurrent and distributed systems. Many of these concepts naturally extend the classical automata and logics over words, which provide a framework for modeling sequential systems. A distributed system, on the other hand, combines several (finite or recursive) processes, and will therefore be modeled as a collection of (finite or pushdown, respectively) automata. A~crucial parameter of a distributed system is the kind of interaction that is allowed between processes. In this lecture, we focus on the message-passing paradigm. In general, communication in a distributed system creates complex dependencies between events, which are hidden when using a sequential, operational semantics.

The approach taken here is based on a faithful preservation of the dependencies of concurrent events. That is, an execution of a system is modeled as a partial order, or graph, rather than a sequence of events. This has to be reflected in high-level languages for formulating requirements to be met by a distributed system. Actually, specifications for distributed systems are, by nature, non-sequential. They should be interpreted directly on the partial order underlying a system execution, rather than on an (arbitrary) linearization of it. It is worth mentioning that using specifications working on linearizations are often the reason for undecidability, as they may assume synchronization that actually cannot happen. We present classical specification formalisms such as monadic second-order (MSO) logic and temporal logic, interpreted over partial-orders or graphs, as well as (high-level) rational expressions. We compare the expressive power of automata and logic and give translations of specifications into automata (synthesis and realizability). Moreover, we consider the satisfiability (Is a given specification consistent?) and the model-checking problem (Does a given distributed system satisfy its specification?). For both problems, we present elementary techniques (based on tree interpretations and tree automata) that yield decision procedures with optimal complexity.


\frontmatter
\tableofcontents
\mainmatter


\newcommand{\Device}{\ensuremath{\mathsf{D}}\xspace}
\newcommand{\Display}{\ensuremath{\mathsf{L}}\xspace}
\newcommand{\Pa}{\ensuremath{\mathsf{P1}}\xspace}
\newcommand{\Pb}{\ensuremath{\mathsf{P2}}\xspace}
\newcommand{\on}{\ensuremath{\mathit{on}}\xspace}
\newcommand{\off}{\ensuremath{\mathit{off}}\xspace}
\newcommand{\noop}{\ensuremath{\mathsf{noop}}\xspace}

\ifIntro

\chapter{Introduction}

Before we go into formal definitions, we give some examples of the kind of problems that are treated in these lecture notes.

\section{Synthesis and Control}

\subsection{A simple protocol}

We consider systems that consist of a fixed finite number of processes that are connected via communication media. Consider the following architecture and assume that communication is synchronous.

\begin{center}
\fbox{\parbox{0.6\textwidth}{
\begin{center}
 \includegraphics[page=1,scale=0.4]{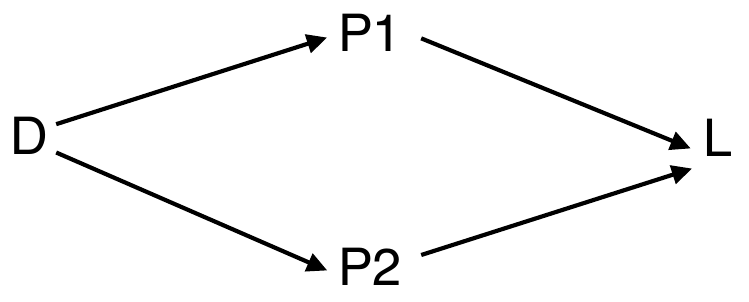}
\end{center}
}}
\end{center}

Consider a process $\Device$ (the device), which communicates with two processes, $\Pa$ and $\Pb$, sending them regularly its current status ($\on$ or $\off$). Every status message is forwarded, by $\Pa$ or $\Pb$, to $\Display$ (a lamp, or display).
The latter is either \on or \off and, upon reception of a message, may change (or not) the current status.
For simplicity, we will assume rendez-vous (i.e., handshake) communication.
A naive implementation looks as follows:

\fbox{\parbox{\textwidth}{
\begin{center}
 \includegraphics[page=2,scale=0.4]{Chap1-figs.pdf}
\end{center}
}}

However, this may result in the following communication scenario, which yields the ``wrong'' output:

\fbox{\parbox{\textwidth}{
\begin{center}
 \includegraphics[page=3,scale=0.4]{Chap1-figs.pdf}
\end{center}
}}

The problem is that \Display does not know which of the last two messages that it receives is the latest one.

Suppose we are not allowed to change the architecture of the system.
We are only allowed to add \emph{local and deterministic} controllers to each process that add additional message contents.
In particular, a controller should not block any of the possible actions of the given system.
Is there still a way to control the protocol in a way such that \Display shows the correct result?

\newpage
\subsection{A first solution.}

A first, rather obvious solution would consist in using timestamps. Strictly increasing timestamps are sent by \Device along with a status message, forwarded by \Pa and \Pb, and then compared by \Display with its latest knowledge.

\fbox{\parbox{\textwidth}{
\begin{center}
\includegraphics[page=4,scale=0.39]{Chap1-figs.pdf}
\end{center}
}}

It is easy to see that the protocol is correct in the sense that \Display does not update its display based on outdated information. The following scenario shows that the error from the previous solution is indeed avoided.

\fbox{\parbox{\textwidth}{
\begin{center}
\includegraphics[page=5,scale=0.4]{Chap1-figs.pdf}
\end{center}
}}

\newpage
\subsection{Towards finite-state solutions}

The previous solution uses infinitely many states, since time stamps are strictly increasing.
This may not be adequate in the realm of \emph{reactive} processes, which are supposed to run forever.
However, finding a finite-state solution is difficult. It is actually impossible for the above specification.

Fortunately, there are guaranteed solutions in the case where every message is immediately followed by an acknowledgment.
So, let us assume the following architecture:

\begin{center}
\fbox{\parbox{0.6\textwidth}{
\begin{center}
 \includegraphics[page=6,scale=0.4]{Chap1-figs.pdf}
\end{center}
}}
\end{center}

Moreover, the allowed communication scenarios look as follows, i.e., every message is \emph{immediately} followed by an acknowledgment (which can be augmented with additional messages):

\fbox{\parbox{\textwidth}{
\begin{center}
\includegraphics[page=7,scale=0.4]{Chap1-figs.pdf}
\end{center}
}}

In this particular case, we get a finite-state local controller. However, it is not always easy to come up with a solution.

\begin{exercise}
Try to find a finite-state controller.
\end{exercise}

The process of finding a controller is actually error-prone. So, it is natural to ask the following question:

\begin{center}
\fbox{\emph{Is there an automatic way to generate a controller from a specification?}}
\end{center}
The answer is YES, provided that the specification satisfies some properties.

\newcommand{\actdaon}{\langle \Device \stackrel{\on}{\rightleftarrows} \Pa \rangle}
\newcommand{\actdbon}{\langle \Device \stackrel{\on}{\rightleftarrows} \Pb \rangle}
\newcommand{\actdaoff}{\langle \Device \stackrel{\off}{\rightleftarrows} \Pa \rangle}
\newcommand{\actdboff}{\langle \Device \stackrel{\off}{\rightleftarrows} \Pb \rangle}

\newcommand{\actadon}{\langle \Pa \stackrel{\on}{\rightleftarrows} \Display \rangle}
\newcommand{\actbdon}{\langle \Pb \stackrel{\on}{\rightleftarrows} \Display \rangle}
\newcommand{\actadoff}{\langle \Pa \stackrel{\off}{\rightleftarrows} \Display \rangle}
\newcommand{\actbdoff}{\langle \Pb \stackrel{\off}{\rightleftarrows} \Display \rangle}

\newcommand{\Lspec}{L_{\mathit{spec}}}

In our example, the specification $\Lspec \subseteq \Sigma^\ast$ could be a word language over the following alphabet:
\[
\begin{array}{rl}
\Sigma = & \{\actdaon\,,\,\actdbon\,,\,\actdaoff\,,\,\actdboff\}\\
\cup & \{\actadon\,,\,\actbdon\,,\,\actadoff\,,\,\actbdoff\}\\[0.5ex]
\cup & \{\on,\off,\noop\}
\end{array}
\]

Let $\Sigma_{\Pa} \subseteq \Sigma$ be the subalphabet containing those actions involving $\Pa$. Moreover, let $\Sigma_{\Pb} \subseteq \Sigma$ contain those with \Pb, and so on. In particular, $\{\on,\off,\noop\} \subseteq \Sigma_\Display$. With this, $\Lspec$ is the set of words $w \in \Sigma^\ast$ satisfying the following:
\begin{itemize}
\item[(R1)] The projection of $w$ to $\Sigma_\Pa$ is contained in \[(\actdaon \actadon + \actdaoff \actadoff)^\ast\,.\]

\item[(R2)] The projection of $w$ to $\Sigma_\Pb$ is contained in \[(\actdbon \actbdon + \actdboff \actbdoff)^\ast\,.\]

\item[(R3)] The projection of $w$ to $\Sigma_\Display$ is contained in \[\Bigl((\actadon+\actbdon+\actadoff+\actbdoff)(\on+\off+\noop)\Bigr)^\ast\,.\]

\item[(R4)] ``The display is updated iff the last (previous) status message emitted by \Device that has already been followed by a forward was not yet followed by a corresponding update by \Display.''
\end{itemize}

\begin{exercise}
Formalize the requirement (R4) in terms of a finite automaton or an MSO formula.
\end{exercise}

Here are some example words to illustrate $\Lspec$:
\begin{itemize}
\item $\actdaon\actdboff\actbdoff\off\actadon\noop ~~\in~~ \Lspec$\,,
\item $\actdaon\actadon\actdaon\on ~~\not\in~~ \Lspec$,
\item $\actdaon\actadon\actdboff\on\actbdoff\off ~~\in~~ \Lspec$\,,
\item $\actdaon\actadon\on\actdboff\actbdoff\off ~~\in~~ \Lspec$\,.

\end{itemize}

Obviously, $\Lspec$ is a regular language. Moreover, it is closed under \emph{permutation rewriting of independent events}. The latter means that changing the order of neighboring independent actions in a word does not affect membership in $\Lspec$.
This seems natural, since the order of independent event cannot be enforced by a distributed protocol.
Here, two actions are said to be \emph{independent} if they involve distinct processes. For example,
\begin{itemize}
\item  $\actdaon$ and $\on$ are independent,
\item  $\actdaon$ and $\actbdon$ are independent,
\item  $\actdaon$ and $\actdboff$ are \emph{not} independent,
\item  $\actdbon$ and $\actbdon$ are \emph{not} independent.
\end{itemize}

\begin{exercise}
Show that $\Lspec$ is closed under permutation rewriting.
\end{exercise}

The following is a fundamental result of concurrency theory (yet informally stated):
\begin{center}
\fbox{\parbox{\textwidth}{
\underline{Theorem [Zielonka 1987]:}\\[1ex]
\emph{Let $L$ be a regular set of words that is closed under permutation rewriting of independent events. There is a deterministic finite-state distributed protocol that realizes $L$.}
}}
\end{center}

Thus, the specification $\Lspec$ could indeed be realized as a distributed program.

There are, however, regular specifications that are not realizable (and, therefore, are not closed under permutation rewriting). Consider the language
\[L= (\actdaon\actbdon)^\ast\,.\]
Though $L$ is regular, it is not closed under permutation rewriting. Even worse, the closure under permutation is not regular anymore. Specification $L$ says that there are as many messages from $\Device$ to $\Pa$ as from $\Pb$ to $\Display$. Intuitively, it is clear that this cannot be realized by a finite-state system in a distributed fashion: There is no communication going on between $\Device$/$\Pa$ on the hand, and $\Pb$/$\Display$ on the other hand.

\section{Modeling behaviors as graphs}

\subsection{Partial orders}

Requirement (R4) from the last section is somehow awkward to write down. The reason is that a word over $\Sigma$ imposes an ordering of events that, actually, are not causally related in the corresponding execution.
When we say ``last'', this refers to the ``last'' position in the word. Consider, for example, the word
\[
\begin{array}{rl}
w= & \actdaon\actadon\actdboff\on\\
& \actdaon\actadon\on\actbdoff\noop \in \Lspec\,.
\end{array}\]
The ``last position'' right before the \on in the first line is actually in no way related to \on. So, it is not natural (and not needed) to include it in what we mean by ``last''.
In our reasoning, a more relaxed ordering has to be recovered from the word ordering. So, why not directly reason about the causal order as it is imposed by a distributed system?

In the following, we do not consider an execution of a system as a word, i.e., a total order, but rather as a partial order. The partial order is already suggested by the message diagrams that we used to argue about our distributed protocols. Consider the execution below, whose partial order is represented by its Hasse diagram. It corresponds to the above word $w$.

\begin{center}
\fbox{
\includegraphics[page=8,scale=0.3]{Chap1-figs.pdf}
}
\end{center}

Suppose this partial order is denoted by $\le$. Then, (R4) can be conveniently rephrased (as we do not need the part ``... that has already been followed by ...'' anymore):
\begin{itemize}
\item[(R4')] ``When \Display performs \on (\off, respectively), then the last (wrt.\ $\le$) status message sent by \Device should also be \on (\off, respectively). Moreover, a display operation should be \noop iff there has already been an acknowledgement between the latest status message and that operation.''
\end{itemize}

This can very easily be expressed in monadic second-order (MSO) logic over partial orders, using the partial order $\le$.

\begin{exercise}
Give an MSO formula for (R4').
\end{exercise}

An advantage of MSO logic that is directly interpreted over partial orders is the following theorem:
\begin{center}
\fbox{\parbox{\textwidth}{
\underline{Theorem [Thomas 1990]:}\\[1ex]
\emph{Let $L$ be an MSO-definable set of partial orders. There is a finite-state distributed protocol that realizes $L$.}
}}
\end{center}

\newcommand{\sreq}{\langle!\mathsf{req}\rangle}
\newcommand{\rreq}{\langle?\mathsf{req}\rangle}
\newcommand{\sack}{\langle!\mathsf{ack}\rangle}
\newcommand{\rack}{\langle?\mathsf{ack}\rangle}
\newcommand{\callact}{\langle\mathsf{call}\rangle}
\newcommand{\retact}{\langle\mathsf{ret}\rangle}

\subsection{Reasoning about recursive processes}

The previous discussion somehow motivates the naming of our lecture.
However, in our modeling we will go a step further.
One single behavior is not just a partial order, but an (acyclic) graph, which is more general.
The edges reflect causal dependencies, but they provide even more information.
For example, they may connect a procedure call with the corresponding return, or the sending of a message with its receive.

Consider a system of two processes, \Pa and \Pb, connected by two unbounded FIFO channels (one in each direction).
From time to time, \Pa sends requests to \Pb. The latter, when receiving a request, calls a procedure, performs some internal actions, and returns from the procedure, before it sends an acknowledgment. In the scope of a procedure, it may call several subprocedures. Thus, \Pa performs actions from the alphabet
$\Sigma_1=\{\sreq,\rack\}$
and $\Pb$ performs actions from
$\Sigma_2=\{\rreq,\sack,\callact,\retact\}$.
Let $\Sigma = \Sigma_1 \cup \Sigma_2$.

Let us try to model the protocol in terms of a word language $L$.
It should say that, whenever a request is received, \Pb should start a subroutine and send the acknowledgment immediately after returning from this subroutine. Thus, we shall have \[w_1 = \sreq\,\underline{\rreq}\,\callact\, \callact\, \retact\, \callact\, \retact\, \retact\, \underline{\sack}\, \rack \in L\,.\]
On the other hand, we should have
\[w_2=\sreq\,\underline{\rreq}\,\callact\, \callact\, \retact\, \callact\, \retact\,\underline{\sack}\, \rack\, \retact \not\in L\,.\]

But how do we express this in a specification language over words such as MSO logic? Unfortunately, there is no such formula, since $L$ is a non-regular property.

\begin{exercise}
Prove that $L$ is not regular.
\end{exercise}

The solution is to equip a word with additional information, which allows us to ``jump'' from a call to its \emph{associated} return position. In other words, we add an edge from a call to the corresponding return:

Then, $w_1$ corresponds to:
\begin{center}
\fbox{
\includegraphics[page=9,scale=0.35]{Chap1-figs.pdf}
}
\end{center}
Moreover, $w_2$ corresponds to:
\begin{center}
\fbox{
\includegraphics[page=10,scale=0.35]{Chap1-figs.pdf}
}
\end{center}
From this point of view, we can now express our property easily in a suitable MSO logic:
\[\forall x\,\rreq(x) \Rightarrow \exists y_1,y_2,z \Bigl(x \to y_1 \wedge \mathsf{cr}(y_1,y_2) \wedge y_2 \to z \wedge \sack(z)\Bigr)\]
It is for similar obvious reasons that, in an asynchronous setting, we connect a send with a receive event.

\section{Underapproximate Verification}

The lecture is concerned with distributed systems with a fixed architecture: There is a finite set of processes, which are connected via some communication media. We will consider stacks, FIFO channels, and bags (with the restriction that stacks connect a process with itself, with the purpose of modeling recursion).%
\footnote{We do not consider \emph{lossy} channels. Cf. Cours 2.9.1.}
A priori, we assume that all media are unbounded.
To get decidability or expressivity results, however, we may sometimes impose a bound on channels or stacks. Recall that, in the introductory example, we assumed synchronous communication, which roughly corresponds to FIFO channels with capacity 0.

The following figure illustrates one possible architecture (source: Aiswarya Cyriac's thesis):
\begin{center}
\fbox{
\includegraphics[page=11,scale=1.2]{Chap1-figs.pdf}
}
\end{center}

In the following, we will actually depict an architecture as a finite graph with directed edges of three types, depending on the type of the communication medium they represent. We will follow the following convention:
\begin{center}
\fbox{
\includegraphics[page=12,scale=0.4]{Chap1-figs.pdf}
}
\end{center}

The most basic verification question is the \emph{reachability problem}, i.e., to ask whether some control state of some process is reachable. Let us examine the decidability status of the reachability problem for the following architectures:
\begin{center}
\fbox{
\includegraphics[page=13,scale=0.4]{Chap1-figs.pdf}
}
\end{center}
\begin{itemize}
\item[a)] Undecidable. A system can simulate a Turing machine TM as follows: Via $c_1$, process 1 (on the left) sends the current configuration of the TM to process 2. Process 2 only sends every message that it receives immediately back to 1. When 1 receives the configuration, it locally modifies it simulating a transition of TM.

\item[b)] Undecidable. Similarly to a), the process sends the current configuration through the channel $c$. When receiving a configuration from $c$, it modifies it locally.

\item[c)] Undecidable. This can be shown by reduction from PCP (Post's correspondence problem). Process 1 guesses a solution (a sequence of indices, and it sends the corresponding words via channels $c_1$ and $c_2$, respectively. Process 2 will then just check if the sequences send through $c_1$ and $c_2$ coincide. To do this, it reads alternately from both channels and checks whether both symbols coincide.

\item[d)] Decidable. This case can be reduced to synchronous communication and, therefore, reachability in a finite-state system.

\item[e)] Decidable. This case corresponds to emptiness of pushdown automata.

\item[f)] Undecidable. One can easily simulate a two-counter machine. Equivalently, we may use the concatenation of two stacks to simulate the unbounded tape of a Turing machine.

\item[g)] Undecidable.  We may use a reduction from the intersection-emptiness
problem for two pushdown automata $\mathcal{A}_{1}$ and $\mathcal{A}_{2}$.
Process 1 guesses an input word $w$, checks with its stack that $w$ is accepted
by $\mathcal{A}_{1}$ and simultaneously sends $w$ letter by letter to process 2
through the FIFO channel $c$.  Process 2 checks with its stack that $w$ is
accepted by $\mathcal{A}_{2}$.

\item[h)] Undecidable.  We use a reduction from case b).  Process 1 simulates
send transitions through channel $c$.  To simulate a receive transition through
$c$, it puts a token into bag $b_1$, whose value is the message to be received,
say $m$.  Process 1 can then only proceed when it finds an acknowledgment token
in bag $b_2$.  The latter is provided by process 2 after removing $m$ from $c$.
\end{itemize}
We conclude that almost all verification problems are undecidable even for very simple system architectures.

In this lecture, we therefore perform \emph{underapproximate verification}: We restrict the behavior of a given system in a non-trivial way that still allows us to reason about it and deduce correctness/faultiness wrt.\ interesting properties. Let us illustrate some restrictions using some of the undecidable architectures above:

\begin{itemize}
\item In all cases, we may assume that communication media have capacity $B$ (existentially $B$-bounded), for some fixed $B$.

\item In case f), assuming an order on the stacks, we can only pop from the first nonempty stack.

\item In case f), we may also impose a bound on the number of \emph{contexts}. In turn, there are several possible definitions of what is allowed in a context:
\begin{itemize}
\item We can only touch one stack.
\item We can only pop from one stack.
\item Many more ...
\end{itemize}
\end{itemize}
Under all these restrictions, most standard verification problems (even model checking against MSO-definable properties) becomes decidable, with varying complexities.

\begin{center}
\emph{In the lecture, we will take a uniform approach to underapproximate verification.}
\end{center}

\fi 

\newcommand{\Arch}{\ensuremath{\mathfrak{A}}\xspace}
\newcommand{\DS}{\ensuremath{\mathsf{DS}}\xspace}
\newcommand{\Procs}{\ensuremath{\mathsf{Procs}}\xspace}
\newcommand{\Queues}{\ensuremath{\mathsf{Queues}}\xspace}
\newcommand{\reader}{\ensuremath{\mathsf{Reader}}\xspace}
\newcommand{\Stacks}{\ensuremath{\mathsf{Stacks}}\xspace}
\newcommand{\Bags}{\ensuremath{\mathsf{Bags}}\xspace}
\newcommand{\writer}{\ensuremath{\mathsf{Writer}}\xspace}

\newcommand{\Act}{\ensuremath{\Sigma}\xspace}
\newcommand{\CPDS}{\ensuremath{\mathsf{CPDS}}\xspace}
\newcommand{\txtCPDS}{\ensuremath{\textup{CPDS}}\xspace}
\newcommand{\CPDSs}{\ensuremath{\textup{CPDSs}}\xspace}
\newcommand{\inLoc}{\ensuremath{\ell_{\textrm{in}}}\xspace}
\newcommand{\inState}{\ensuremath{\ell_{\textrm{in}}}\xspace}
\newcommand{\FinLocs}{\ensuremath{\mathsf{Fin}}\xspace}
\newcommand{\Locs}{\ensuremath{\mathsf{Locs}}\xspace}
\newcommand{\States}{\ensuremath{\mathsf{States}}\xspace}
\newcommand{\Sys}{\ensuremath{\mathcal{S}}\xspace}
\newcommand{\A}{\ensuremath{\mathcal{A}}\xspace}
\newcommand{\B}{\ensuremath{\mathcal{B}}\xspace}
\newcommand{\C}{\ensuremath{\mathcal{C}}\xspace}
\newcommand{\Trans}{\ensuremath{\mathsf{Trans}}\xspace}
\newcommand{\Val}{\ensuremath{\mathsf{Val}}\xspace}
\newcommand{\TS}{\ensuremath{\mathcal{TS}}\xspace}
\newcommand{\last}{\ensuremath{\mathsf{last}}\xspace}
\newcommand{\dsword}{z}

\newcommand{\Akswcbm}{\ensuremath{\A_\cbm^{k\text{-}\mathsf{sw}}}\xspace}
\newcommand{\Akswsys}{\ensuremath{\A_\Sys^{k\text{-}\mathsf{sw}}}\xspace}
\newcommand{\Akswphi}{\ensuremath{\A_\varphi^{k\text{-}\mathsf{sw}}}\xspace}
\newcommand{\Aksttvalid}{\ensuremath{\A_\mathsf{valid}^{k\text{-}\mathsf{stt}}}\xspace}
\newcommand{\Aksttphi}{\ensuremath{\A_\varphi^{k\text{-}\mathsf{stt}}}\xspace}
\newcommand{\Asttphi}[1]{\ensuremath{\A_\varphi^{#1\text{-}\mathsf{stt}}}\xspace}
\newcommand{\Akstw}{\ensuremath{\A^{k\text{-}\mathsf{stw}}}\xspace}
\newcommand{\Akstwcbm}{\ensuremath{\A_\cbm^{k\text{-}\mathsf{stw}}}\xspace}
\newcommand{\Akstwsys}{\ensuremath{\A_\Sys^{k\text{-}\mathsf{stw}}}\xspace}
\newcommand{\Akstwphi}{\ensuremath{\A_\varphi^{k\text{-}\mathsf{stw}}}\xspace}
\newcommand{\AkstwPhi}{\ensuremath{\A_\Phi^{k\text{-}\mathsf{stw}}}\xspace}
\newcommand{\AkstwPsi}{\ensuremath{\A_\Psi^{k\text{-}\mathsf{stw}}}\xspace}
\newcommand{\Akstwacyclic}{\ensuremath{\A_\mathsf{acyclic}^{k\text{-}\mathsf{stw}}}\xspace}
\newcommand{\Akstwedges}{\ensuremath{\A_\mathsf{edges}^{k\text{-}\mathsf{stw}}}\xspace}

\newcommand{\CBM}{\ensuremath{\mathsf{CBM}}\xspace}
\newcommand{\txtCBM}{\ensuremath{\textup{CBM}}\xspace}
\newcommand{\CBMs}{\ensuremath{\mathsf{CBMs}}\xspace}
\newcommand{\MSC}{\ensuremath{\mathsf{MSC}}\xspace}
\newcommand{\txtMSC}{\ensuremath{\textup{MSC}}\xspace}
\newcommand{\MSCs}{\ensuremath{\mathsf{MSCs}}\xspace}
\newcommand{\MNW}{\ensuremath{\mathsf{MNW}}\xspace}
\newcommand{\MNWs}{\ensuremath{\mathsf{MNWs}}\xspace}
\newcommand{\smscn}{\ensuremath{\mathcal{M}}\xspace}
\newcommand{\mscn}{\ensuremath{\mathcal{M}}\xspace}
\newcommand{\Events}{\ensuremath{\mathcal{E}}\xspace}
\newcommand{\ActLabel}{\lambda}
\newcommand{\pid}{\ensuremath{\mathsf{pid}}\xspace}
\newcommand{\matchrel}{\mathbin{\vartriangleleft}}
\newcommand{\procrel}{\rightarrow}
\newcommand{\elastic}{\dashrightarrow}
\newcommand{\lin}{\ensuremath{\mathsf{lin}}\xspace}
\newcommand{\cbm}{\ensuremath{\mathsf{cbm}}\xspace}
\newcommand{\msc}{\ensuremath{\mathsf{msc}}\xspace}
\newcommand{\width}{\textsf{width}}
\newcommand{\Lang}{L}
\newcommand{\opLang}{\Lang_\textup{op}}
\newcommand{\Lin}{\ensuremath{\mathsf{Lin}}\xspace}
\newcommand{\class}{\ensuremath{\mathcal{C}}\xspace}

\newcommand{\LTL}{\ensuremath{\mathsf{LTL}}\xspace}
\newcommand{\FO}{\ensuremath{\mathsf{FO}}\xspace}
\newcommand{\MSO}{\ensuremath{\mathsf{MSO}}\xspace}
\newcommand{\EMSO}{\ensuremath{\mathsf{EMSO}}\xspace}
\newcommand{\PDL}{\ensuremath{\mathsf{PDL}}\xspace}
\newcommand{\CPDL}{\ensuremath{\mathsf{CPDL}}\xspace}
\newcommand{\IPDL}{\ensuremath{\mathsf{IPDL}}\xspace}
\newcommand{\ICPDL}{\ensuremath{\mathsf{ICPDL}}\xspace}
\newcommand{\LCPDL}{\ensuremath{\mathsf{LCPDL}}\xspace}
\newcommand{\EQICPDL}{\ensuremath{\mathsf{EQ\text{-}ICPDL}}\xspace}
\newcommand{\EQLCPDL}{\ensuremath{\mathsf{EQ\text{-}LCPDL}}\xspace}
\newcommand{\EQCPDL}{\ensuremath{\mathsf{EQ\text{-}CPDL}}\xspace}
\newcommand{\EQPDL}{\ensuremath{\mathsf{EQ\text{-}PDL}}\xspace}
\newcommand{\F}{\mathop{\mathsf{F}\vphantom{a}}\nolimits}
\newcommand{\G}{\mathop{\mathsf{G}\vphantom{a}}\nolimits}
\newcommand{\false}{\ensuremath{\mathsf{false}}\xspace}
\newcommand{\true}{\ensuremath{\mathsf{true}}\xspace}
\newcommand{\Exists}[2]{\langle #1\rangle#2}
\newcommand{\test}[1]{\mathsf{test}({#1})}
\newcommand{\phiCBM}{\ensuremath{\Phi_\cbm}\xspace}
\newcommand{\phiSQ}{\ensuremath{\Phi_\sqcbm}\xspace}
\newcommand{\phiSys}{\ensuremath{\Phi_\Sys}\xspace}

\newcommand{\fequiv}{=}

\newcommand{\say}[1]{\textcolor{red}{#1}}
\newcommand{\csSys}{\Sys_\textup{cs}}
\newcommand{\holds}{\gamma}

\newcommand{\minev}{\mathit{min}}
\newcommand{\maxev}{\mathit{max}}

\newcommand{\shuffle}{\mathbin{{\sqcup}\mathchoice{\mkern-3mu}{\mkern-3mu}{\mkern-3.4mu}{\mkern-3.8mu}{\sqcup}}}

\newcommand{\Oh}{\mathcal{O}}
\newcommand{\poly}{\textsf{poly}}

\newcommand{\myline}{~\\[-1.5ex]\hrule~\hrule~\\[-1ex]}

\ifCPDS

\chapter{Concurrent Processes with Data Structures}

\vspace{-6ex}
Notation is taken from \cite{CG-fsttcs14}.

\vspace{-2ex}
\section{The Model}

\begin{definition}[Architecture] An architecture is a tuple \[\Arch=(\Procs,\DS,\writer,\reader)\]
\begin{itemize}
\item $\Procs$ \quad finite set of \emph{processes}
\item $\DS=\Stacks\uplus\Queues\uplus\Bags$ \quad finite set of \emph{data structures}
\item $\writer: \DS \to \Procs$
\item $\reader: \DS \to \Procs$
\end{itemize}
such that $\writer(s)=\reader(s)$ for all $s\in\Stacks$.
\end{definition}

\myline

\begin{example} Consider the following architecture:
\begin{center}
\scalebox{0.6}{
\unitlength=1.2mm
\gasset{frame=false,AHangle=30,AHlength=1.4,AHLength=1.6}
\begin{gpicture}
\gasset{Nadjust=wh, Nframe = true, AHLength = 2, AHlength =1.8}
\node(D1)(0,20){Stack $d_1$}
\node(D2)(20,20){Queue $d_2$}
\node(D3)(40,20){Queue $d_3$}
\node(D4)(60,20){Bag $d_4$}
\node(P1)(10,0){Process $p_1$}
\node(P2)(50,0){Process $p_2$}
\gasset{curvedepth=-3}
\drawedge(P1,D1){}
\drawedge(P2,D4){}
\drawedge(P1,D3){}
\drawedge[curvedepth = 3](P2,D2){}

\drawedge(D1,P1){}
\drawedge(D4,P2){}
\drawedge[exo=3](D2,P1){}
\drawedge[curvedepth = 3, exo = -3](D3,P2){}
\end{gpicture}
}
\end{center}
We have $\Procs = \{p_1,p_2\}$, $\Stacks=\{d_1\}$, $\Queues=\{d_2,d_3\}$, and $\Bags=\{d_4\}$.
E.g., $\writer(d_1) = \reader(d_1) = p_1$.
\end{example}

\begin{definition}[\txtCPDS] \label{def:cpds}
  A \emph{system of concurrent processes with data structures} (\txtCPDS) over
  $\Arch$ and an alphabet $\Sigma$ is a tuple 
  $\Sys=((\Sys_p)_{p\in\Procs},\Val,\FinLocs)$ where for each $p\in\Procs$, 
  $\Sys_p=(\Locs_p,\Delta_p,\iota_p)$ is the local transition system for process $p$:
  \begin{itemize}
    \item $\Val$ \quad  finite set of \emph{values}
    \item $\Locs_p$ \quad  finite set of \emph{locations}
    \item $\iota_p\in\Locs_p$ \quad \emph{initial location}
    \item $\Locs=\prod_{p\in\Procs}\Locs_p$ \quad set of \emph{global} locations
    \item $\inLoc=(\iota_p)_{p\in\Procs}$ \quad \emph{global initial location}
    \item $\FinLocs\subseteq\Locs$ \quad \emph{global final locations}
    \item $\Delta_p=\Delta_p^{i}\cup\Delta_p^{!}\cup\Delta_p^{?}$ \quad transitions of $p$
    \begin{itemize}
      \item $\Delta_p^{i}$ \quad internal transition \quad $\ell \xrightarrow{{a}}_p \ell'$
      \item $\Delta_p^{!}$ \quad write transition \quad $\ell \xrightarrow{a,d!v}_p \ell'$ with $\writer(d) = p$
      \item $\Delta_p^{?}$ \quad read transition \quad $\ell \xrightarrow{a,d?v}_p \ell'$ with $\reader(d) = p$
    \end{itemize}
    where $\ell,\ell'\in\Locs_p$, $a\in\Act$, $d\in\DS$, and $v\in\Val$.
    
    For a transition $t\in\mathsf{Trans}=\biguplus_{p\in\Procs}\Delta_p$, we write
    $\mathsf{src}(t)=\ell$ (source), $\mathsf{tgt}(t)=\ell'$ (target),
    $\mathsf{lbl}(t)=a$ (label), $\mathsf{ds}(t)=d$ and $\mathsf{val}(t)=v$.
  \end{itemize}
  We let $\CPDS(\Arch,\Act)$ be the set of \CPDSs over \Arch and \Act.
\end{definition}

\begin{example}\label{ex:client-server}
Let $\Arch$ be given by $\Procs=\{p_1,p_2\}$, $\Queues=\{c_1,c_2\}$, $\Stacks=\{s\}$, $\Bags=\emptyset$, with $\writer(c_1)=\reader(c_2)=p_1$ and $\writer(c_2)=\reader(c_1)=p_2$ and $\writer(s)=\reader(s)=p_2$. Moreover, let $\Sigma=\{a,b\}$. Consider the client-server system $\csSys$ over $\Arch$ and $\Sigma$ given as follows:
\begin{center}
\fbox{
\includegraphics[scale=0.45]{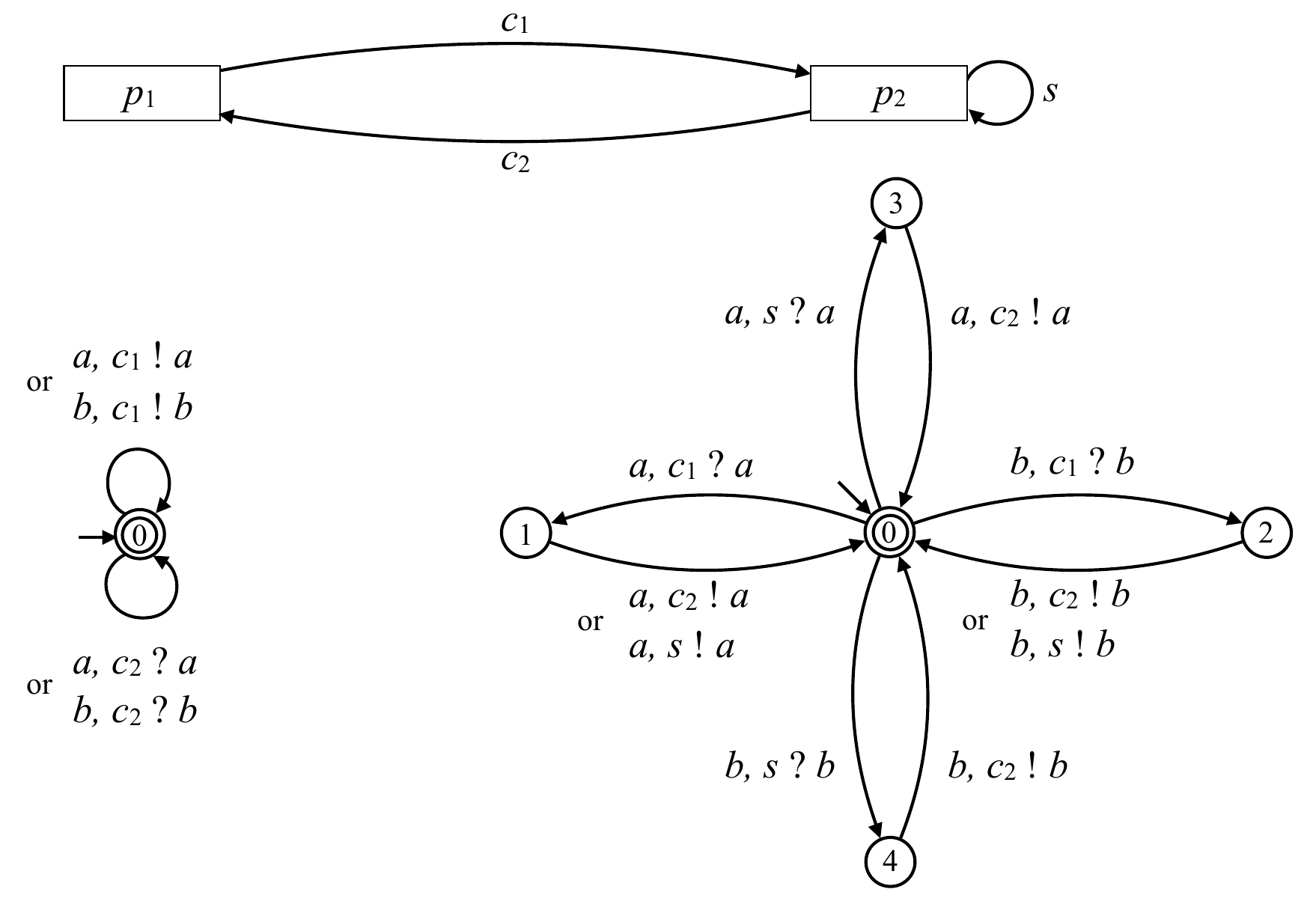}
}
\end{center}
Process $p_1$, the client, sends requests of type $a$ or $b$ to process $p_2$, the server.
The latter may acknowledge the request immediately, or put it on its stack
(either, because it is busy or because the request does not have a high
priority).  At any time, however, the server may pop a task from the stack and
acknowledge it.

We have
$\Locs_{p_2}=\{0,1,2,3,4\}$,
$\iota_{p_2}=0$,
$\FinLocs=\{(0,0)\}$, and
$\Val=\{a,b\}$.
\end{example}

\section{Operational Semantics}

\Sys defines (infinite) transistion system $\A_\Sys=(\States,\Delta,s_\textup{in},F)$ 
over $\Gamma = (\Procs\times\Act)\cup(\Procs\times\Act\times\DS\times\{!,?\})$
\begin{itemize}
\item $\States = \Locs\times(\Val^*)^{\DS}$\\
for $(\overline{\ell},\overline{\dsword}) \in \States$, we denote $\overline{\ell}=(\ell_p)_{p\in\Procs}$
and $\overline{\dsword}=(\dsword_d)_{d\in\DS}$
\item $s_\textup{in} = (\inLoc,(\varepsilon,\ldots,\varepsilon))$
\item $F = \FinLocs \times \{\varepsilon\}^\DS$
\item $\Delta \subseteq \States \times \Gamma \times \States$ \quad global transition
\begin{itemize}
  \item internal transition \quad $(\overline{\ell},\overline{\dsword})\xRightarrow{p,a}(\overline{\ell'},\overline{\dsword})$\\[0.5ex]
  if $\ell_p \xrightarrow{a}_p \ell'_p$ ~~and~~ $\ell'_q=\ell_q$ for all $q\neq p$,
  
  \item write transition \quad $(\overline{\ell},\overline{\dsword})\xRightarrow{p,a,d!}(\overline{\ell'},\overline{\dsword}')$
  \quad if for some $v \in \Val$:\\[0.5ex]
  $\ell_p \xrightarrow{a,d!v}_p \ell'_p$, $\dsword'_d=\dsword_dv$,
  $\ell'_q=\ell_q$ for all $q\neq p$ and
  $\dsword'_{c}=\dsword_{c}$ for all $c\neq d$  
  
  \item read transition \quad $(\overline{\ell},\overline{\dsword})\xRightarrow{p,a,d?}(\overline{\ell'},\overline{\dsword}')$
  \quad if for some $v \in \Val$:\\[0.5ex]
  $\ell_p \xrightarrow{a,d?v}_p \ell'_p$, $\ell'_q=\ell_q$ for all $q\neq p$,
  $\dsword'_{c}=\dsword_{c} \textup{~for~all~} c\neq d$ and
  \\[0.5ex]
  $\left\{\begin{array}{ll}
  d \in \Stacks\textup{:} & \dsword_d=\dsword'_d v \\
  d \in \Queues\textup{:} & \dsword_d=v\dsword'_d \\
  d \in \Bags\textup{:} & \dsword_d=uvw \textup{ and } \dsword'_d=uw
  \text{ for some } u,w \in \Val^\ast
  \end{array}\right.$
\end{itemize}
\end{itemize}

We let $\opLang(\Sys) := L(\A_\Sys) \subseteq \Gamma^\ast$ (discarding the empty word).

\begin{example}
In our client-server system, $\opLang(\csSys)$ contains:
\begin{itemize}
\item $(p_1,a,c_1!)\textcolor{blue}{(p_2,a,c_1?)(p_2,a,c_2!)}(p_1,a,c_2?)$
\item $(p_1,a,c_1!)(p_1,b,c_1!)\textcolor{blue}{(p_2,a,c_1?)(p_2,a,s!)(p_2,b,c_1?)(p_2,b,c_2!)}$\\[1ex]
$(p_1,b,c_2?)\textcolor{blue}{(p_2,a,s?)(p_2,a,c_2!)}(p_1,a,c_2?)$
\end{itemize}
\end{example}

\begin{exercise}
Show that $\opLang(\csSys)$ is not regular.
\end{exercise}

\subsection{Nonemptiness/Reachability Checking}

For an architecture $\Arch$ and an alphabet $\Sigma$, consider the following problem:

\begin{center}
\begin{tabular}{ll}
\toprule
{{\sc{Nonemptiness}}$(\Arch,\Sigma)$}:\\
\midrule
{Instance:} & \hspace{-4em}$\Sys \in \CPDS(\Arch,\Sigma)$\\[0.5ex]
{Question:} & \hspace{-4em}$\opLang(\Sys) \neq \emptyset$\,?\\
\bottomrule
\end{tabular}
\end{center}

\myline

\begin{mytheorem}
Let $\Arch$ be any of the following architectures: a, b, c, f, g, h.\\
Then, {\sc{Nonemptiness}}$(\Arch,\Sigma)$ is undecidable.
\end{mytheorem}
\begin{center}
\fbox{
\includegraphics[page=13,scale=0.4]{Chap1-figs.pdf}
}
\end{center}

The following table summarizes some special cases:
\begin{center}
\begin{tabular}{lllll}
\toprule
$\Arch$ & automata type & CBM & \sc{Nonemptiness}$(\Arch,\Sigma)$\\
\midrule
$|\Procs| = 1$ & finite automaton & word & decidable\\
$|\DS| = 0$\\
\midrule
$|\Procs| = |\DS| = 1$ & (visibly) & nested word & decidable\\
$\DS = \Stacks$ &  pushdown automaton\\
\midrule
$|\Procs| = 1$ & multi-pushdown & multiply & undecidable\\
$|\DS| \ge 2$ & automaton & nested word\\
$\DS = \Stacks$\\
\midrule
$\DS = \Bags$ & $\approx$ Petri net & & decidable\\
\midrule
$|\Procs| \ge 2$ & message-passing & message & undecidable\\
$\DS = \Queues$ & automaton & sequence\\
$= (\Procs \times \Procs) \setminus \textup{Id}$ & & chart (MSC)\\
where\\
$\textup{Id} = \{(p,p) \mid p \in \Procs\}$\\[0.5ex]
$\writer(p,q) = p$\\
$\reader(p,q) = q$\\
\bottomrule
\end{tabular}
\end{center}

\section{Graph Semantics}\label{sec:graph-semantics}

\begin{example}\label{ex:cbm}
Let us represent behaviors as graphs. We start with an example.
The following graph will be in the language of $\csSys$.
The source and the target of an edge represent the exchange of a value through some data structure.
In the example, their labeling ($a$ or $b$) is the same.
However, this is not always the case.
\begin{center}
\fbox{
\includegraphics[scale=0.4]{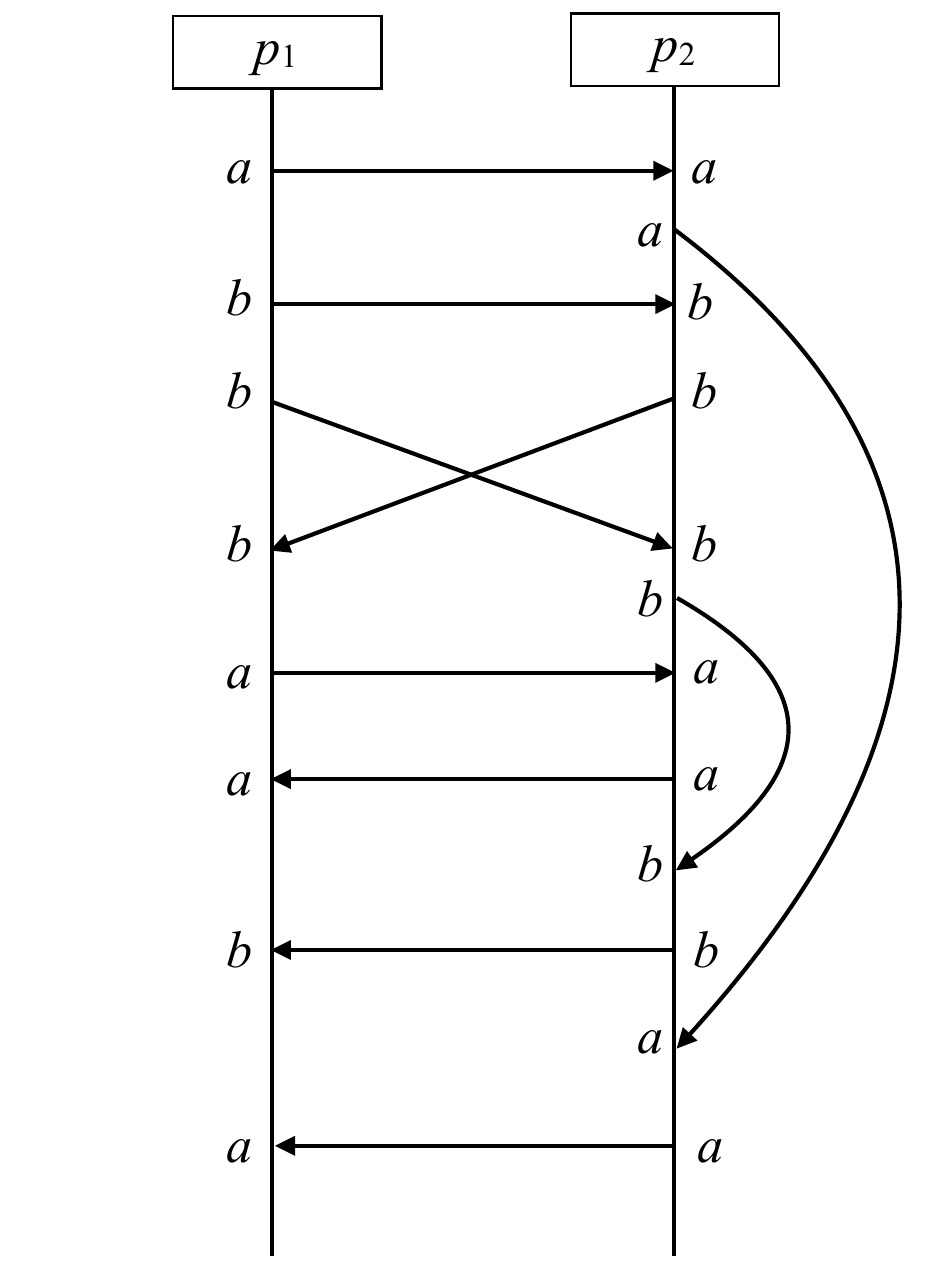}
}
\end{center}
\end{example}


\begin{definition}\label{def:cbm}
  A \emph{concurrent behavior with matching} (\txtCBM) over $\Arch$ and $\Act$ is a tuple
  \[\mscn=((w_p)_{p\in\Procs},(\matchrel^d)_{d\in\DS})\]
\begin{itemize}
\item $w_p\in\Sigma^*$ \quad sequence of actions on process $p$\\[1ex]
\fbox{
\parbox{12cm}{
\underline{Notation:}\\[0.5ex]
$\Events_p=\{(p,i)\mid 1\leq i\leq|w_p|\}$ \quad events of process $p$\\[0.5ex]
$\Events=\bigcup_{p\in\Procs}\Events_p$\\[0.5ex]
$(p,i)\procrel(p,i+1)$ \quad if $1\leq i<|w_p|$\\[0.5ex]
for $e=(p,i)\in\Events_p$, let $\pid(e)=p$ and $\ActLabel(e) \in \Sigma$ be the $i$-th letter of $w_p$
}}
\item ${\matchrel^d}\subseteq\Events_{\writer(d)}\times\Events_{\reader(d)}$ such that:
\begin{itemize}[nosep]\itemsep=.5ex
\item if $e_1\matchrel^d e_2$
  and $e_3\matchrel^{d'} e_4$ are different edges ($d\neq d'$ or
  $(e_1,e_2)\neq(e_3,e_4)$), then they are disjoint
  ($|\{e_{1},e_{2},e_{3},e_{4}\}| = 4$)
  
\item ${<}=({\procrel}\cup{\matchrel})^+ \subseteq \Events \times \Events$ is a strict partial order\footnote{For a binary relation $R$, we let $R^\ast = \bigcup_{n \ge 0} R^n$ and $R^+ = \bigcup_{n \ge 1} R^n$.}\\[0.5ex]
where ${\matchrel}=\bigcup_{d\in\DS}{\matchrel^d}$

    \item $\forall d \in \Stacks$ (LIFO):\\[0.5ex]
    $e_{1} \matchrel^d f_{1} \textup{~and~} e_{2} \matchrel^d f_{2} \textup{~and~} e_1<e_2<f_1 ~~\Longrightarrow ~~f_2<f_1$

    \item $\forall d \in \Queues$ (FIFO):\\[0.5ex]
    $e_{1} \matchrel^d f_{1} \textup{~and~} e_{2} \matchrel^d f_{2} \textup{~and~} e_1<e_2 ~~\Longrightarrow ~~f_1<f_2$
    
\end{itemize}
\end{itemize}
We also write $\mscn=(\Events, \procrel, (\matchrel^{d})_{d\in\DS}, \pid, \ActLabel)$.

We let $\CBM(\Arch,\Act)$ be the set of \CBMs over \Arch and \Act.
\end{definition}

\underline{{\bf Run:}}

Consider a mapping $\rho:\Events\to\mathsf{Trans}=\bigcup_{p\in\Procs}\Delta_p$.

Now, $\rho$ is a \emph{run} of \Sys on $\mscn$ if the following hold:
\begin{itemize}[nosep]
  \item for all $e\in\Events$: \quad $\rho(e)\in\Delta_{\pid(e)}$ and 
  $\ActLabel(e)=\mathsf{lbl}(\rho(e))$

  \item for all $e\procrel f$: \quad $\mathsf{tgt}(\rho(e))=\mathsf{src}(\rho(f))$

  \item for all $e\matchrel^d f$: \quad $\rho(e)\in\Delta^{!}$ is a \emph{write} 
  transition, $\rho(f)\in\Delta^{?}$ is a \emph{read} transition, 
  $\mathsf{ds}(\rho(e))=d=\mathsf{ds}(\rho(f))$ and 
  $\mathsf{val}(\rho(e))=\mathsf{val}(\rho(f))$.
  
  \item initial: \quad for all $p\in\Procs$, either $\Events_p=\emptyset$ or 
  $\mathsf{src}(\rho(\min\Events_p))=\iota_p$
\end{itemize}


\underline{{\bf Accepting:}}

A run $\rho$ is accepting if $(\ell_p)_{p \in \Procs} \in\FinLocs$ where
\[\ell_p =
\begin{cases}
\iota_p & \textup{if~} \Events_p=\emptyset\\
\mathsf{tgt}(\rho(\max\Events_p)) & \textup{otherwise}
\end{cases}\]
We let $\Lang(\Sys)$ denote the set of \CBMs accepted by \Sys.


\begin{example}
The following figure depicts a run of $\csSys$:
\begin{center}
\fbox{
\includegraphics[scale=0.33]{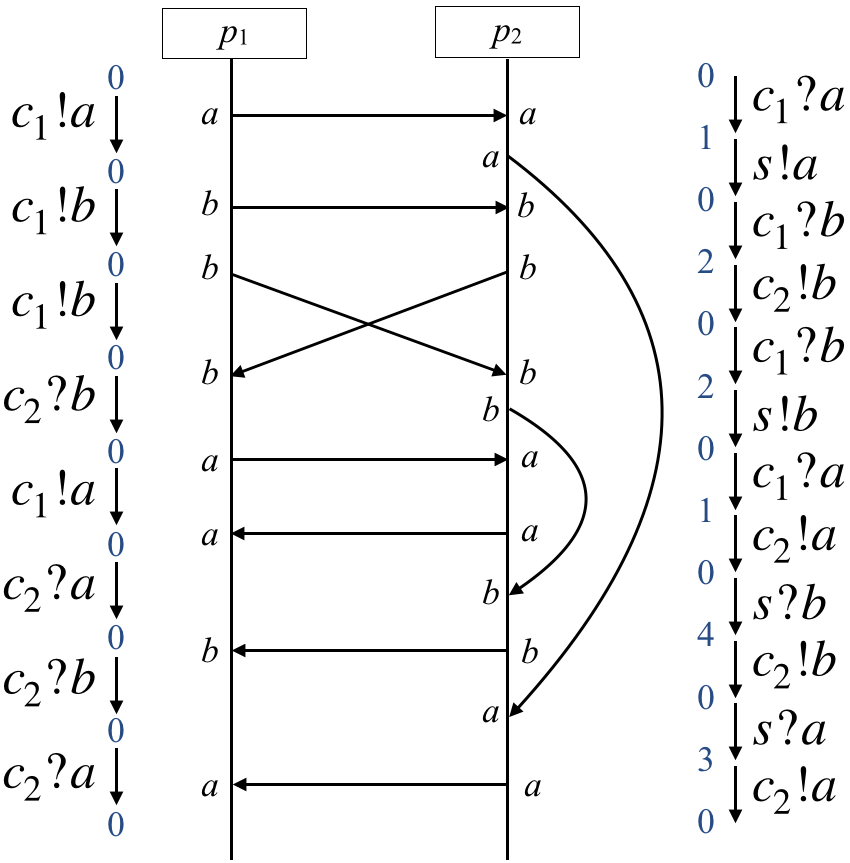}
}
\end{center}
\end{example}

\begin{exercise}
  Show that the class of \CBM languages accepted by \txtCPDS is closed under 
  union and intersection.
\end{exercise}

We will see that \txtCPDS are not closed under complement: 
Exercise~\ref{exo:cpds-complement}.

\underline{{\bf Relation between operational and graph semantics:}}

Every \txtCBM $\mscn=((w_p)_{p\in\Procs},(\matchrel^d)_{d\in\DS})$ defines a set of words over $\Gamma$.

Let $\gamma_\mscn: \Events \to \Gamma~~~~\bigl(=(\Procs\times\Act)\cup(\Procs\times\Act\times\DS\times\{!,?\})\bigr)$ ~~~~be defined by
\[\gamma_\mscn(e) =
\begin{cases}
(\pid(e),\ActLabel(e))  & \textup{if~} e \textup{~is~internal}\\
(\pid(e),\ActLabel(e),d!)  & \textup{if~} e \matchrel^d f\\
(\pid(e),\ActLabel(e),d?)  & \textup{if~} f \matchrel^d e
\end{cases}\]

A \emph{linearization} of $\mscn$ is any (strict) total order ${\sqsubset} \subseteq \Events \times \Events$ such that ${<} \subseteq {\sqsubset}$.\\
(recall that ${<}=({\procrel}\cup{\matchrel})^+$).

Suppose $\Events=\{e_1,\ldots,e_n\}$ and $e_1 \sqsubset \ldots \sqsubset e_n$.

Then, $\sqsubset$ induces the word $\gamma_\mscn(e_1) \ldots \gamma_\mscn(e_n) \in \Gamma^\ast$.

Let $\Lin(\mscn)\subseteq\Gamma^\ast$ be the set of words that are induced by the linearisations of \mscn.

\begin{remark}
\begin{itemize}
\item If $\Bags = \emptyset$, then for every $w \in \Gamma^\ast$, there is at most one $\mscn \in \CBM(\Arch,\Sigma)$ such that $w \in \Lin(\mscn)$. 

\item If $\Bags = \{d\}$, this is not the case: $(p,a,d!)(p,a,d!)(p,a,d?)(p,a,d?)$ is a linearization of two different \CBMs.
\end{itemize}
\end{remark}

\begin{example}
Let $\mscn$ be the following \txtCBM.
\begin{center}
\fbox{
\includegraphics[scale=0.4]{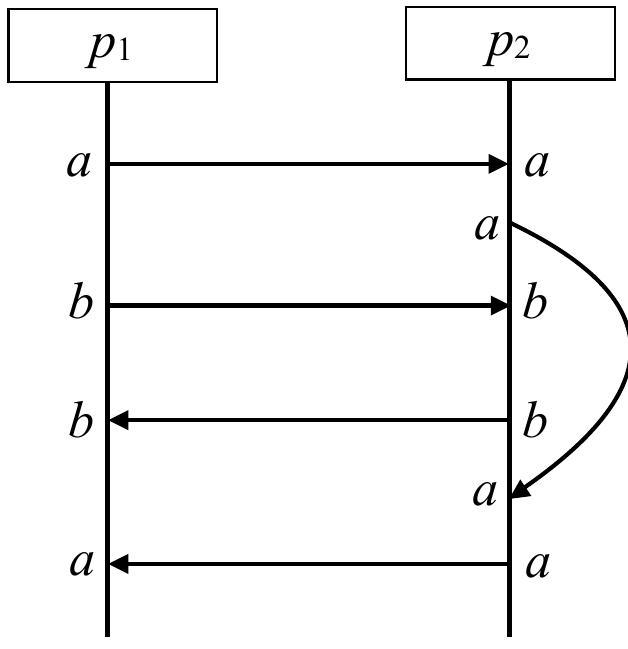}
}
\end{center}
Then, $\Lin(\mscn)$ contains:
\begin{itemize}[nosep]
\item $(p_1,a,c_1!)(p_1,b,c_1!)\textcolor{blue}{(p_2,a,c_1?)(p_2,a,s!)}$\\
$\textcolor{blue}{(p_2,b,c_1?)(p_2,b,c_2!)(p_2,a,s?)(p_2,a,c_2!)}(p_1,b,c_2?)(p_1,a,c_2?)$
\item $(p_1,a,c_1!)\textcolor{blue}{(p_2,a,c_1?)(p_2,a,s!)}(p_1,b,c_1!)$\\
$\textcolor{blue}{(p_2,b,c_1?)(p_2,b,c_2!)(p_2,a,s?)(p_2,a,c_2!)}(p_1,b,c_2?)(p_1,a,c_2?)$
\end{itemize}
Actually, $\mscn$ has $9$ linearizations.
\end{example}

\begin{mytheorem}\label{thm:semantics}
For all $\Sys \in \CPDS(\Arch,\Sigma)$, we have $\Lin(\Lang(\Sys))=\opLang(\Sys)$.
\end{mytheorem}

Without proof.

\fi 

\renewcommand{\phi}{\varphi}
\renewcommand{\epsilon}{\varepsilon}
\newcommand{\Var}{\textup{Var}}
\newcommand{\VAR}{\textup{VAR}}
\newcommand{\Int}{\mathcal{I}}
\newcommand{\writeev}{\mathit{write}}
\newcommand{\readev}{\mathit{read}}
\newcommand{\locev}{\mathit{local}}

\newcommand{\goright}{\mathsf{go\textup{-}right}}
\newcommand{\godown}{\mathsf{go\textup{-}down}}

\newcommand{\sem}[1]{\ensuremath{[\![#1]\!]}}
\newcommand{\ftrans}[1]{\ensuremath{\widetilde{#1}}}
\newcommand{\Sem}[2]{\sem{#1}_#2}
\newcommand{\Aut}{\mathcal{A}}
\newcommand{\CL}{\texttt{CL}}

\newcommand{\transf}{\mathit{trans}}
\newcommand{\dsf}{\mathit{data}}

\newcommand{\Conc}{\cdot}

\ifMSO

\chapter{Monadic Second-Order Logic}

\section{Monadic Second-Order Logic}

Example: $\forall x (a(x) \Rightarrow \exists y (x \matchrel y \wedge b(y)))$

\smallskip

\underline{{\bf Syntax:}}

Let $\Var=\{x,y,\ldots\}$ be an infinite set of first-order variables.

Let $\VAR=\{X,Y,\ldots\}$ be an infinite set of second-order variables.

The set $\MSO(\Arch,\Act)=\MSO(\Sigma,\Procs,\procrel,(\matchrel^{d})_{d\in\DS})$
of formulas from \emph{monadic second-order logic} is given by the grammar:
$$
  \varphi ::= a(x) \mid p(x) \mid x=y \mid x \matchrel^d y \mid x \procrel y
  \mid x \in X \mid \varphi \lor \varphi \mid \neg \varphi 
  \mid \exists x\, \varphi\mid \exists X\, \varphi
$$
where $x,y \in \Var$, $X \in \VAR$, $a \in \Act$, $p \in \Procs$, $d \in \DS$.

The fragment
$\EMSO(\Arch,\Sigma)=\EMSO(\Sigma,\Procs,\procrel,(\matchrel^{d})_{d\in\DS})$
consists of the formulas of the form $\exists X_1 \ldots \exists X_n \phi$ where
$\phi$ is a first-order formula, i.e., it does not contain any second-order
quantification.


\underline{{\bf Semantics:}}

Let $\mscn=((w_p)_{p\in\Procs},(\matchrel^d)_{d\in\DS})=(\Events, \procrel,
(\matchrel^{d})_{d\in\DS}, \pid, \ActLabel)$ be a \txtCBM. An
$\mscn$-interpretation is a function $\Int$ that maps every
\begin{itemize}[nosep]\itemsep=1ex
  \item $x \in \Var$ to some element of $\Events$
  \item $X \in \VAR$ to some subset of $\Events$
\end{itemize}
Satisfaction $\mscn \models_\Int \phi$ is defined inductively as follows:
\begin{itemize}[nosep]\itemsep=1ex
\item \parbox{3cm}{$\mscn \models_\Int a(x)$} if $\lambda(\Int(x)) = a$
\item \parbox{3cm}{$\mscn \models_\Int p(x)$} if $\pid(\Int(x)) = p$
\item \parbox{3cm}{$\mscn \models_\Int x = y$} if $\Int(x) = \Int(y)$
\item \parbox{3cm}{$\mscn \models_\Int x \matchrel^d y$} if $\Int(x) \matchrel^d \Int(y)$
\item \parbox{3cm}{$\mscn \models_\Int x \procrel y$} if $\Int(x) \procrel \Int(y)$
\item \parbox{3cm}{$\mscn \models_\Int x \in X$} if $\Int(x) \in \Int(X)$
\item \parbox{3cm}{$\mscn \models_\Int \phi \vee \psi$} if $\mscn \models_\Int \phi$ or $\mscn \models_\Int \psi$
\item \parbox{3cm}{$\mscn \models_\Int \neg\phi$} if $\mscn \not\models_\Int \phi$
\item \parbox{3cm}{$\mscn \models_\Int \exists x \phi$} if there is $e \in \Events$ such that $\mscn \models_{\Int[x \mapsto e]} \phi$
\item \parbox{3cm}{$\mscn \models_\Int \exists X \phi$} if there is $E \subseteq \Events$ such that $\mscn \models_{\Int[X \mapsto E]} \phi$
\end{itemize}
Here, $\Int[x \mapsto e]$ maps $x$ to $e$ and coincides with $\Int$ on $(\Var \setminus \{x\}) \cup \VAR$.

When $\phi$ is a sentence, then $\Int$ is irrelevant, and we write $\mscn \models \phi$ instead of $\mscn \models_\Int \phi$.

We let $\Lang(\phi) := \{\mscn \in \CBM(\Arch,\Sigma) \mid \mscn \models \phi\}$.


\begin{example}
We use the following abbreviations:
\begin{itemize}[nosep]\itemsep=1ex
\item $\phi \wedge \psi ~\fequiv~ \neg(\neg\phi \vee \neg\psi)$
\qquad $\forall x \phi ~\fequiv~ \neg \exists x \neg\phi$
\qquad $\phi \Rightarrow \psi ~\fequiv~ \neg \phi \vee \psi$
\item $x \matchrel y \fequiv \bigvee_{d \in \DS} (x \matchrel^d y)$
\item $\writeev(x) \fequiv \exists y (x \matchrel y)$
\qquad $\readev(x) \fequiv \exists y (y \matchrel x)$
\item $\locev(x) \fequiv \neg\writeev(x) \wedge \neg\readev(x)$
\item $\minev(x) \fequiv \neg\exists y(y \procrel x)$ \qquad $\maxev(x) \fequiv \neg\exists y(x \procrel y)$
\item $\textup{``}\Events_p = \emptyset\,\textup{''} \,\fequiv \neg\exists x\, p(x)$
\item $x \le y \fequiv \forall X (x \in X \wedge \forall z \forall z'((z \in X \wedge (z \procrel z' \vee z \matchrel z')) \Rightarrow z' \in X) \Rightarrow y \in X)$
\item On \CBMs, the latter formula is equivalent to
$$
\exists X~[~ y\in X \wedge \forall z \in X~
(~z=x \vee \exists z' \in X~ (z' \procrel z \vee z' \matchrel z)~)~]
$$
\end{itemize}
\end{example}


\begin{example} We consider some formulas for $\csSys$:
\begin{itemize}[nosep]\itemsep=1ex
\item $\phi_1 \fequiv \forall x (a(x) \Rightarrow \exists y (x \le y \wedge b(y)))$
\item $\mathit{req\textup{-}ack}(x,y) \fequiv
\left(\begin{array}{rl}
& \exists x_1,x_2 (x \matchrel^{c_1} x_1 \procrel x_2 \matchrel^{c_2} y)\\[1ex]
\vee & \exists x_1,\ldots,x_4 (x \matchrel^{c_1} x_1 \procrel x_2 \matchrel^{s} x_3 \procrel x_4 \matchrel^{c_2} y)
\end{array}\right)$
\item $\phi_2 \fequiv \forall x,y \Bigl(\mathit{req\textup{-}ack}(x,y) \Rightarrow \bigl((a(x) \wedge a(y)) \vee (b(x) \wedge b(y))\bigr)\Bigr)$
\end{itemize}
For the client-server system $\csSys$ from Example~\ref{ex:client-server}, we have $\Lang(\csSys) \not\subseteq \Lang(\phi_1)$ and $\Lang(\csSys) \subseteq \Lang(\phi_2)$.
\end{example}

\clearpage
\newcommand{\phisys}{\ensuremath{\Phi_\Sys}\xspace}

\section{Expressive Power of MSO Logic}\label{sec:exprMSO}

Recall a theorem from the sequential case:

\begin{mytheorem}[B{\"u}chi-Elgot-Trakhtenbrot \cite{Buchi60,Elgot61,Trakhtenbrot62}]~\\
Suppose $|\Procs|=1$ and $\DS=\emptyset$. Let $\Lang \subseteq \CBM(\Arch,\Sigma)$, which can be seen as a word language $L \subseteq \Sigma^\ast$.
Then, the following are equivalent:
\begin{itemize}
\item There is $\Sys \in \CPDS(\Arch,\Sigma)$ such that $\Lang(\Sys) = \Lang$.
\item There is a sentence $\phi \in \MSO(\Arch,\Sigma)$ such that $\Lang(\phi) = \Lang$.
\end{itemize}
\end{mytheorem}

The theorem also holds for $|\Procs|=1$, $|\DS| = 1$, and $\DS=\Stacks$ \cite{Alur2009}.

\bigskip

One direction is actually independent of achitecture:

\begin{mytheorem}\label{thm:CPDStoMSO}
  For every $\Sys \in \CPDS(\Arch,\Sigma)$, there is a sentence
  $\phisys\in\EMSO(\Arch,\Sigma)$ of size $\mathcal{O}(|\Sys|^{2})$ such that
  $\Lang(\phisys) = \Lang(\Sys)$.
\end{mytheorem}

\begin{proof}
  Fix $\Sys=((\Sys_p)_{p\in\Procs},\Val,\FinLocs)\in\CPDS(\Arch,\Sigma)$.  
  Recall the notations of Definition~\ref{def:cpds}.
  We define
  \[
  \phisys =
  \begin{array}[t]{rl}
    \multicolumn{2}{l}{\exists (X_t)_{t\in\mathsf{Trans}} \,\Bigl[ 
    \forall x \, \bigvee_{t\in\mathsf{Trans}} \Big( X_t(x) ~\wedge~ 
    \bigwedge_{t'\neq t} \neg X_{t'}(x) \Big)
    }
    \\[2ex]
    \wedge & \forall x \, \bigwedge_{p\in\Procs, a\in\Sigma} \bigl(p(x) \wedge a(x) \Rightarrow 
    \bigvee_{t\in\Delta_p\mid\mathsf{lbl}(t)=a} X_t(x) \bigr)
    \\[2ex]
    \wedge & \forall x,x' \,\bigl(x \procrel x' \Rightarrow 
    \bigvee_{\substack{t,t'\in\mathsf{Trans} \mid \mathsf{tgt}(t)=\mathsf{src}(t')}} X_t(x) \wedge X_{t'}(x') \bigr)
    \\[2ex]
    \wedge & \forall x,x' \, \Big( x \matchrel x' \Rightarrow 
    \bigvee_{\substack{t\in\Delta^{!},t'\in\Delta^{?} \mid \\ 
    \mathsf{val}(t)=\mathsf{val}(t') \wedge \mathsf{ds}(t)=\mathsf{ds}(t')}} 
    X_t(x) \wedge X_{t'}(x') \wedge x \matchrel^{\mathsf{ds}(t)} x' \Big)
    \\[2ex]
    \wedge & \forall x \, \Big( (\neg\exists y~y\procrel x) \Rightarrow 
    \bigvee_{\substack{p\in\Procs,t\in\Delta_p \mid\mathsf{src}(t)=\iota_p}} 
    p(x) \wedge X_t(x) \Big)
    \\[2ex]
    \wedge & 
    \bigvee_{(\ell_p)_{p\in\Procs}\in\FinLocs}
    \Bigg(\!\!
    \begin{array}{l}
      \bigwedge_{p\,\in\,\Procs \mid \ell_p\neq\iota_p}~ \exists x\, p(x)
      \\[1ex]
      {}\wedge
      \forall x\, \Big( ( \neg\exists y~x\procrel y ) \implies
      \bigvee_{\substack{p\in\Procs,t\in\Delta_p \\ \mid\mathsf{tgt}(t)=\ell_p}} 
      ~p(x)\wedge X_{t}(x) \Big)
    \end{array}\!\!\!\Bigg)\Bigr]
  \end{array}
  \]
  This completes the construction of the formula $\phisys$. 
  We have $\Lang(\phisys) = \Lang(\Sys)$.
  Note that $\phisys$ does not use $\leq$.
\end{proof}

Unfortunately, the other direction does not hold in general:

\begin{mytheorem}\label{thm:grids}
Suppose $\Sigma=\{a,b,c\}$.
Suppose that $\Arch$ is given by $\Procs=\{p_1,p_2\}$ and $\DS=\Queues=\{c_1,c_2\}$ with $\writer(c_1)=\reader(c_2)=p_1$ and $\writer(c_2)=\reader(c_1)=p_2$:
\begin{center}
{
\includegraphics[scale=0.5]{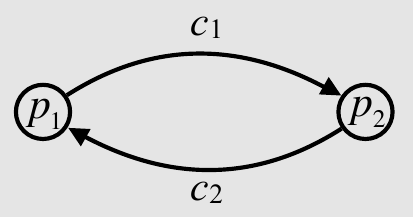}
}
\end{center}
There is a sentence $\phi \in \MSO(\Arch,\Sigma)$ such that, for all $\Sys \in \CPDS(\Arch,\Sigma)$, we have $\Lang(\Sys) \neq \Lang(\phi)$.
\end{mytheorem}

\begin{proof}
  To illustrate the proof idea, which goes back to \cite{ThoPOMIV96} we consider
  $(n\times m)$-pictures over alphabet $\Sigma=\{a,b,c\}$, i.e., maps
  $\mathsf{pict}\colon[n]\times[m]\to\Sigma$.  Here is an example of a $(3\times
  7)$-picture:
  \begin{center}
\fbox{
\scalebox{1}{
\unitlength=1.2mm
\gasset{frame=false,AHangle=30,AHlength=1.4,AHLength=1.6}
\begin{gpicture}%
  \gasset{Nw=1.5,Nh=1.5,Nframe=n,Nfill=y}%

\put(0,0){ 
    \drawccurve[fillgray=0.9](-13,-3)(-16,10)(-13,23)(0,26)(13,23)(16,10)(13,-3)(0,-6)
   }

\put(40,0){ 
    \drawccurve[fillgray=0.9](-13,-3)(-16,10)(-13,23)(0,26)(13,23)(16,10)(13,-3)(0,-6)
   }

\node(V1)(-10,20){}%
\node(V2)(-10,10){}%
\node(V3)(-10,0){}%
\node(A1)(0,20){}%
\node(A2)(0,10){}%
\node(A3)(0,0){}%
\node(B1)(10,20){}%
\node(B2)(10,10){}%
\node(B3)(10,0){}%
\node(C1)(20,20){}%
\node(C2)(20,10){}%
\node(C3)(20,0){}%
\node(D1)(30,20){}%
\node(D2)(30,10){}%
\node(D3)(30,0){}%
\node(E1)(40,20){}%
\node(E2)(40,10){}%
\node(E3)(40,0){}%
\node(N1)(50,20){}%
\node(N2)(50,10){}%
\node(N3)(50,0){}%
\drawedge(A1,A2){}%
\drawedge(A2,A3){}%
\drawedge(B1,B2){}%
\drawedge(B2,B3){}%
\drawedge(C1,C2){}%
\drawedge(C2,C3){}%
\drawedge(D1,D2){}%
\drawedge(D2,D3){}%
\drawedge(E1,E2){}%
\drawedge(E2,E3){}%
\drawedge(A1,B1){}%
\drawedge(B1,C1){}%
\drawedge(C1,D1){}%
\drawedge(D1,E1){}%
\drawedge(A2,B2){}%
\drawedge(B2,C2){}%
\drawedge(C2,D2){}%
\drawedge(D2,E2){}%
\drawedge(A3,B3){}%
\drawedge(B3,C3){}%
\drawedge(C3,D3){}%
\drawedge(D3,E3){}%

\drawedge(V1,V2){}%
\drawedge(V2,V3){}%
\drawedge(N1,N2){}%
\drawedge(N2,N3){}%

\drawedge(V1,A1){}%
\drawedge(V2,A2){}%
\drawedge(V3,A3){}%

\drawedge(E1,N1){}%
\drawedge(E2,N2){}%
\drawedge(E3,N3){}%

\gasset{Nframe=n,Nadjust=w,Nh=6,Nmr=0,Nfill=n}

\node(V1)(-12,22){$a$}%
\node(V2)(-12,12){$b$}%
\node(V3)(-12,2){$b$}%
\node(A1)(-2,22){$b$}%
\node(A2)(-2,12){$a$}%
\node(A3)(-2,2){$b$}%
\node(B1)(8,22){$b$}%
\node(B2)(8,12){$a$}%
\node(B3)(8,2){$a$}%
\node(C1)(18,22){$c$}%
\node(C2)(18,12){$c$}%
\node(C3)(18,2){$c$}%
\node(D1)(28,22){$a$}%
\node(D2)(28,12){$b$}%
\node(D3)(28,2){$b$}%
\node(E1)(38,22){$b$}%
\node(E2)(38,12){$a$}%
\node(E3)(38,2){$b$}%
\node(N1)(48,22){$b$}%
\node(N2)(48,12){$a$}%
\node(N3)(48,2){$a$}%
\end{gpicture}
}
}
\end{center}

Consider the set $P_=$ of pictures that are of the form $ACA$ where
\begin{itemize}[nosep]\itemsep=1ex
\item $A$ is a nonempty square picture with labels in $\{a,b\}$, and
\item $C$ is a $c$-labeled column.
\end{itemize}
The above picture is a member of $P_=$.

\myline

The language $P_=$ is definable by an MSO formula $\Phi_=$ over pictures using predicates:
\[\goright(x,y) \qquad\qquad \godown(x,y)\]
Here the interpretation of first-order variables $x,y$ are pixels in 
$[n]\times[m]$.

The formula $\Phi_=$ is easy to obtain once we have a ``matching'' predicate
$\mu(x,y)$ that relates two coordinates $x$ and $y$ iff they refer to identical
positions in the two different square grids: $\mathcal{I}(x)=(i,j)$ with 
$j\in[n]$ and 
$\mathcal{I}(y)=(i,j+n+1)$.

Essentially, $\mu(x,y)$ says that $x$ and $y$ are on the same line using the 
transitive closure $\godown^{*}(x,y)$ (recall that transitive closure can be 
expressed in \MSO). It is a bit more difficult to state that there are exactly 
$n$ columns between the columns of $x$ and $y$. For this we use further 
transitive closure, and in particular for the binary relation defined by
$$
\textsf{go-diag}(x,y) = \exists z~(\godown(x,z) \wedge \goright(z,y)
$$
So we define
\begin{align*}
  \mu(x,y)=\exists x_1,x_2,x_3,z ~
  & \mathsf{firstline}(x_1) \wedge \godown^{*}(x_1,x) \\
  {}\wedge{} & \mathsf{lastline}(x_2) \wedge \textsf{go-diag}^{*}(x_1,x_2) \\
  {}\wedge{} & \goright(x_2,z) \wedge \goright(z,x_3) \wedge \godown^{*}(y,x_3)
\end{align*}
where
\begin{align*}
  \mathsf{firstline}(z) & = \neg\exists z' ~\godown(z',z) \\
  \mathsf{lastline}(z) & = \neg\exists z' ~\godown(z,z')
\end{align*}
The idea for $\mu(x,y)$ is illustrated below:
\begin{center}
\fbox{
\includegraphics[scale=0.45]{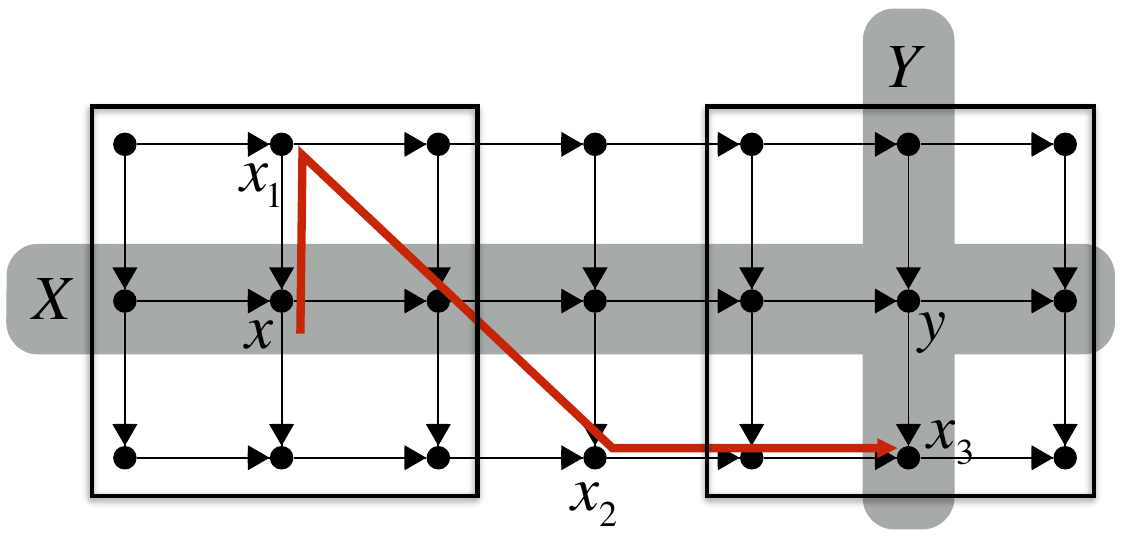}
}
\end{center}
Note that $\mu(x,y)$ can be written in \EMSO over pictures.

The formula $\Phi_=$ says that the picture is of size $n\times(2n+1)$ for some 
$n\geq1$ (again using transitive closures of $\godown$ and $\textsf{go-diag}$), 
that the middle column is labeled $c$, and
$$
\forall x,y~(\mu(x,y) \implies ((a(x)\wedge a(y))\vee(b(x)\wedge b(y)))
$$
Note that $\Phi_=$ is not in \EMSO.

\myline
\clearpage
Next, we encode pictures into \CBMs over $\Arch$ and $\Sigma$. The above picture is encoded as follows:

\begin{center}
\fbox{
  \scalebox{0.8}{
  \unitlength=1.2mm
  \gasset{frame=false,AHangle=30,AHlength=1.4,AHLength=1.6}
  \begin{gpicture}
    \gasset{AHLength=2,AHlength=1.8,AHangle=30,Nw=0,Nh=0,Nframe=n,Nfill=y}
    \unitlength=0.5mm

\put(0,0){ 
    \drawccurve[fillgray=0.9](30,77)(0,55)(-15,50)(-30,110)(-18,170)(30,180)(80,170)(87,110)(70,99)
   }

\put(0,-110){ 
    \drawccurve[fillgray=0.9](-5,25)(-30,90)(-20,155)(0,158)(30,158)(75,149)(90,90)(85,40)
   }


\put(0,120){ 
    \drawccurve[fillgray=0.7](-15,-5)(-20,20)(-15,45)(-9,45)(-4,20)(-9,-5)
   }

\put(0,60){ 
    \drawccurve[fillgray=0.7](-15,-5)(-20,20)(-15,45)(-9,45)(-4,20)(-9,-5)
   }

\put(0,0){ 
    \drawccurve[fillgray=0.7](-15,-5)(-20,20)(-15,45)(-9,45)(-4,20)(-9,-5)
   }

\put(0,-60){ 
    \drawccurve[fillgray=0.7](-15,-5)(-20,20)(-15,45)(-9,45)(-4,20)(-9,-5)
   }

\put(84,110){ 
    \drawccurve[fillgray=0.7](-15,-5)(-20,20)(-15,45)(-9,45)(-4,20)(-9,-5)
   }

\put(84,50){ 
    \drawccurve[fillgray=0.7](-15,-5)(-20,20)(-15,45)(-9,45)(-4,20)(-9,-5)
   }

\put(84,-10){ 
    \drawccurve[fillgray=0.7](-15,-5)(-20,20)(-15,45)(-9,45)(-4,20)(-9,-5)
   }

\node(A1)(0,160){}
\node(A2)(0,140){}
\node(A3)(0,120){}
\node(A4)(0,110){}
\node(A5)(0,100){}
\node(A6)(0,90){}
\node(A7)(0,80){}
\node(A8)(0,70){}
\node(A9)(0,60){}
\node(A10)(0,50){}
\node(A11)(0,40){}
\node(A12)(0,30){}
\node(A13)(0,20){}
\node(A14)(0,10){}
\node(A15)(0,0){}
\node(A16)(0,-10){}
\node(A17)(0,-20){}
\node(A18)(0,-30){}
\node(A19)(0,-40){}
\node(A20)(0,-50){}
\node(A21)(0,-60){}

\node(B1)(60,160){}
\node(B2)(60,150){}
\node(B3)(60,140){}
\node(B4)(60,130){}
\node(B5)(60,120){}
\node(B6)(60,110){}
\node(B7)(60,100){}
\node(B8)(60,90){}
\node(B9)(60,80){}
\node(B10)(60,70){}
\node(B11)(60,60){}
\node(B12)(60,50){}
\node(B13)(60,40){}
\node(B14)(60,20){}
\node(B15)(60,0){}
\node(B16)(60,30){}
\node(B17)(60,10){}
\node(B18)(60,-10){}
\node(B19)(60,-20){}
\node(B20)(60,-30){}
\node(B21)(60,-40){}
\node(B22)(60,-50){}
\node(B23)(60,-60){}
\node(B24)(60,-70){}

\drawedge[AHnb=0](A1,A2){}
\drawedge[AHnb=0](A2,A3){}
\drawedge[AHnb=0](A3,A4){}
\drawedge[AHnb=0](A4,A5){}
\drawedge[AHnb=0](A5,A6){}
\drawedge[AHnb=0](A6,A7){}
\drawedge[AHnb=0](A7,A8){}
\drawedge[AHnb=0](A8,A9){}
\drawedge[AHnb=0](A9,A10){}
\drawedge[AHnb=0](A10,A11){}
\drawedge[AHnb=0](A11,A12){}
\drawedge[AHnb=0](A12,A13){}
\drawedge[AHnb=0](A13,A14){}
\drawedge[AHnb=0](A14,A15){}
\drawedge[AHnb=0](A15,A16){}
\drawedge[AHnb=0](A16,A17){}
\drawedge[AHnb=0](A17,A18){}
\drawedge[AHnb=0](A18,A19){}
\drawedge[AHnb=0](A19,A20){}
\drawedge[AHnb=0](A20,A21){}

\drawedge[AHnb=0](B1,B2){}
\drawedge[AHnb=0](B2,B3){}
\drawedge[AHnb=0](B3,B4){}
\drawedge[AHnb=0](B4,B5){}
\drawedge[AHnb=0](B5,B6){}
\drawedge[AHnb=0](B6,B7){}
\drawedge[AHnb=0](B7,B8){}
\drawedge[AHnb=0](B8,B9){}
\drawedge[AHnb=0](B9,B10){}
\drawedge[AHnb=0](B10,B11){}
\drawedge[AHnb=0](B11,B12){}
\drawedge[AHnb=0](B12,B13){}
\drawedge[AHnb=0](B13,B16){}
\drawedge[AHnb=0](B16,B14){}
\drawedge[AHnb=0](B14,B17){}
\drawedge[AHnb=0](B17,B15){}
\drawedge[AHnb=0](B15,B18){}
\drawedge[AHnb=0](B18,B19){}
\drawedge[AHnb=0](B19,B20){}
\drawedge[AHnb=0](B20,B21){}
\drawedge[AHnb=0](B21,B22){}
\drawedge[AHnb=0](B22,B23){}

\drawedge(A1,B1){}
\drawedge(A2,B3){}
\drawedge(A3,B5){}
\drawedge(A5,B7){}
\drawedge(A7,B9){}
\drawedge(A9,B11){}
\drawedge(A11,B13){}
\drawedge(A13,B14){}
\drawedge(A15,B15){}

\drawedge(A17,B19){}
\drawedge(A19,B21){}
\drawedge(A21,B23){}

\drawedge(B16,A16){}
\drawedge(B17,A18){}
\drawedge(B18,A20){}

\drawedge(B2,A4){}
\drawedge(B4,A6){}

\drawedge(B6,A8){}

\drawedge(B8,A10){}
\drawedge(B10,A12){}

\drawedge(B12,A14){}

\gasset{Nframe=n,Nadjust=w,Nw=0,Nh=0,Nmr=0,Nfill=n}

\node(A)(-12,160){{\large $a$}}
\node(A)(-12,140){{\large $b$}}
\node(A)(-12,120){{\large $b$}}

\node(A)(-12,100){{\large $b$}}
\node(A)(-12,80){{\large $a$}}
\node(A)(-12,60){{\large $a$}}

\node(A)(-12,40){{\large $a$}}
\node(A)(-12,20){{\large $b$}}
\node(A)(-12,0){{\large $b$}}

\node(A)(72,150){{\large $b$}}
\node(A)(72,130){{\large $a$}}
\node(A)(72,110){{\large $b$}}

\node(A)(72,90){{\large $c$}}
\node(A)(72,70){{\large $c$}}
\node(A)(72,50){{\large $c$}}

\node(A)(72,30){{\large $b$}}
\node(A)(72,10){{\large $a$}}
\node(A)(72,-10){{\large $b$}}

\node(A)(-12,-20){{\large $b$}}
\node(A)(-12,-40){{\large $a$}}
\node(A)(-12,-60){{\large $a$}}

\end{gpicture}
}
}
\end{center}
We obtain a formula $\ftrans{\Phi_=} \in \MSO(\Arch,\Sigma)$ for the encodings of the above picture language $P_=$ inductively:
\begin{itemize}[nosep]\itemsep=1ex
\item $\ftrans{\exists x \phi} ~=~ \exists x (\writeev(x) \wedge \ftrans{\phi})$
\item $\ftrans{\goright}(x,y) ~=~ \exists z (x \matchrel z \procrel y)$
\item $\ftrans{\godown}(x,y) ~=~ \neg \mathit{bottom}(x) \wedge (x \procrel y \vee \exists z (x \procrel z \procrel y \wedge \neg \writeev(z)))$
\end{itemize}
Here, $\mathit{bottom}(x)$ says that $x$ is an element that is located on the 
last row:
\begin{align*}
  \mathsf{last}_p(x) & = p(x) \wedge \mathsf{write}(x) \wedge
  \neg\exists y~(x\procrel y) \\
  \mathsf{bottom}(x) & = \exists y~(\mathsf{last}_p(y) \wedge 
  \ftrans{\goright}^{*}(x,y)
\end{align*}
Other formulas remain unchanged.

Using Theorem~\ref{thm:CPDStoMSO}, we can moreover determine a formula $\psi_{\textup{pict}}$ that describes the encodings of (arbitrary) pictures.

Let $\phi = \psi_{\textup{pict}} \wedge \ftrans{\Phi_=}$.

\myline

Towards a contradiction, suppose that there is
$\Sys=((\Sys_p)_{p\in\Procs},\Val,\FinLocs)\in\CPDS(\Arch,\Sigma)$ such that
$\Lang(\Sys) = \Lang(\phi)$.

An accepting run of $\Sys$ has to transfer all the information it has about the upper part of the \txtCBM along the middle part of size $2n$ (where $n$ is the length of a column), to the lower part.

However, there are
\begin{itemize}
\item $2^{n^2}$ square pictures of width/height $n$, and
\item $|\Delta_q|^{2n}$-many assignments of transitions to the middle part.
\end{itemize}
Thus, for sufficiently large $n$, we can find an accepting run of $\Sys$ on a \txtCBM $\mscn$ whose upper part and lower part do not match, i.e., $\mscn \not\in \Lang(\phi)$.
\end{proof}

However, there is a fragment of \MSO that allows for a positive result (we do not present the proof).

\begin{mytheorem}[\cite{BL-tcs06}]\label{thm:EMSOtoCPDS}
  Suppose $\DS = \Queues$.  Then, for every sentence $\phi \in
  \EMSO(\Arch,\Sigma)=\EMSO(\Sigma,\Procs,\procrel,(\matchrel^{d})_{d\in\DS})$,
  there is a \txtCPDS $\Sys$ such that $\Lang(\Sys) = \Lang(\phi)$.
\end{mytheorem}

\begin{corollary}
  The formula $\ftrans{\Phi_=}$ cannot be expressed in $\EMSO(\Arch,\Sigma)$.
\end{corollary}

\begin{exercise}\label{exo:cpds-complement}
Prove that \CPDSs are, in general, not closed under complementation:
Suppose $\Sigma=\{a,b,c\}$ and assume the architecture $\Arch$ from Theorem~\ref{thm:grids}.
Show that there is $\Sys \in \CPDS(\Arch,\Sigma)$ such that, for all $\Sys' \in \CPDS(\Arch,\Sigma)$, we have $\Lang(\Sys') \neq \CBM(\Arch,\Sigma) \setminus \Lang(\Sys)$.
\end{exercise}

\begin{exercise}
Show that Theorem~\ref{thm:grids} also holds when $|\Procs|=1 $, $|\DS|=2$, and $\DS=\Stacks$.
\end{exercise}

\begin{mytheorem}[\cite{BFG-concur18}]\label{thm:EMSO+toCPDS}
  Suppose $\DS = \Queues$.  Then, for every sentence $\phi \in
  \EMSO(\Sigma,\Procs,<,(\matchrel^{d})_{d\in\DS})$, there is a \txtCPDS $\Sys$
  such that $\Lang(\Sys) = \Lang(\phi)$.
\end{mytheorem}

Notice that $x\procrel y \equiv x<y \wedge \bigvee_{p\in\Procs} p(x)\wedge p(y)
\wedge \neg\exists z~(x<z<y \wedge p(z))$
but $<$ cannot be expressed from $\procrel$ and $\matchrel$ in \EMSO.

\clearpage
\section{Satisfiability and Model Checking}

For an architecture $\Arch$ and an alphabet $\Sigma$, consider the following problems:

\begin{center}
\begin{tabular}{ll}
\toprule
{{\sc{MSO-Satisfiability}}$(\Arch,\Sigma)$}:\\
\midrule
{Instance:} & \hspace{-4em}$\phi \in \MSO(\Arch,\Sigma)$\\[0.5ex]
{Question:} & \hspace{-4em}$L(\phi) \neq \emptyset$\,?\\
\bottomrule
\end{tabular}
\end{center}

\begin{center}
\begin{tabular}{ll}
\toprule
{{\sc{MSO-ModelChecking}}$(\Arch,\Sigma)$}:\\
\midrule
{Instance:} & \hspace{-4em}$\Sys \in \CPDS(\Arch,\Sigma)$\,; $\phi \in \MSO(\Arch,\Sigma)$\\[0.5ex]
{Question:} & \hspace{-4em}$L(\Sys) \subseteq L(\phi)$\,?\\
\bottomrule
\end{tabular}
\end{center}

\bigskip

\begin{mytheorem}
Let $\Arch$ be given as follows (and $\Sigma$ be arbitrary):
\begin{center}
{
\includegraphics[scale=0.5]{simplearch.pdf}
}
\end{center}
Then, all the abovementioned problems are undecidable.
\end{mytheorem}

\fi 

\newcommand{\vDST}[1]{\mathsf{DST}_{\textup{valid}}^{#1}}
\newcommand{\DST}[1]{\mathsf{DST}^#1}
\newcommand{\stwCBM}[1]{\ensuremath{\mathsf{CBM}^{#1\textup{-}\mathsf{stw}}}}
\newcommand{\kstwCBM}{\stwCBM{k}}
\newcommand{\swCBM}[1]{\ensuremath{\mathsf{CBM}^{#1\textup{-}\mathsf{sw}}}}
\newcommand{\kswCBM}{\swCBM{k}}
\newcommand{\EB}[1]{\ensuremath{#1\textup{-}\exists\textup{B}}}
\newcommand{\eEB}{\ensuremath{\exists\textup{B}}}
\newcommand{\NTA}{\textup{{NTA}}\xspace}
\newcommand{\NTAs}{\textup{{NTAs}}\xspace}
\newcommand{\TWA}{\textup{{TWA}}\xspace}
\newcommand{\TWAs}{\textup{{TWAs}}\xspace}
\newcommand{\ATWA}{\textup{{ATWA}}\xspace}
\newcommand{\ATWAs}{\textup{{ATWAs}}\xspace}
\newcommand{\ATA}{\textup{{A2A}}\xspace}
\newcommand{\ATAs}{\textup{{A2As}}\xspace}
\newcommand{\WA}{\textup{{WA}}\xspace}
\newcommand{\WAs}{\textup{{WAs}}\xspace}
\newcommand{\eventof}[1]{#1}
\newcommand{\stay}{\mathsf{id}}
\newcommand{\type}{\mathit{type}}
\newcommand{\wafinal}{\mathsf{f}}
\newcommand{\dir}{\mathit{dir}}

\newcommand{\Context}[1]{\mathsf{Context}_{#1}}
\newcommand{\AllContext}{\mathsf{Context}}
\newcommand{\Scope}[1]{\mathsf{Scope}_{#1}}
\newcommand{\AllScope}{\mathsf{Scope}}
\newcommand{\Phase}[1]{\mathsf{Phase}_{#1}}
\newcommand{\AllPhase}{\mathsf{Phase}}

\newcommand{\ebMSCs}[1]{\CBM_{\exists{#1}}}
\newcommand{\ubMSCs}[1]{\CBM_{\forall{#1}}}
\newcommand{\allebMSCs}{\CBM_{\exists}}
\newcommand{\allubMSCs}{\CBM_{\forall}}

\ifUnderapprox

\chapter{Underapproximate Verification}

Recall that most verification problems such as nonemptiness, global-state reachability, and model checking are undecidable even for very simple architectures.

\section{Principles of Underapproximate Verification}

To get decidability, we will restrict decision problems to a subclass $\mathcal{C} \subseteq \CBM(\Arch,\Sigma)$ of \CBMs:

\begin{center}
\begin{tabular}{ll}
\toprule
{{\sc{MSO-Validity}}$(\Arch,\Sigma,\C)$}:\\
\midrule
{Instance:} & \hspace{-4em}$\phi \in \MSO(\Arch,\Sigma)$\\[0.5ex]
{Question:} & \hspace{-4em}$\C \subseteq L(\phi)$\,?\\
\bottomrule
\end{tabular}
\end{center}

\begin{center}
\begin{tabular}{ll}
\toprule
{{\sc{MSO-ModelChecking}}$(\Arch,\Sigma,\C)$}:\\
\midrule
{Instance:} & \hspace{-4em}$\Sys \in \CPDS(\Arch,\Sigma)$\,; $\phi \in \MSO(\Arch,\Sigma)$\\[0.5ex]
{Question:} & \hspace{-4em}$L(\Sys) \cap \C\subseteq L(\phi)$\,?\\
\bottomrule
\end{tabular}
\end{center}

For example, we may only consider the \CBMs that can be executed when the data structures have bounded capacity:

\begin{definition}
Let $k \ge 0$. A \txtCBM $\mscn$ is called $k$-\emph{existentially} bounded (\EB{k} for short) if there is a linearization $w \in \Lin(\mscn)$ such that, for every prefix $u$ of $w$, the number of unmatched writes in $u$ is at most $k$. A class $\mathcal{C} \subseteq \CBM(\Arch,\Sigma)$ is $\EB{k}$ if $\mscn$ is $\EB{k}$ for every $\mscn \in \mathcal{C}$. Finally, $\mathcal{C}$ is called $\eEB$ if it is $\EB{k}$ for some $k$.
\end{definition}

\begin{example}\label{ex:underappr}
We will give some examples:
\begin{itemize}
\item[(a)] The \txtCBM from Example~\ref{ex:cbm} is \EB{3}.
\item[(b)] The class of nested words ($|\Procs| = 1$, $|\DS| = 1$, and $\DS = \Stacks$) is not $\eEB$, as illustrated by the following figure:
\begin{center}
\fbox{
\includegraphics[scale=0.5]{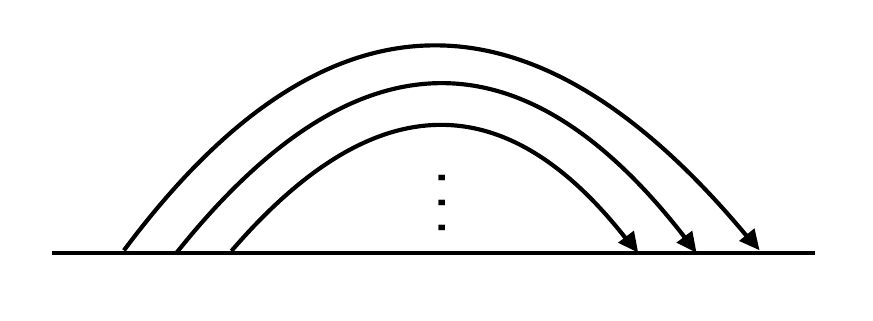}
}
\end{center}
\item[(c)] The class of MSCs ($|\Procs| \ge 2$,
$\DS = \Queues = \Procs \times \Procs \setminus \{(p,p) \mid p \in \Procs\}$,
$\writer(p,q) = p$, and $\reader(p,q) = q$) is not $\eEB$:
\begin{center}
\fbox{
\includegraphics[scale=0.5]{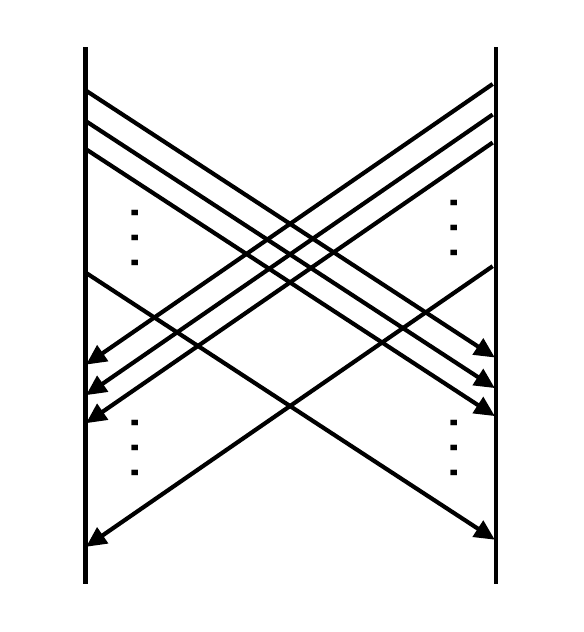}
}
\end{center}
\end{itemize}
\end{example}

\begin{exercise}
Consider the encoding of pictures as \CBMs from Section~\ref{sec:exprMSO}. Show that the encoding of a picture of height $k$ yields a \txtCBM that is $\EB{k}$.
\end{exercise}

We are looking for ``reasonable'' classes of \CBMs that are suitable for underapproximate verification.

\begin{definition}
Let $\C=(\C_k)_{k \ge 0}$ with $\C_k \subseteq \CBM(\Arch,\Sigma)$ be a family of classes of \CBMs.
Then, $\C$ is called
\begin{itemize}
\item \emph{monotone} if $\C_k \subseteq \C_{k+1}$ for all $k \ge 0$,
\item \emph{complete} if $\bigcup_{k \ge 0} \C_k = \CBM(\Arch,\Sigma)$,
\item \emph{decidable} if the usual decision problems are decidable when the domain of \CBMs is restricted to $\C_k$,
\item \emph{MSO-definable} if, for all $k \ge 0$, there is a sentence $\phi_k \in \MSO(\Arch,\Sigma)$ such that $L(\phi_k) = \C_k$, and
\item \emph{CPDS-definable} if, for all $k \ge 0$, there is a CPDS $\Sys_k \in \CPDS(\Arch,\Sigma)$ such that $L(\Sys_k) = \C_k$.
\end{itemize}
\end{definition}

Below, we first present a generic family, which is based on the notion of 
\emph{special tree-width}.

\newcommand{\phicbm}{\ensuremath{\Phi_\mathsf{cbm}}\xspace}

\section{Graph-Theoretic Approach}

In the following, we will use tools from graph theory.  Actually, the pair
$(\Arch,\Sigma)$ defines a \emph{signature} of unary ($\Procs\uplus\Sigma$) and
binary ($\{\procrel,\matchrel^{d}\mid d\in\DS\}$) relation symbols.  Thus, one can consider general
graphs over $(\Arch,\Sigma)$, with node labels from $\Sigma'=\Procs\uplus\Sigma$ and edge labels
from $\Gamma=\{\mathsf{succ}\} \cup \DS$.  Here, $\mathsf{succ}$ stands for
\emph{process successor}, and $d\in\DS$ is the labeling of an edge that connects
a write and a read event.  Those graphs that satisfy the axioms from
Definition~\ref{def:cbm} can then be considered as \CBMs.

\begin{proposition}\label{prop:mso-cbm}
  Let $\Arch$ be an architecture.  The class $\CBM(\Arch,\Sigma)$ is
  MSO-definable, i.e., there is an $\MSO(\Arch,\Sigma)$ sentence $\phicbm$ such
  that for all $(\Sigma',\Gamma)$-labeled graphs $G$ we have
  $G\models\phicbm$ iff $G\in\CBM(\Arch,\Sigma)$.
\end{proposition}

\begin{proof}
  Let $G=(V,(V_\sigma)_{\sigma\in\Sigma'},(E_\gamma)_{\gamma\in\Gamma})$ be a
  $(\Sigma',\Gamma)$-labelled graph with $V_\sigma\subseteq V$ and 
  $E_\gamma\subseteq V^{2}$.  The formula $\phicbm$ has to check that
  all conditions of Definition~\ref{def:cbm} are satisfied.
  
  This is left as an exercise.
\end{proof}

\begin{proposition}\label{prop:SAT-MC-reduction}
Fix a class $\C \subseteq \CBM(\Arch,\Sigma)$. The following problems are inter-reducible:
\begin{enumerate}
\item {{\sc{MSO-Validity}}$(\Arch,\Sigma,\C)$}
\item {{\sc{MSO-ModelChecking}}$(\Arch,\Sigma,\C)$}
\end{enumerate}
\end{proposition}

\begin{proof}
  For the reduction from 1.\ to 2., let $\phi \in \MSO(\Arch,\Sigma)$ be a
  sentence.  Let $\Sys_\mathsf{univ}$ be the \emph{universal} CPDS, satisfying
  $L(\Sys_\mathsf{univ}) = \CBM(\Arch,\Sigma)$.  Note that we can define
  $\Sys_\mathsf{univ}$ such that $|\Locs_p| = 1 = |\Val|$ and where we have full
  transition tables.  We have:
  \[\C \subseteq L(\phi) \textup{~~~iff~~~} L(\Sys_\mathsf{univ}) \cap \C \subseteq L(\phi)\]

For the reduction from 2.\ to 1., let $\Sys \in \CPDS(\Arch,\Sigma)$ and $\phi \in \MSO(\Arch,\Sigma)$. Let $\phi_\Sys \in \MSO(\Arch,\Sigma)$ such that $L(\phi_\Sys) = L(\Sys)$ (cf.\ Theorem~\ref{thm:CPDStoMSO}). We have:
\[
\begin{array}{rl}
& L(\Sys) \cap \C \subseteq L(\phi)\\[0.5ex]
\textup{iff~~} & L(\phi_\Sys) \cap \C \subseteq L(\phi)\\[0.5ex]
\textup{iff~~} & \C \subseteq L(\phi \vee \neg \phi_\Sys)
\end{array}
\]
\end{proof}

The decidability of the MSO theory of classes of graphs has been extensively studied (cf.\ the book by Courcelle and Engelfriet \cite{CourcelleBook}):

\begin{theorem}\label{thm:main-graph}
  Let \class be a class of bounded degree graphs which is MSO-definable. The following statements are equivalent:
  \begin{enumerate}[nosep]
    \item $\C$ has a decidable MSO theory.
  
    \item $\C$ can be interpreted in binary trees.
  
    \item $\C$ has bounded tree-width.
  
    \item $\C$ has bounded clique-width.
  \end{enumerate}
\end{theorem}

For a class $\C \subseteq \CBM(\Arch,\Sigma)$ that is MSO-definable, we prove 
\emph{bounded (special) tree-width}
\begin{itemize}
\item to get decidability,
\item to get the interpretation in binary trees,
\item to reduce verification problems to problems on tree automata, and
\item to get efficient algorithms with optimal complexity.
\end{itemize}

In the theorem above, graphs are interpreted in binary trees.
We need to identify which trees are \emph{valid encodings}, i.e., do encode graphs in the class $\C$.
This is why we assumed the class of graphs to be MSO-definable.
From this, we can build a tree automaton for the valid encodings.
Actually, we can replace MSO-definability of the class $\C$ by the existence of
a tree automaton $\A_\C$ for the \emph{valid encodings} of CBMs in $\C$.
It is often better to define the tree automaton directly.
Its size has a direct impact on the decision procedures arising from the 
tree-interpretation.

\section{Graph (De)composition and Tree Interpretation}\label{sec:cographs}


We will illustrate the concept of tree interpretation by means of \emph{cographs}. Undirected and labeled cographs are generated by \emph{cograph terms}. A cograph term is built from the grammar (cograph algebra)
\[C ~::=~ a ~\mid~ C  \oplus C ~\mid~ C \otimes C\]
where $a \in \Sigma$. A term $C$ defines a cograph $\sem{C}=(V,E,\lambda)$ as follows:
\begin{itemize}
\item $\sem{a}$ is the graph $(\{1\},\emptyset,1 \mapsto a)$ with one $a$-labeled vertex and no edges

\item if $\sem{C_i} = (V_i,E_i,\lambda_i)$ for $i = 1,2$, with $V_1 \cap V_2 = \emptyset$, then
\[\sem{C_1 \oplus C_2} = (V_1 \cup V_2,E_1 \cup E_2,\lambda_1 \cup \lambda_2)\]

\item if $\sem{C_i} = (V_i,E_i,\lambda_i)$ for $i = 1,2$, with $V_1 \cap V_2 = \emptyset$, then
\[\sem{C_1 \otimes C_2} = (V_1 \cup V_2,E_1 \cup E_2 \cup \underbrace{(V_1 \times V_2) \cup (V_2 \times V_1)}_{\textup{undirected graphs}},\lambda_1 \cup \lambda_2)\]
(this is called the \emph{complete join})
\end{itemize}

\begin{example}
Consider the cograph term $C=((a \otimes a) \otimes b) \otimes (a \oplus (a \otimes b))$. The figure below shows a tree representation of $C$ as well as the cograph $\sem{C}$ defined by $C$.
\begin{center}
\fbox{
\includegraphics[scale=0.5]{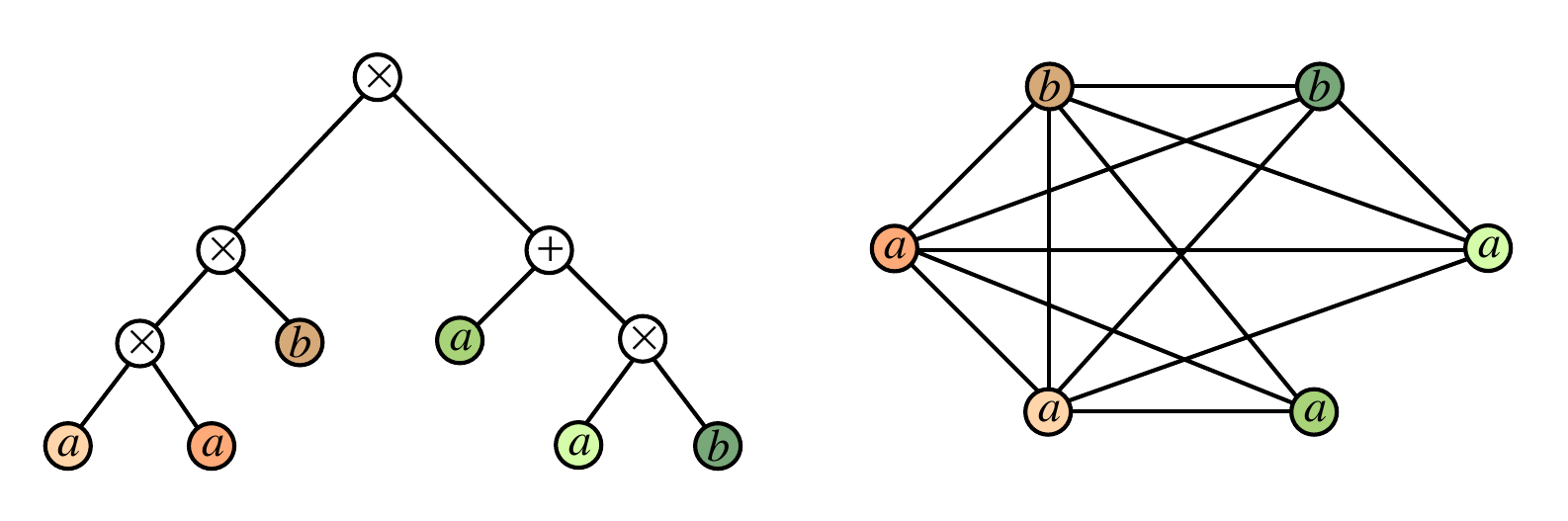}
}
\end{center}
Note that the tree representation offers a top-down decomposition of a cograph.
\end{example}

\begin{remark}
One can also include edge labelings $d$, using operators $\otimes_d$ with the expected meaning.
\end{remark}

Clearly, the set of cograph terms (considered as trees) is a regular tree language:

\begin{lemma}\label{lem:validcographs}
  The set of cograph terms is MSO-definable in binary trees, i.e. with a
  sentence $\varphi_\textup{cograph}\in
  \MSO(\Sigma\cup\{\oplus,\otimes\},\downarrow_0,\downarrow_1)$ where
  $\downarrow_0$ stands for ``left child'' and $\downarrow_1$ for ``right
  child''.  \\
  Alternatively, there is a tree automaton which accepts the set of binary 
  trees that are cograph terms.
\end{lemma}

\begin{proof}
  The set of ``valid'' binary trees is defined by the sentence
  \[
  \phi_\textup{cograph} = \forall x \left(
  \begin{array}{rl}
    & \mathit{leaf}(x) \to \bigvee_{a \in \Sigma} a(x)\\[1ex]
    \wedge & \neg \mathit{leaf}(x) \to \oplus(x) \vee \otimes(x) \\[1ex]
    \wedge & (\exists y~ x\downarrow_0 y) \leftrightarrow (\exists y~ 
    x\downarrow_1 y)
  \end{array} \right)
  \]
  where $\mathit{leaf}(x)=\neg\exists y~ x\downarrow y$ and
  $x\downarrow y = x\downarrow_0 y \vee x\downarrow_1 y$.
\end{proof}

\underline{\bf{Tree interpretation:}}

Next, we demonstrate that a cograph can be \emph{recovered} from its cograph term using MSO formulas, i.e., an \emph{MSO interpretation}. Let $C$ be a cograph term (i.e., a binary tree), and let $G = \sem{C} = (V,E,\lambda)$.
\begin{itemize}
\item The nodes in $V$ correspond to the leaves of $C$:
$\phi_\textsf{vertex}(x) = \mathit{leaf}(x)$.
\item The set $E$ of edges is defined by $\phi_\textsf{edge}(x,y)$
\[
\begin{array}{rl}
= & \textup{``}\mathit{least\textup{-}common\textup{-}ancestor}(x,y) \textup{ is labeled } \otimes\textup{''}\\[1ex]
= &
\exists z,x',y'
\left(
\begin{array}{rl}
& z \downarrow_0 x' ~\wedge~ z \downarrow_1 y'\\
\vee & z \downarrow_1 x' ~\wedge~ z \downarrow_0 y'
\end{array}
\right)
\wedge (x' \downarrow^\ast x) \wedge (y' \downarrow^\ast y) \wedge \otimes(z)
\end{array}
\]
\end{itemize}

\begin{proposition}[Backward translation]
  MSO formulas over cographs can be ``translated'' to MSO formulas over cograph terms:
  For all sentences $\phi \in \MSO(\Sigma,E)$, there is a sentence $\widetilde{\phi} \in \MSO(\Sigma \cup \{\oplus,\otimes\},\downarrow_0,\downarrow_1)$ such that, for every cograph term $C$, say with $G = \sem{C}$, we have
  \[G \models \phi \textup{~~~iff~~~} C \models \widetilde{\phi}\,.\]
\end{proposition}

\begin{proof}
We proceed by induction on $\phi$:
\begin{itemize}
\item $\widetilde{a(x)} = a(x)$
\item $\widetilde{x E y} = \phi_\mathsf{edge}(x,y)$
\item $\widetilde{x \in X} = x \in X$
\item $\widetilde{x = y} = (x = y)$
\item $\widetilde{\neg\phi} = \neg\,\widetilde{\phi}$
\item $\widetilde{\phi_1 \vee \phi_2} = \widetilde{\phi_1} \vee \widetilde{\phi_2}$
\item $\widetilde{\exists x\phi} = \exists x(\phi_\textup{vertex}(x) \wedge \widetilde{\phi})$
\item $\widetilde{\exists X\phi} = \exists X((\forall x (x \in X \to \phi_\textup{vertex}(x))) \wedge \widetilde{\phi})$
\end{itemize}
Note that, in the correctness proof, we have to deal with free variables. Actually, the inductive statement is as follows: For all $\phi \in \MSO(\Sigma,E)$, there is $\widetilde{\phi} \in \MSO(\Sigma \cup \{\oplus,\otimes\},\downarrow_0,\downarrow_1)$ such that, for every cograph term $C$ and every interpretation $\mathcal{I}$ of $\Var \cup \VAR$ in $\mathit{Leaves}(C) = \mathit{Vertices}(G = \sem{C})$, we have $G \models_\mathcal{I} \phi$ iff $C \models_\mathcal{I} \widetilde{\phi}$.
\end{proof}

\begin{corollary}
The MSO theory of cographs is decidable.
\end{corollary}

\begin{proof}
  Let $\phi \in \MSO(\Sigma,E)$.  Then, $\phi$ is valid on cographs iff
  $\phi_\textup{cograph} \to \widetilde{\phi}$ is valid on binary trees (cf.\
  Lemma~\ref{lem:validcographs}).  Note that, moreover, $\phi$ is satisfiable on
  cographs iff $\phi_\textup{cograph} \wedge \widetilde{\phi}$ is satisfiable on
  binary trees.  The corollary follows, since MSO validity (satisfiability) is
  decidable on binary trees.  Indeed, the problem can be reduced to
  tree-automata emptiness, see Thatcher and Wright \cite{ThaWri68}.
\end{proof}

\begin{gpicture}[ignore, name = cbm-stt]
  \gasset{Nw=1.5, Nh = 1.5, Nfill = n, Nframe = n}
  \node(p)(-6,11.5){$p$}
  \node(q)(-6,-0.5){$q$}
  \gasset{Nfill=y,NLangle=90,NLdist=2.5}
  
  \node(p3)(0,12){}\nodelabel(p3){$b$}\nodelabel[NLangle=-90](p3){$2$}
  \node(p4)(14,12){}\nodelabel(p4){$c$}
  \node(p5)(28,12){}\nodelabel(p5){$d$}\nodelabel[NLangle=-90](p5){$3$}
  \node(q2)(0,0){}\nodelabel(q2){$a$}\nodelabel[NLangle=-90](q2){$0$}
  \drawedge(p3,p4){}\drawedge(p4,p5){}
  \drawedge[ELside=r](q2,p4){$c_2$}
  \drawedge[curvedepth = 5](p3,p5){$s$}
\end{gpicture}

\newpage
\section{Special tree-width} 
\label{sec:stt}

\newcommand{\atomicSTT}{\textsf{atomicSTT}\xspace}
\newcommand{\STT}{\textsf{STT}\xspace}
\newcommand{\STTs}{\textsf{STTs}\xspace}
\newcommand{\kSTT}{$k$-\textsf{STT}\xspace}
\newcommand{\kSTTs}{$k$-\textsf{STTs}\xspace}
\newcommand{\STW}{\textsf{STW}\xspace}
\newcommand{\STWs}{\textsf{STWs}\xspace}
\newcommand{\stt}{\tau}
\newcommand{\add}[2]{\mathop{\textsf{Add}_{#1}^{#2}}}
\newcommand{\forget}[1]{\mathop{\textsf{Forget}_#1}}
\newcommand{\phiksttvalid}{\ensuremath{\Phi_\mathsf{valid}^{k\text{-}\mathsf{stt}}}\xspace}

\newcommand{\rename}[2]{\mathop{\textsf{Rename}_{#1,#2}}}
\newcommand{\sttunion}{\oplus}
\newcommand{\Nat}{\mathbb{N}}
\newcommand{\dom}{\textsf{dom}}

\newcommand{\da}{{\downarrow}}
\newcommand{\ua}{{\uparrow}}

In this section, we 
introduce special tree terms (\STTs) and their semantics as labeled graphs.
A special tree term using at most $k+1$ colors (\kSTT) defines a graph of special 
tree-width at most $k$. Special tree-width is similar to tree-width. See 
\cite{Courcelle10} for more details on special tree-width and tree-width.
We also give an \MSO interpretation of the graph $G_\tau$ defined by a special 
tree term $\tau$ in the binary tree associated with $\tau$.

A $(\Sigma,\Gamma)$-labeled graph is a tuple
$G=(V,(V_a)_{a\in\Sigma},(E_\gamma)_{\gamma\in\Gamma})$ where $V_a\subseteq V$ is
the set of vertices labeled $a$ and $E_\gamma \subseteq V^2$ is the set of edges for each
label $\gamma\in\Gamma$.

\subparagraph{Special tree terms} form an algebra to define labeled graphs.  The
syntax of \kSTTs over $(\Sigma,\Gamma)$ is given by
$$
\stt ::= i \mid \add{i}{a} \stt \mid \add{i,j}{\gamma} \stt \mid \forget{i} \stt \mid
\rename{i}{j} \stt \mid \stt \sttunion \stt 
$$
where $a \in \Sigma$, $\gamma \in \Gamma$ and $i,j\in[k]=\{0,1,\ldots,k\}$ are
colors.  

Each \kSTT represents a colored graph $\sem\stt=(G_\tau,\chi_\tau)$ where
$G_\tau$ is a $(\Sigma,\Gamma)$-labeled graph
and $\chi_\tau\colon [k]\to V$ is a partial injective function coloring some vertices of
$G_\tau$.
\begin{itemize}
  \item $\sem{i}$ consists of a single vertex with color $i$.

  \item $\add{i}{a}$ adds label $a$ to the vertex colored $i$, if it exists.

  Formally, if $\sem{\tau}=(V,(V_a)_{a\in\Sigma},(E_\gamma)_{\gamma\in\Gamma},\chi)$ then
  $\sem{\add{i}{b}\tau}=(V,(V'_a)_{a\in\Sigma},(E_\gamma)_{\gamma\in\Gamma},\chi)$ 
  with $V'_a=V_a$ if $a\neq b$ and
  $V'_b= \begin{cases}
    V_b & \text{if } \{i\}\not\subseteq\dom(\chi) \\
    V_b\cup\{\chi(i)\} & \text{otherwise.}
  \end{cases}$
  
  \item $\add{i,j}{\gamma}$ adds a $\gamma$-labeled edge to the vertices
  colored $i$ and $j$ (if such vertices exist).
  
  Formally, if $\sem{\tau}=(V,(V_a)_{a\in\Sigma},(E_\gamma)_{\gamma\in\Gamma},\chi)$ then
  $\sem{\add{i,j}{\alpha}\tau}=(V,(V_a)_{a\in\Sigma},(E'_\gamma)_{\gamma\in\Gamma},\chi)$ 
  with $E'_\gamma=E_\gamma$ if $\gamma\neq\alpha$ and
  $E'_\alpha= \begin{cases}
    E_\alpha & \text{if } \{i,j\}\not\subseteq\dom(\chi) \\
    E_\alpha\cup\{(\chi(i),\chi(j))\} & \text{otherwise.}
  \end{cases}$
  
  \item $\forget{i}$ removes color $i$ from the domain of the color map.
  
  Formally, if $\sem{\tau}=(V,(V_a)_{a\in\Sigma},(E_\gamma)_{\gamma\in\Gamma},\chi)$ then
  $\sem{\forget{i}\tau}=(V,(V_a)_{a\in\Sigma},(E_\gamma)_{\gamma\in\Gamma},\chi')$ 
  with $\dom(\chi')=\dom(\chi)\setminus\{i\}$ and $\chi'(j)=\chi(j)$ for all 
  $j\in\dom(\chi')$.
  
  \item $\rename{i}{j}$ exchanges the colors $i$ and $j$.
  
  Formally, if $\sem{\tau}=(V,(V_a)_{a\in\Sigma},(E_\gamma)_{\gamma\in\Gamma},\chi)$ then
  $\sem{\rename{i}{j}\tau}=(V,(V_a)_{a\in\Sigma},(E_\gamma)_{\gamma\in\Gamma},\chi')$ 
  with $\chi'(\ell)=\chi(\ell)$ if $\ell\in\dom(\chi)\setminus\{i,j\}$, 
  $\chi'(i)=\chi(j)$ if $j\in\dom(\chi)$ and $\chi'(j)=\chi(i)$ if 
  $i\in\dom(\chi)$.
  
  \item Finally, $\sttunion$ constructs the disjoint union of the two graphs
  provided they use different colors.  This operation is undefined otherwise.
  
  Formally, if $\sem{\tau_i}=(G_i,\chi_i)$ for $i=1,2$ and 
  $\dom(\chi_1)\cap\dom(\chi_2)=\emptyset$ then 
  $\sem{\tau_1\sttunion\tau_2}=(G_1\uplus G_2,\chi_1\uplus\chi_2)$.
  Otherwise, $\tau_1\sttunion\tau_2$ is not a valid \STT.
\end{itemize}
The special tree-width of a graph $G$ is the least $k$ such that 
$G=G_\tau$ for some $\kSTT$ $\stt$.

\begin{example}
  For \CBMs, we have process edges and data edges, so we take
  $\Gamma=\{\procrel\}\cup\DS$.  Also, vertices of \CBMs are labeled with a letter
  from $\Sigma$ and a process from $\Procs$.  Hence the labels of vertices are
  from $\Sigma'=\Sigma\uplus\Procs$.
  For $i\in[k]$, $a\in\Sigma$ and $p\in\Procs$, we simply write
  $(i,a,p)=\add{i}{p}~ \add{i}{a}~ i$.

  Consider the following 3-\STT: 
  $$
  \stt=\forget{1} \add{1,3}{\procrel} \add{2,1}{\procrel}
  (\add{0,1}{c_2}((0,a,q)\sttunion(1,c,p))
  \sttunion
  \add{2,3}{s}((2,b,p)\sttunion(3,d,p)))
  $$
  It is depicted below (left) as a tree and the graph $G_\stt$ is 
  given on the right.
\end{example}

\begin{minipage}{80mm}
  \Tree
    [.$\forget{1}$ [.$\add{1,3}{\procrel}$ [.$\add{2,1}{\procrel}$
       [.$\sttunion$
        [.$\add{0,1}{c_2}$ [.$\sttunion$ [.$(0,a,q)$ ] [.$(1,c,p)$ ] ] ]
        [.$\add{2,3}{s}$   [.$\sttunion$ [.$(2,b,p)$ ] [.$(3,d,p)$ ] ] ]
       ]
    ] ] ]
\end{minipage}
\hfil
\gusepicture{cbm-stt}

\begin{definition}
  Let $\kstwCBM(\Arch,\Sigma)$ denote the set of CBMs with special tree-width
  bounded by $k$.
\end{definition}

\begin{exercise}
  Prove that trees have special tree-width (at most) 1.
\end{exercise}

\begin{exercise}
  Give a 3-\STT for the following graph:
  \raisebox{-2mm}{\includegraphics[scale=0.4, page=4]{stw-export.pdf}}
\end{exercise}

\begin{exercise}
  Show that the rename operation is redundant. More precisely, show that for 
  every \kSTT $\stt$ (possibly using the rename operation) we can construct a 
  \kSTT $\stt'$ which does not use the rename operation and such that 
  $\sem{\tau}=\sem{\tau'}$, i.e., $G_\stt=G_{\stt'}$ and
  $\chi_\tau=\chi_{\tau'}$.
\end{exercise}

\clearpage
\section{Decomposition game for special tree-width}
\label{sec:game-stw}

\paragraph{The decomposition game for special tree-width} is a two player turn
based game $\mathsf{Arena}(\Sigma,\Gamma)=(\mathsf{Pos}_\exists\uplus
\mathsf{Pos}_\forall,\mathsf{Moves})$.  Eve's set of positions
$\mathsf{Pos}_\exists$ consists of marked (colored) graphs $(G,U)$ where
$G=(V,(V_a)_{a\in\Sigma},(E_\gamma)_{\gamma\in\Gamma})$ is a $(\Sigma,\Gamma)$-labeled
graph and $U\subseteq V$ is the subset of marked vertices.  Adam's set of
positions $\mathsf{Pos}_\forall$ consists of pairs of marked graphs.
The edges $\mathsf{Moves}$ of $\mathsf{Arena}$ reflect the moves of the players.
Eve's moves from $(G,U)$ consist in 
\begin{enumerate}[nosep]
  \item  marking some vertices of the graph resulting in $(G,U')$ with 
  $U\subseteq U'\subseteq V$,

  \item  removing edges whose endpoints are marked, resulting in $(G',U)$ with 
  $E'_\gamma\subseteq E_\gamma\subseteq E'_\gamma \cup U\times U$ for all 
  $\gamma\in\Gamma$,

  \item dividing $(G,U)$ in $(G_1,U_1)$ and $(G_2,U_2)$ such that $G$ is the
  disjoint union of $G_1$ and $G_2$ (in particular $V_1\cap V_2=\emptyset$ and
  $V=V_1\cup V_2$) and marked nodes are inherited ($U_1=U\cap V_1$ and
  $U_2=U\cap V_2$).
\end{enumerate}
Adam's moves amount to choosing one of the two marked graphs.
Terminal positions of the game are graphs where all vertices are marked: $U=V$. 
Neither Eve nor Adam can move from terminal positions which are winning for Eve.

A (maximal) play is a path in $\mathsf{Arena}$ starting from some marked graph
$(G,U)$ and leading to a terminal position.  The cost of the play is the maximum
number of marked vertices in the positions of the path.  Eve's objective is to
minimize the cost and Adam's objective is to maximize the cost.

A (positional) strategy for Eve starting from a marked graph $(G,U)$ is
$k$-winning if all plays starting from $(G,U)$ and following the strategy
have cost at most $k+1$.  

\begin{theorem}
  The special tree-width of a graph $G$ is the least $k$ such that Eve has a
  $k$-winning strategy starting from $(G,\emptyset)$ (initially $G$ is
  unmarked).
\end{theorem}

\begin{exercise}
  Prove that a $k$-winning strategy for Eve starting from $(G,U)$ can be
  described with a valid \kSTT $\tau$ with $\sem{\tau}=(G,\chi)$ and
  $\dom(\chi)=U$.
\end{exercise}

\includegraphics[scale=0.4, page=1]{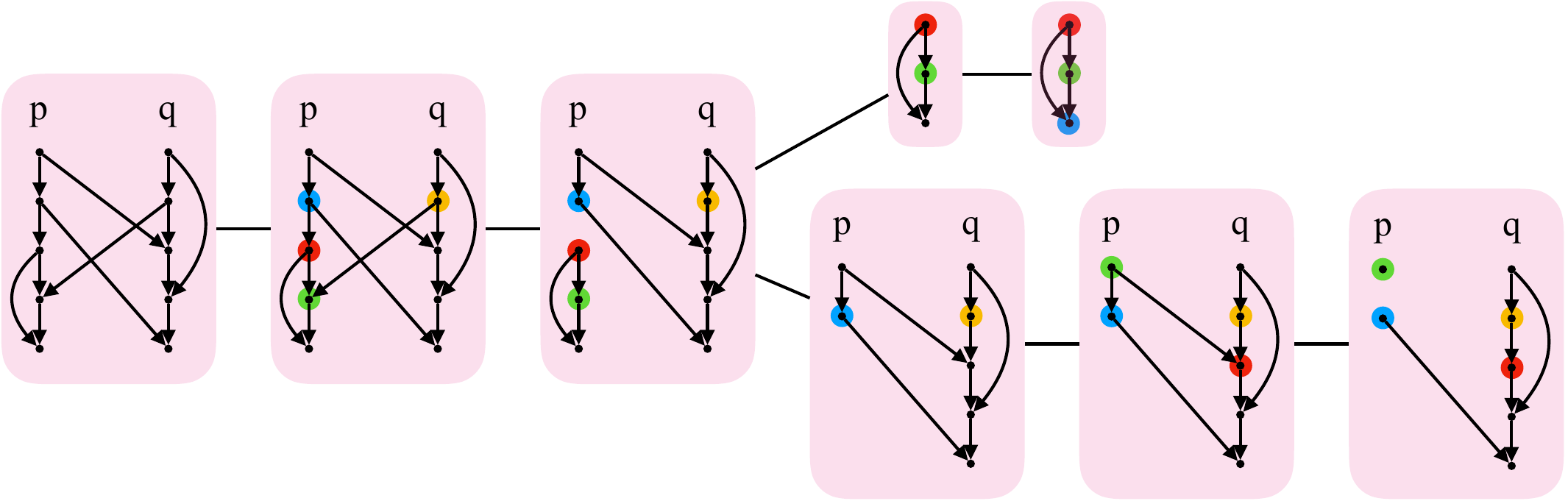}

\begin{example}
  The CBM on the left of the Figure above has \STW at most 3.  The beginning of
  a 3-winning strategy for Eve is depicted as a tree.  She starts by marking
  four nodes, then she removes two edges and the resulting graph is
  disconnected.  The component with three vertices can easily be made terminal
  by marking the last node.  On the second component (bottom branch), Eve marks
  two more vertices and removes two edges so that the green node is
  disconnected.  On the remaining component, she marks the last node and removes
  one edge so that the blue node is disconnected.  Then she marks the second to
  last node and removes one edge disconnecting the last node.  Finally, she
  marks the first node reaching a terminal position.
\end{example}

\hfil\includegraphics[scale=0.6, page=2]{stw-export.pdf}

\begin{example}\label{ex:stw-nw}
  Nested words have special tree-width bounded by 3.

  We describe a $3$-winning strategy for Eve. First she marks the 
  first and last point of the word. Then she repeats the following steps until 
  reaching a terminal position:
  \begin{enumerate}[nosep]
    \item If the first point is not a push (the source of a $\matchrel$) then 
    she marks the second point and she removes the first linear edge. 
    The graph is disconnected. One component is terminal. The other one is a 
    nested word with the endpoints marked.
  
    \item If there is a $\matchrel$-edge from the first point to the last 
    point, then Eve removes this edge and continues as above.
  
    \item If there is a $\matchrel$-edge from the first point, call it $e$, to
    some middle point, call it $f$, then Eve marks $f$ and its linear successor
    $g$.  Then she removes the matching edge $e\matchrel f$ and the successor
    edge $f\procrel g$.  The resulting graph is disconnected.  Both connected
    components are nested words.  Each one has its endpoints marked.
  \end{enumerate}
\end{example}

\includegraphics[scale=0.6, page=3]{stw-export.pdf}

\begin{example}\label{ex:stw-mnw-bounded-context}
  Multiply nested words with at most $k\geq2$ contexts have special tree-width
  bounded by $2k-1$.
  
  Eve marks each endpoint of each context.  Doing so, she marks at most $2k$
  vertices.  Then, she removes the successor edges between contexts.  The graph
  is disconnected.  Each connected component is a (simply) nested word with at
  most $\ell$ contexts with $2\ell-1\leq k$, hence, at most $k+1$ marked 
  vertices.
  Using two extra marks, Eve applies on each connected component the strategy
  described in Example~\ref{ex:stw-nw}.
\end{example}

\begin{example}\label{ex:stw-exist-bounded}
  Let $n,k\geq1$.  \CBMs over $n$ processes and which are $k$-existentially
  bounded have special tree-width bounded by $k+n$.
  
  Let $\mscn$ be a \CBM and let $e_1,e_2,\ldots,e_m$ be a $k$-bounded 
  linearisation of the events in $\mscn$. Notice that for all $1\leq\ell\leq m$ 
  we have
  \begin{align*}
    |\{(i,j)\mid i\leq\ell<j \wedge e_i\matchrel e_j \}| & \leq k \\
    |\{(i,j)\mid i\leq\ell<j \wedge e_i\procrel e_j \}| & \leq n
  \end{align*}
  Initially, Eve's strategy is to mark the first $r=n+k$ vertices and to remove
  all edges between the marked vertices.  Write $(G_r,U_r)$ for the resulting
  marked graph.  Notice that $|U_r|=n+k$, $U_r=G_r\cap\{e_1,\ldots,e_r\}$ and
  there are no edges between the vertices in $U_r$.  This will be an invariant
  of the construction.  Now, while $r<m$, Eve's strategy proceeds as follows:
  \begin{enumerate}
    \item she marks vertex $e_{r+1}$ and remove all edges between $U_r$ and 
    $e_{r+1}$. 
    In the resulting marked graph $(G'_{r},U'_{r})$, at least one vertex 
    $e\in U'_{r}=U_r\cup\{e_{r+1}\}$ is isolated.
  
    \item she divides $(G'_{r},U'_r)$ in the isolated vertex $e$ and the rest
    $(G_{r+1},U_{r+1})$.  To avoid losing immediately, Adam must choose
    $(G_{r+1},U_{r+1})$ and this marked graph satisfies the invariant.
  \end{enumerate}
\end{example}

\section{Main Results} 

Consider the following problems:

\begin{center}
\begin{tabular}{ll}
\toprule
{{\sc{stw-Nonemptiness}}$(\Arch,\Sigma)$}:\\
\midrule
{Instance:} & \hspace{-4em}$\Sys \in \CPDS(\Arch,\Sigma)$\,; $k>0$\\[0.5ex]
{Question:} & \hspace{-4em}$L(\Sys) \cap \kstwCBM \neq \emptyset$\,?\\
\bottomrule
\end{tabular}
\end{center}

\begin{center}
\begin{tabular}{ll}
\toprule
{{\sc{stw-Inclusion}}$(\Arch,\Sigma)$}:\\
\midrule
{Instance:} & \hspace{-4em}$\Sys,\Sys' \in \CPDS(\Arch,\Sigma)$\,; $k>0$\\[0.5ex]
{Question:} & \hspace{-4em}$L(\Sys) \cap \kstwCBM \subseteq L(\Sys')$\,?\\
\bottomrule
\end{tabular}
\end{center}

\begin{center}
\begin{tabular}{ll}
\toprule
{{\sc{stw-Universality}}$(\Arch,\Sigma)$}:\\
\midrule
{Instance:} & \hspace{-4em}$\Sys \in \CPDS(\Arch,\Sigma)$\,; $k>0$\\[0.5ex]
{Question:} & \hspace{-4em}$\kstwCBM \subseteq L(\Sys)$\,?\\
\bottomrule
\end{tabular}
\end{center}

\begin{center}
\begin{tabular}{ll}
\toprule
{{\sc{stw-Satisfiability}}$(\Arch,\Sigma)$}:\\
\midrule
{Instance:} & \hspace{-4em}$\phi \in \MSO(\Arch,\Sigma)$; $k>0$\\[0.5ex]
{Question:} & \hspace{-4em}$L(\phi) \cap \kstwCBM \neq \emptyset$\,?\\
\bottomrule
\end{tabular}
\end{center}

\begin{center}
\begin{tabular}{ll}
\toprule
{{\sc{stw-Validity}}$(\Arch,\Sigma)$}:\\
\midrule
{Instance:} & \hspace{-4em}$\phi \in \MSO(\Arch,\Sigma)$; $k>0$\\[0.5ex]
{Question:} & \hspace{-4em}$\kstwCBM \subseteq L(\phi)$\,?\\
\bottomrule
\end{tabular}
\end{center}

\begin{center}
\begin{tabular}{ll}
\toprule
{{\sc{stw-ModelChecking}}$(\Arch,\Sigma)$}:\\
\midrule
{Instance:} & \hspace{-4em}$\Sys \in \CPDS(\Arch,\Sigma)$\,; $\phi \in \MSO(\Arch,\Sigma)$; $k>0$\\[0.5ex]
{Question:} & \hspace{-4em}$L(\Sys) \cap \kstwCBM \subseteq L(\phi)$\,?\\
\bottomrule
\end{tabular}
\end{center}

In the following, we will prove that:

\begin{theorem}
All these problems are decidable.
\end{theorem}

The proof technique is via an interpretation of $\kstwCBM$ in binary trees and
reduction to problems on tree automata.  This is actually similar to
decidability for cographs as explained in Section~\ref{sec:cographs}.

In the following, we introduce
\begin{itemize}
\item $\Akstwcbm$: a tree automaton accepting \kSTTs denoting graphs in 
$\kstwCBM$.

\item $\Akstwsys$: a tree automaton for $\Sys\in\CPDS(\Arch,\Sigma)$ such that
for all $\tau\in\Lang(\Akstwcbm)$, $\tau$ is accepted by $\Akstwsys$ iff
$G_\tau \in\Lang(\Sys)$.

\item $\Akstwphi$: a tree automaton for $\phi\in\MSO(\Arch,\Sigma)$ such that
for all $\tau\in\Lang(\Akstwcbm)$, $\tau$ is accepted by $\Akstwphi$ iff
$G_\tau\models\varphi$.
\end{itemize}

Then, 
\begin{itemize}
  \item  for \textsc{stw-Satisfiability}$(\Arch,\Sigma)$ we check if 
  $\Lang(\Akstwcbm\cap\Akstwphi)=\emptyset$,

  \item  for \textsc{stw-ModelChecking}$(\Arch,\Sigma)$ we check if 
  $\Lang(\Akstwcbm\cap\Akstwsys\cap{\A_{\neg\varphi}^{k\text{-}\mathsf{stw}}})=\emptyset$,

  \item  etc.
\end{itemize}

\clearpage
\section{Special tree-width and Tree interpretation}\label{sec:stw-tree-interpretation}

We show now that a graph $G$ defined by an \STT $\tau$ 
can be interpreted in the binary tree $\tau$.  Notice that the vertices of $G$
are in bijection with the leaves of $\tau$.  The main difficulty is to interpret
the edge relations $E_\gamma$ of $G$ in the tree $\tau$.

Let us denote by $\Lambda^k$ the alphabet of \kSTTs (we do not include the 
rename operation since it is redundant):
$$
\Lambda^k=\{\sttunion, i, \add{i}{a}, \add{i,j}\gamma, \forget{i} \mid
i,j\in[k], a\in\Sigma, \gamma\in\Gamma\} \,.
$$
Clearly, the set of \kSTTs considered as binary trees over $\Lambda^k$ is a 
regular tree language.

\begin{lemma}\label{lem:sttvalid}
  There is a formula $\phiksttvalid\in\MSO(\Lambda^k,\da_0,\da_1)$
  which defines the set of \emph{valid} \kSTTs.
\end{lemma}

\begin{proof}
  First, the binary tree is correctly
  labeled: leaves should have labels in $[k]$, unary nodes should
  have labels in $\{\add{i}{a},\add{i,j}\gamma, \forget{i}\mid i,j\in[k], a\in\Sigma
  \text{ and } \gamma\in\Gamma\}$ and binary nodes should be labeled
  $\sttunion$.  Moreover, for the \kSTT to be valid, the children of a binary
  node should have disjoint sets of \emph{active} colors.  This can be expressed
  with
  $$
  \neg\exists x,x',y,y',z
    \bigvee_{\substack{i\in[k]}}
    i(x) \wedge i(y) \wedge \sttunion(z) \wedge
    z\da_0 x' \wedge \beta_{i}(x',x) \wedge z\da_1 y' \wedge \beta_{i}(y',y)
  $$
  where $\beta_{i}(u',u)$ is a macro stating that $u'$ is an ancestor of $u$ in
  the tree and that $\forget{i}$ does not occur in the tree between node $u'$ and
  node $u$.  This formula can be written in \MSO with a transitive closure. 
  Recall that $u'\da u = u'\da_0 u \vee u'\da_1 u$. Let $\varphi_i(u',u)=u'\da 
  u \wedge \neg\forget{i}(u')$. Then, $\beta_{i}(u',u)=\varphi_i^{*}(u',u)$.
\end{proof}

\begin{exercise}\label{exo:ksttvalid}
  Give a tree automaton \Aksttvalid
  with $2^{k+1}$ states that accepts the set of \emph{valid} \kSTTs.
\end{exercise}

\begin{proposition}[Backward translation]\label{prop:stt-mso-interpretation}
  Fix some $k>0$.
  For all sentences $\phi\in\MSO(\Sigma,\Gamma)$, there is a
  sentence $\widetilde{\phi}^k\in\MSO(\Lambda^k,\da_0,\da_1)$ of size 
  $\mathcal{O}(k^{2}|\varphi|)$ such
  that, for every valid \kSTT $\tau$ with $\sem{\tau}=(G,\chi)$, we have
  \[G\models\phi \textup{~~~iff~~~} \tau\models \widetilde{\phi}^k \,.\]
\end{proposition}

\begin{proof}
  We proceed by induction on $\phi$.  Hence we also have to deal with free
  variables.  We denote by $\Int$ an interpretation of variables to (sets of)
  vertices of $G$ which are identified with leaves of $\tau$.  Hence, we prove
  by induction that for all formulas $\phi\in\MSO(\Sigma,\Gamma)$,
  there is a formula
  $\widetilde{\phi}^k\in\MSO(\Lambda^k,\da_0,\da_1)$ such that,
  for all valid \kSTT $\tau$ with $\sem{\tau}=(G,\chi)$, and all interpretations
  $\Int$ in vertices of $G$, we have
  \[G\models_\Int\phi \textup{~~~iff~~~} \tau\models_\Int \widetilde{\phi}^k \,.\]

  The difficult case is to translate the edge relations. We define
  \begin{align*}
    \widetilde{x \mathrel{E_\gamma} y}^k &= \exists z
    \bigvee_{\substack{i,j\in[k]}}
    i(x) \wedge j(y) \wedge \add{i,j}\gamma(z) \wedge
    \beta_i(z,x) \wedge \beta_j(z,y)
    \\
    \widetilde{a(x)}^k &= \exists z
    \bigvee_{\substack{i\in[k]}}
    i(x) \wedge \add{i}a(z) \wedge \beta_i(z,x) \,.
  \end{align*}
  Note that $\widetilde{x \mathrel{E_\gamma} y}^k$ and $\widetilde{a(x)}^k$ are of size
  $\mathcal{O}(k^{2})$ and $\mathcal{O}(k)$.
  The other cases are easy:
  \begin{align*}
     \widetilde{\neg\phi}^k &= \neg\,\widetilde{\phi}^k
     &
     \widetilde{\phi_1 \vee \phi_2}^k &= \widetilde{\phi_1}^k \vee 
     \widetilde{\phi_2}^k
     \\
     \widetilde{\exists x\phi}^k &= \exists x\;(\widetilde{\phi}^k \wedge \mathsf{leaf}(x))
     & 
     \widetilde{x \in X}^k &= x \in X  
     \\
     \widetilde{\exists X\phi}^k &= \exists X\,( \widetilde{\phi}^k \wedge \forall x (x \in X \rightarrow \mathsf{leaf}(x)) )
     &
     \widetilde{x = y}^k &= (x = y)
  \end{align*}
This concludes the proof.
\end{proof}

\begin{corollary}
  The MSO theory of graphs of special tree-width at most $k$ is decidable.
\end{corollary}

\begin{proof}
  Let $\phi\in\MSO(\Sigma,\Gamma)$.  Then, $\phi$ is valid on graphs of special
  tree-width at most $k$ iff $\phiksttvalid \to \widetilde{\phi}^k$ is
  valid on binary trees.
  The result of the corollary follows, since MSO validity is decidable on binary
  trees.  Indeed, the problem can be reduced to tree-automata emptiness
  \cite{ThaWri68}.
  Note that, moreover, $\phi$ is satisfiable on graphs of special tree-width at
  most $k$ iff $\phiksttvalid \wedge \widetilde{\phi}^k$ is satisfiable on
  binary trees.
\end{proof}

\begin{exercise}
  Construct a tree automaton $\A^k_\gamma$ with $\mathcal{O}(k^2)$ states which
  accepts a valid \kSTT $\tau$ with two marked leaves $x$ and $y$ iff there is a
  $\gamma$-edge between $x$ and $y$ in the graph $\sem{\tau}$.
\end{exercise}

\textbf{Solution:} $\A^k_\gamma$ is a deterministic bottom-up tree automaton.
It keeps in its state a pair of colors $(i,j)\in([k]\cup\{\bot\})^2$.  At some
node $z$, the state is $(i,j)$ where $i=\bot$ if $x$ is not in the
subtree of $z$, otherwise $i$ is the color of $x$; same for $j$ and $y$.
The state is initialized at leaves.  It is updated at $\sttunion$-nodes.  The
automaton goes to an accepting state if it is in state $(i,j)$ when reading a
node labeled $\add{i,j}\gamma$.
It goes to a rejecting state
at a node $\forget{\ell}$ if it is in state $(i,j)$ with $\ell\in\{i,j\}$.

\begin{exercise}\label{exo:TWAedgeSTT}
  Construct a tree-walking automaton $\B^k_\gamma$ with $\mathcal{O}(k)$ states
  which runs on a valid \kSTT $\tau$ starting from a leaf (say $x$) and accepts when
  reaching a leaf (say $y$) such that there is a $\gamma$-edge between $x$ and
  $y$ in the graph $\sem{\tau}$.
\end{exercise}

\textbf{Solution:} First, walking up the tree, the automaton $\B^k_\gamma$ keeps
in its state the color of the leaf $x$.  It 
makes sure that the color is not forgotten, until it
reaches a node labeled $\add{i,j}\gamma$ where $i$ is the color in its state.
Then, it updates its state with color $j$ and it enters a second phase where it
walks down the tree (non-deterministically at $\sttunion$-nodes), 
making sure that the color is not forgotten, until it
reaches a leaf having the color that corresponds to its state.

\underline{Remark:} The automaton $\B^k_\gamma$ can be made deterministic if 
there is at most one $\gamma$ edge with source $x$. Indeed, walking up the tree 
is deterministic and we can search for the target leaf $y$ using a DFS.

\clearpage
\section{Decision procedures for \kstwCBM}  

Recall that the class $\CBM(\Arch,\Sigma)$ is $\MSO(\Arch,\Sigma)$-definable by
a sentence $\phicbm$ (Proposition~\ref{prop:mso-cbm}). Recall that for \CBMs: 
$\Gamma=\{\procrel\}\cup\DS$ and $\Sigma'=\Sigma\cup\Procs$.

\begin{corollary}\label{cor:stw-mso-sat}
  The problem \textsc{stw-Satisfiability} and \textsc{stw-Validity} are
  decidable and the complexity is non-elementary.
\end{corollary}

\begin{proof}
  Let $\varphi\in\MSO(\Arch,\Sigma)$ and let $k>0$.
  
  Using Proposition~\ref{prop:stt-mso-interpretation}, we construct the formulas
  $\widetilde{\varphi}^k,\widetilde{\phicbm}^k\in\MSO(\Lambda^k,\da_0,\da_1)$.
  Consider the formula
  $\phiksttvalid\in\MSO(\Lambda^k,\da_0,\da_1)$ given by
  Lemma~\ref{lem:sttvalid}.  Then, we define
  $\Phi=\phiksttvalid\wedge\widetilde{\phicbm}^k\wedge\widetilde{\varphi}^k
  \in\MSO(\Lambda^k,\da_0,\da_1)$.
  
  We prove that $\varphi$ is
  satisfiable over $\kstwCBM$ iff $\Phi$ is satisfiable over 
  $\Lambda^k$-labeled binary trees.

  Indeed, let $\tau$ be a binary tree over $\Lambda^k$ such that
  $\tau\models\Phi$. Then, $\tau$ is a \kSTT and $\sem{\tau}=(G,\chi)$ where $G$ 
  is a \CBM and $G\models\varphi$. 
  Conversely, Let $G\in\kstwCBM$ be such that $G\models\varphi$. Let $\tau$ be 
  any \kSTT with $\sem{\tau}=(G,\chi)$. We have $\tau\models\Phi$.
  
  From \cite{ThaWri68}, we can construct a tree automaton $\A_\Phi$ equivalent to $\Phi$ 
  and then check $\A_\Phi$ for emptiness. Recall that emptiness for tree 
  automata can be solved in \textsc{PTime}, in fact linear time by a reduction 
  to Horn-Sat.
  
  The decidability of \textsc{stw-Validity} follows since validity of $\varphi$ 
  over $\kstwCBM$ is equivalent to non-satifiability of $\neg\varphi$ over 
  $\kstwCBM$.
\end{proof}

\begin{corollary}\label{cor:stw-emptiness}
  The problem \textsc{stw-Nonemptiness} is decidable and  
  \textsc{ExpTime}-complete.
\end{corollary}

\begin{proof}
  Let $\Sys\in\CPDS(\Arch,\Sigma)$ and let $k>0$.
  Let $\phisys\in\MSO(\Arch,\Sigma)$ be the equivalent formula given by 
  Theorem~\ref{thm:CPDStoMSO}. Then, we have
  $L(\Sys) \cap \kstwCBM \neq \emptyset$ iff $\phisys$ is satisfiable over 
  $\kstwCBM$. We conclude using Corollary~\ref{cor:stw-mso-sat}.

  We will prove the \textsc{ExpTime} upper-bound later using direct 
  constructions of tree automata.
\end{proof}

\begin{corollary}\label{cor:stw-mso-mc}
  The problem \textsc{stw-ModelChecking} is decidable and the complexity is 
  non-elementary.
\end{corollary}

\begin{proof}
  This is a consequence of Proposition~\ref{prop:SAT-MC-reduction}.
  Let $\Sys\in\CPDS(\Arch,\Sigma)$ be the system, let 
  $\varphi\in\MSO(\Arch,\Sigma)$ be the specification and let $k>0$.
  We consider as above the formula $\phisys\in\MSO(\Arch,\Sigma)$.
  The problem reduces to the validity of $\neg\phisys\vee\varphi$ over 
  $\kstwCBM$. We apply Corollary~\ref{cor:stw-mso-sat}.
\end{proof}

\begin{exercise}
  Prove  \textsc{stw-Inclusion} and \textsc{stw-Universality} decidable.
\end{exercise}

\begin{gpicture}[ignore, name = split-cbm]
  \drawoval[fillcolor=green!20!white](14.5,4,46,16,2)
  \gasset{Nw=2, Nh = 3, Nfill = n, Nframe = n}
  \node(name)(-12,4){$\mscn$}
  \node(p)(-6,7.5){$p$}
  \node(q)(-6,-0.5){$q$}
  \node(p1)(21,8){$a$}
  \node(p2)(28,8){$a$}
  \node(p3)(0,8){$b$}
  \node(p4)(7,8){$c$}
  \node(p5)(14,8){$d$}
  \node(q1)(14,0){$b$}
  \node(q2)(0,0){$a$}
  \node(q3)(21,0){$c$}
  \node(q4)(28,0){$d$}
  \node(q5)(35,0){$c$}
  \drawedge(p1,p2){}
  \drawedge(p3,p4){}\drawedge(p4,p5){}
  \drawedge(q3,q4){}\drawedge(q4,q5){}
  \drawedge(p1,q3){}
  \drawedge(p2,q5){}
  \drawedge(q2,p4){}
  \drawedge[curvedepth = 3](p3,p5){}
  \drawedge[curvedepth = -3, ELside=r](q1,q4){}
\end{gpicture}
\begin{gpicture}[ignore, name = split-cbm1]
  \drawoval[fillcolor=green!20!white](14.5,4,46,16,2)
  \gasset{Nw=2, Nh = 3, Nfill = n, Nframe = n}
  \node(name)(-12,4){$\mscn_1$}
  \node(p)(-6,7.5){$p$}
  \node(q)(-6,-0.5){$q$}
  \node(p1)(21,8){$a$}
  \node(p2)(28,8){$a$}
  \node(p3)(0,8){$b$}
  \node(p4)(7,8){$c$}
  \node(p5)(14,8){$d$}
  \node(q1)(14,0){$b$}
  \node(q2)(0,0){$a$}
  \node(q3)(21,0){$c$}
  \node(q4)(28,0){$d$}
  \node(q5)(35,0){$c$}
  \drawedge(p1,p2){}
  \drawedge[linecolor=red,dash={1}0](p5,p1){}
  \drawedge(p3,p4){}\drawedge(p4,p5){}
  \drawedge[linecolor=red,dash={1}0](q2,q1){}
  \drawedge[linecolor=red,dash={1}0](q1,q3){}
  \drawedge(q3,q4){}\drawedge(q4,q5){}
  \drawedge(p1,q3){}
  \drawedge(p2,q5){}
  \drawedge(q2,p4){}
  \drawedge[curvedepth = 3](p3,p5){}
  \drawedge[curvedepth = -3, ELside=r](q1,q4){}
\end{gpicture}
\begin{gpicture}[ignore, name = split-cbm2]
  \drawoval[fillcolor=green!20!white](14.5,4,46,16,2)
  \gasset{Nw=2, Nh = 3, Nfill = n, Nframe = n}
  \node(name)(-12,4){$\mscn_2$}
  \node(p)(-6,7.5){$p$}
  \node(q)(-6,-0.5){$q$}
  \node(p1)(21,8){$a$}
  \node(p2)(28,8){$a$}
  \node(p3)(0,8){$b$}
  \node(p4)(7,8){$c$}
  \node(p5)(14,8){$d$}
  \node(q1)(0,0){$b$}
  \node(q2)(7,0){$a$}
  \node(q3)(21,0){$c$}
  \node(q4)(28,0){$d$}
  \node(q5)(35,0){$c$}
  \drawedge(p1,p2){}
  \drawedge[linecolor=red,dash={1}0](p5,p1){}
  \drawedge(p3,p4){}\drawedge(p4,p5){}
  \drawedge[linecolor=red,dash={1}0](q1,q2){}
  \drawedge[linecolor=red,dash={1}0](q2,q3){}
  \drawedge(q3,q4){}\drawedge(q4,q5){}
  \drawedge(p1,q3){}
  \drawedge(p2,q5){}
  \drawedge(q2,p4){}
  \drawedge[curvedepth = 3](p3,p5){}
  \drawedge[curvedepth = -3, ELside=r](q1,q4){}
\end{gpicture}
\begin{gpicture}[ignore, name = split-cbm3]
  \drawoval[fillcolor=green!20!white](14.5,4,46,16,2)
  \gasset{Nw=2, Nh = 3, Nfill = n, Nframe = n}
  \node(name)(-12,4){$\mscn_3$}
  \node(p)(-6,7.5){$p$}
  \node(q)(-6,-0.5){$q$}
  \node(p1)(0,8){$a$}
  \node(p2)(7,8){$a$}
  \node(p3)(14,8){$b$}
  \node(p4)(21,8){$c$}
  \node(p5)(28,8){$d$}
  \node(q1)(7,0){$b$}
  \node(q2)(0,0){$a$}
  \node(q3)(14,0){$c$}
  \node(q4)(21,0){$d$}
  \node(q5)(28,0){$c$}
  \drawedge(p1,p2){}
  \drawedge[linecolor=red,dash={1}0](p2,p3){}
  \drawedge(p3,p4){}\drawedge(p4,p5){}
  \drawedge[linecolor=red,dash={1}0](q2,q1){}
  \drawedge[linecolor=red,dash={1}0](q1,q3){}
  \drawedge(q3,q4){}\drawedge(q4,q5){}
  \drawedge(p1,q3){}
  \drawedge(p2,q5){}
  \drawedge(q2,p4){}
  \drawedge[curvedepth = 3](p3,p5){}
  \drawedge[curvedepth = -3, ELside=r](q1,q4){}
\end{gpicture}
\begin{gpicture}[ignore, name = split-cbm4]
  \drawoval[fillcolor=green!20!white](14.5,4,46,16,2)
  \gasset{Nw=2, Nh = 3, Nfill = n, Nframe = n}
  \node(name)(-12,4){$\mscn_4$}
  \node(p)(-6,7.5){$p$}
  \node(q)(-6,-0.5){$q$}
  \node(p1)(0,8){$a$}
  \node(p2)(7,8){$a$}
  \node(p3)(14,8){$b$}
  \node(p4)(21,8){$c$}
  \node(p5)(28,8){$d$}
  \node(q1)(0,0){$b$}
  \node(q2)(7,0){$a$}
  \node(q3)(14,0){$c$}
  \node(q4)(21,0){$d$}
  \node(q5)(28,0){$c$}
  \drawedge(p1,p2){}
  \drawedge[linecolor=red,dash={1}0](p2,p3){}
  \drawedge(p3,p4){}\drawedge(p4,p5){}
  \drawedge[linecolor=red,dash={1}0](q1,q2){}
  \drawedge[linecolor=red,dash={1}0](q2,q3){}
  \drawedge(q3,q4){}\drawedge(q4,q5){}
  \drawedge(p1,q3){}
  \drawedge(p2,q5){}
  \drawedge(q2,p4){}
  \drawedge[curvedepth = 3](p3,p5){}
  \drawedge[curvedepth = -3, ELside=r](q1,q4){}
\end{gpicture}
\begin{gpicture}[ignore, name = split-cbm5]
  \drawoval[fillcolor=green!20!white](14.5,4,46,16,2)
  \gasset{Nw=2, Nh = 3, Nfill = n, Nframe = n}
  \node(name)(-12,4){$\mscn_5$}
  \node(p)(-6,7.5){$p$}
  \node(q)(-6,-0.5){$q$}
  \node(p1)(0,8){$a$}
  \node(p2)(7,8){$a$}
  \node(p3)(21,8){$b$}
  \node(p4)(28,8){$c$}
  \node(p5)(35,8){$d$}
  \node(q1)(0,0){$b$}
  \node(q2)(28,0){$a$}
  \node(q3)(7,0){$c$}
  \node(q4)(14,0){$d$}
  \node(q5)(21,0){$c$}
  \drawedge(p1,p2){}
  \drawedge[linecolor=red,dash={1}0](p2,p3){}
  \drawedge(p3,p4){}\drawedge(p4,p5){}
  \drawedge[linecolor=red,dash={1}0](q1,q3){}
  \drawedge[linecolor=red,dash={1}0](q5,q2){}
  \drawedge(q3,q4){}\drawedge(q4,q5){}
  \drawedge(p1,q3){}
  \drawedge(p2,q5){}
  \drawedge(q2,p4){}
  \drawedge[curvedepth = 3](p3,p5){}
  \drawedge[curvedepth = -3, ELside=r](q1,q4){}
\end{gpicture}

\clearpage
\section{Tree automata for efficient decision procedures on \kstwCBM}  

Recall that for an \STT $\tau$ we write $\sem{\tau}=(G_\tau,\chi_\tau)$.
In this section, we only consider \STTs without useless operations: if 
$\add{i}{a}\tau$ or $\forget{i}\tau$ is a subterm then $i$ is active in $\tau$,
if $\add{i,j}{\gamma}\tau$ is a subterm then $i,j$ are active in $\tau$.

Not all \kSTTs $\tau$ define graphs $G_\tau$ which are \CBMs.  In order to use
the tree interpretation in \STTs to efficiently solve problems on \CBMs of
bounded special tree-width, we will construct a tree automaton \Akstwcbm which 
accepts \kSTTs denoting \CBMs (Proposition~\ref{prop:Akstwcbm}).

As a first warm-up, we construct a tree automaton checking that the input 
binary tree $\tau$ is a (valid) \kSTT and that the graph $G_\tau$ is acyclic.

\begin{proposition}\label{prop:STT-acyclic}
  There is a deterministic bottom-up tree automaton \Akstwacyclic of size
  $2^{\mathcal{O}(k^2)}$ which accepts all binary trees $\tau$ such that $\tau$
  is a \kSTT and $G_\tau$ is acyclic.
\end{proposition}

\begin{proof}
  A state of \Akstwacyclic is a pair $(P,\prec)$ where $P\subseteq[k]$ and
  $\prec$ is a strict order on $P$.  When reading an \STT $\tau$ with
  $\sem{\tau}=(G,\chi)$ the automaton will reach the state $(P,\prec)$
  satisfying the following two conditions:
  \begin{enumerate}[nosep,label=($\mathsf{A}_{\arabic*}$),ref=$\mathsf{A}_{\arabic*}$]
    \item\label{item:A1} $P=\dom(\chi)\subseteq[k]$ is the set of active 
    colors, 

    \item\label{item:A2} $G$ is acyclic and $\prec$ is the restriction induced by $E^{+}$
    on $P$: $i\prec j$ iff $\chi(i)\mathrel{E}^{+}\chi(j)$ where
    $E=\bigcup_{\gamma\in\Gamma}E_\gamma$.
  \end{enumerate}
  The transitions are defined below (with $s=(P,\prec)$, $s'=(P',\prec')$, 
  $s''=(P'',\prec'')$).
  
  \begin{tabular}{|c|p{110mm}|}
    \hline
    $\bot \xrightarrow{i}s$ 
    &
    if $P=\{i\}$ and ${\prec}=\emptyset$.
    \\ \hline
    $s'\xrightarrow{\add{i}{a}}s$ 
    &
    if $i\in P'$ and $s=s'$.
    \\ \hline
    $s'\xrightarrow{\add{i,j}{\gamma}}s$ 
    &
    if $i,j\in P'$, $i\neq j$ and $\neg(j\prec' i)$.\par
    Then, $P=P'$ and ${\prec}=({\prec}'\cup\{(i,j)\})^{+}$.
    \\ \hline
    $s'\xrightarrow{\forget{i}}s$ 
    &
    if $i\in P'$. Then $P=P'\setminus\{i\}$ and 
    ${\prec}={\prec}'\cap(P\times P)$.
    \\ \hline
    $s',s''\xrightarrow{\oplus}s$
    &
    if $P'\cap P''=\emptyset$ (active colors should be disjoint).
    
    Then, $P=P'\uplus P''$ and ${\prec}={\prec}'\cup{\prec}''$.
    \\ \hline
  \end{tabular}
  
  We can easily check that if \Akstwacyclic has a run on a binary tree $\tau$
  then $\tau$ is a \kSTT and $G_\tau$ is acyclic. The number of states of 
  \Akstwacyclic is at most $2^{k+1+(k+1)^{2}}$.
\end{proof}

In the following, we concentrate on the signature for \CBMs: 
$\Gamma=\{\procrel\}\cup\DS$ and $\Sigma'=\Sigma\cup\Procs$. Moreover, since 
every event should be labeled with exactly one letter from $\Sigma$ and one 
process from $\Procs$, we use $(i,a,p)=\add{i}{p}\add{i}{a}i$ as atomic \STTs 
instead of $i$. So we consider terms defined with the syntax:
$$
\tau ::= (i,a,p) \mid \add{i,j}{\gamma}\tau \mid \forget{i}\tau \mid \tau\oplus\tau
$$
where $i\in[k]$, $a\in\Sigma$, $p\in\Procs$ and 
$\gamma\in\Gamma=\{\procrel\}\cup\DS$.

The second warm-up is a tree automaton checking the local conditions of edges 
in a \CBM (see Definition~\ref{def:cbm}).

\begin{proposition}\label{prop:STT-edges-local}
  There is a deterministic bottom-up tree automaton \Akstwedges of size
  $2^{\mathcal{O}(k)}\cdot|\Procs|^{\mathcal{O}(k)}$ which accepts all binary trees $\tau$ such that $\tau$
  is a \kSTT and $G_\tau$ satisfies the following conditions:
  \begin{enumerate}[nosep]
    \item  process edges are not branching and are between events of the same 
    process,
  
    \item  data edges are disjoint and respect the $\writer$/$\reader$ 
    constraints of data structures.
  \end{enumerate}
\end{proposition}

\begin{proof}
  A state of \Akstwedges is a tuple $s=(P,\pi,\alpha,\beta,\gamma)$ where 
  $P\subseteq[k]$, $\alpha,\beta,\gamma\subseteq P$ and $\pi\colon P\to\Procs$.
  When reading an \STT $\tau$ with $\sem{\tau}=(G,\chi)$, the automaton 
  will reach the state $s$ satisfying the following two conditions:
  \begin{enumerate}[nosep,label=($\mathsf{B}_{\arabic*}$),ref=$\mathsf{B}_{\arabic*}$]
    \item\label{item:B1} $P=\dom(\chi)\subseteq[k]$ is the set of active 
    colors, 

    \item\label{item:B2} $\pi\colon P\to\Procs$ gives the process of the 
    colored event: $\pi(i)=\pid(\chi(i))$,

    \item\label{item:B3} if $e\procrel f$ in $G$ then $\pid(e)=\pid(f)$ and for 
    all $i\in P$ we have $i\in\alpha$ iff $\chi(i)$ is the source of some 
    $\procrel$-edge, and $i\in\beta$ iff $\chi(i)$ is the target of some 
    $\procrel$-edge, 

    \item\label{item:B4} if $e\matchrel^{d} f$ in $G$ then 
    $\pid(e)=\writer(d)$, $\pid(f)=\reader(d)$ and for 
    all $i\in P$ we have $i\in\gamma$ iff $\chi(i)$ is the source or target of some 
    $\matchrel$-edge.
  \end{enumerate}
  The transitions are defined below (with $s=(P,\pi,\alpha,\beta,\gamma)$,
  $s'=(P',\pi',\alpha',\beta',\gamma')$,
  $s''=(P'',\pi'',\alpha'',\beta'',\gamma'')$).
  
  \begin{tabular}{|c|p{110mm}|}
    \hline
    $\bot \xrightarrow{(i,a,p)}s$ 
    &
    if $P=\{i\}$, $\pi(i)=p$ and $\alpha=\beta=\gamma=\emptyset$.
    \\ \hline
    $s'\xrightarrow{\add{i,j}{\procrel}}s$ 
    &
    if $i,j\in P'$, $i\neq j$, $\pi'(i)=\pi'(j)$, 
    $i\notin\alpha'$ and $j\notin\beta'$.\par
    Then, $P=P'$, $\pi=\pi'$, $\gamma=\gamma'$, $\alpha=\alpha'\cup\{i\}$ and 
    $\beta=\beta'\cup\{j\}$.
    \\ \hline
    $s'\xrightarrow{\add{i,j}{d}}s$ 
    &
    if $i,j\in P'$, $i\neq j$, $i,j\notin\gamma'$, 
    $\pi'(i)=\writer(d)$ and $\pi'(j)=\reader(d)$.
    Then, $P=P'$, $\pi=\pi'$, $\alpha=\alpha'$, $\beta=\beta'$ and 
    $\gamma=\gamma'\cup\{i,j\}$.
    \\ \hline
    $s'\xrightarrow{\forget{i}}s$ 
    &
    if $i\in P'$. Then $P=P'\setminus\{i\}$ and 
    $\pi,\alpha,\beta,\gamma$ are the restrictions of 
    $\pi',\alpha',\beta',\gamma'$ to $P$.
    \\ \hline
    $s',s''\xrightarrow{\oplus}s$
    &
    if $P'\cap P''=\emptyset$ (active colors should be disjoint).
    
    Then, $s$ is the disjoint union of $s'$ and $s''$.
    \\ \hline
  \end{tabular}
  
  We can easily check that if \Akstwedges has a run on a binary tree $\tau$
  then $\tau$ is a \kSTT and $G_\tau$ satisfies the local conditions of 
  Definition~\ref{def:cbm}. The number of states of 
  \Akstwedges is at most $2^{4(k+1)}\cdot|\Procs|^{k+1}$.
  
  Remark: some terms with redundant operations, such as
  $\add{i,j}\procrel \add{i,j}\procrel ((i,a,p)\oplus(j,b,p))$ are not 
  accepted. This can be fixed easily.
\end{proof}

We turn now to the full automaton checking that an \STT defines a \CBM. We start
with a definition.  A \emph{split}-\CBM is a CBM in which behaviors of processes
may be split in several factors.

\begin{definition}
  A graph
  $\smscn=(\Events,\procrel,(\matchrel^d)_{d\in\DS},\pid,\lambda)$ is a 
  \emph{split}-\CBM if it is possible to obtain a CBM
  $(\Events,{\procrel}\uplus{\elastic},(\matchrel^d)_{d\in\DS},\pid,\lambda)$ 
  by adding some missing process edges ${\elastic}\in\Events^2\setminus{\procrel}$.
  A \emph{factor} (or \emph{block}) of $\smscn$ is a maximal sequence
  of events connected by process edges.
\end{definition}

\begin{example}
  The split-\CBM $\mscn$ depicted below has 5 factors: 2 on process $p$ and 3 on
  process $q$.  It has two connected components.  There are 5 ways to add the
  missing process edges in order to get a \CBM: $\mscn_1$,\ldots,$\mscn_5$.
  \begin{center}
    \medskip
    \gusepicture{split-cbm}
    \hfil
    \gusepicture{split-cbm1}
    \\[5mm]
    \gusepicture{split-cbm2}
    \hfil
    \gusepicture{split-cbm3}
    \\[5mm]
    \gusepicture{split-cbm4}
    \hfil
    \gusepicture{split-cbm5}
  \end{center}
\end{example}

\begin{proposition}\label{prop:Akstwcbm}
  There is a tree automaton \Akstwcbm of size $2^{\mathcal{O}(k^2|\Arch|)}$ 
  which accepts all binary trees $\tau$ such that $\tau$ is a \kSTT and 
  $\sem{\tau}=(G_\tau,\chi_\tau)$ is an uncolored \CBM: $G_\tau$ is a \CBM and 
  $\dom(\chi_\tau)=\emptyset$.
\end{proposition}

Actually, the automaton \Akstwcbm admits a run (not necessarily accepting) on
a binary tree $\tau$ iff $\tau$ is a \kSTT and $\sem{\tau}=(G_\tau,\chi_\tau)$
where $G_\tau$ is a split-\CBM.

\newcommand{\precdot}{\mathrel{\prec\!\!\!\cdot}}

\begin{proof}
  For simplicity, we construct the automaton assuming that all data
  structures are bags. In case of stacks or queues, we could enforce the LIFO or
  FIFO properties with an intersection with a further automaton (see Exercises
  below).
    
  The automaton \Akstwcbm is nondeterministic and computes an abstraction of the
  split-\CBM defined by a given term.  More precisely, when reading bottom-up a
  binary tree $\tau$, the automaton \Akstwcbm checks that $\tau$ is a \kSTT and
  reaches a state
  $s=(P,\pi,\alpha,\beta,\gamma,\prec,L,R)$ satisfying the
  following conditions with $\sem{\tau}=(G,\chi)$:
  \begin{enumerate}[nosep,label=($\mathsf{I}_{\arabic*}$),ref=$\mathsf{I}_{\arabic*}$]

    \item\label{item:I1} $P=\dom(\chi)\subseteq[k]$ is the set of active 
    colors, 
  
    \item\label{item:I2} $\pi\colon P\to\Procs$ gives the associated processes:
    $\pi(i)=\pid(\chi(i))$ for all $i\in P$,
    
    \item\label{item:I3} $\alpha,\beta,\gamma\subseteq P$ are such that 
    \begin{itemize}[nosep]
      \item $i\in\alpha$ iff $\chi(i)$ is the source of a
      $\procrel$-edge in $G$.
    
      \item $i\in\beta$ iff $\chi(i)$ is the target of a
      $\procrel$-edge in $G$.
    
      \item $i\in\gamma$ iff $\chi(i)$ is the source or target of a
      $\matchrel$-edge in $G$.
    \end{itemize}
  
    \item\label{item:I4} ${\prec}\subseteq P^2$ is a strict partial order such that for all
    $i,j\in P$, 
    \begin{enumerate}[nosep]
      \item\label{item:I4a} $\chi(i)\mathrel{({\procrel}\cup{\matchrel})}^+\chi(j)$ in
      $G$ implies $i\prec j$,
    
      \item\label{item:I4b} $\pi(i)=\pi(j)$ and $i\neq j$ implies $i\prec j$ or $j\prec i$.
      
      Hence, for each $p\in\Procs$, $\prec$ defines a total order on 
      $\pi^{-1}(p)\subseteq P$. 
      
      We denote by $\precdot_p$ the successor relation of this total order.

      \item\label{item:I4c} If $i\precdot_p j$ then $i\in\alpha$ iff $j\in\beta$,
      
      \item\label{item:I4d} If $i\precdot_p j \wedge i\in\alpha$ then $\chi(i)\procrel^+\chi(j)$,
      
    \end{enumerate}
    The partial order $\prec$ is guessed by \Akstwcbm so that it will eventually
    correspond to the order of the final \CBM defined by the global term.  So
    $\prec$ must be compatible with all $\procrel$ and $\matchrel$ edges already
    added in the subterm $\tau$ \eqref{item:I4a} and for each process the final
    ordering has been guessed already \eqref{item:I4b}, even though some
    $\procrel$-edges may still be missing in $G$.
    
    Together with $\alpha$ and $\beta$, the partial order $\prec$ allows to locate the 
    \emph{holes} between consecutive factors of processes. Formally, the hole relation 
    is defined by $i\leadsto j$ if
    $i\notin\alpha \wedge i \precdot_p j$ for some 
    $p\in\Procs$.

    \item\label{item:I5} $\sem{\tau}=(G,\chi)$ is a split-\CBM with the additional 
    process edges defined by $\chi(i)\elastic\chi(j)$ iff $i\leadsto j$.
    
    {$L,R\subseteq\Procs$ give the processes whose
    minimal/maximal event in the split-\CBM $(\sem{\tau},\elastic)$ is no more
    colored.}
    More precisely, consider the sequence $w_1,\ldots,w_n$ of factors of some
    process $p\in\Procs$.  Let $e_1,f_1,\ldots,e_n,f_n$ be the endpoints of
    these factors.  We have
    $e_1\procrel^*f_1\elastic e_2\procrel^*f_2\elastic\cdots e_n\procrel^*f_n$.
    \\
    By definition of $\elastic$, the events $f_1,e_2,\ldots,e_n$ must be 
    colored. \\
    Now, we have $p\in L$ iff $e_1$ is not colored, and $p\in R$ iff 
    $f_n$ is not colored.
    
  \end{enumerate}

  The number of states of \Akstwcbm is at most
  $2^{4(k+1)}\cdot|\Procs|^{k+1}\cdot 2^{(k+1)^2}\cdot{2^{2|\Procs|}}$.
  
  \begin{table}[!p]
    \noindent\hspace{-3mm}
    \begin{tabular}{|c|p{120mm}|}
      \hline
      $\bot \xrightarrow{(i,a,p)}s$ 
      &
      if $P=\{i\}$, $\pi(i)=p$, 
      $\alpha=\beta=\gamma=\emptyset$,
      ${\prec}=\emptyset$ and $L=R=\emptyset$.
      \\ \hline
      $s\xrightarrow{\add{i,j}{\procrel}}s'$ 
      &
      if $i,j\in P$ and $i\leadsto j$.
      Then, $P'=P$, $\pi'=\pi$, $\alpha'=\alpha\cup\{i\}$,
      $\beta'=\beta\cup\{j\}$, $\gamma'=\gamma$, ${\prec}'={\prec}$, $L'=L$ and
      $R'=R$. 
      \\ \hline
      $s\xrightarrow{\add{i,j}{d}}s'$ 
      &
      if $i,j\in P$, $i\prec j$, $\pi(i)=\writer(d)$, 
      $\pi(j)=\reader(d)$, $i,j\notin\gamma$.

      Then, $P'=P$, $\pi'=\pi$, $\alpha'=\alpha$, $\beta'=\beta$,
      $\gamma'=\gamma\cup\{i,j\}$, ${\prec}'={\prec}$, $L'=L$ and $R'=R$.
      \\ \hline
      $s\xrightarrow{\forget{i}}s'$ 
      &
      if $i\in P$ and \par
      $i\in\alpha \vee (\pi(i)\notin R \wedge \forall k,~\pi(k)=\pi(i)\implies k\preceq i)$
      and \par
      $i\in\beta \vee (\pi(i)\notin L \wedge \forall k,~\pi(k)=\pi(i)\implies i\preceq k)$.
      
      For color $i$ to be forgotten, the corresponding event $e$ should already
      be the source of a $\procrel$-edge ($i\in\alpha)$ or it should be
      maximal on its process, and symmetrically, it should be the target of a
      $\procrel$-edge or it should be minimal on its process.  
     
      Then, $P'=P\setminus\{i\}$ and \par
      $\pi',\alpha',\beta',\gamma',{\prec}'$ are the
      restrictions of $\pi,\alpha,\beta,\gamma,{\prec}$ to $P'$, and \par
      $L'=L$ if $i\in\beta$ and $L'=L\cup\{\pi(i)\}$ otherwise, and \par
      $R'=R$ if $i\in\alpha$ and $R'=R\cup\{\pi(i)\}$ otherwise.
      \\ \hline
      $s',s''\xrightarrow{\oplus}s$
      &
      if $P'\cap P''=\emptyset$ (active colors should be disjoint),
      $L'\cap L''=\emptyset$ (the minimal event of some process cannot belong 
      to both subterms) and $R'\cap R''=\emptyset$. 
      
      Then, $s$ is the disjoint union of
      $s'$ and $s''$: $P=P'\uplus P''$, $\pi=\pi'\uplus\pi''$,
      $\alpha=\alpha'\uplus\alpha''$, $\beta=\beta'\uplus\beta''$,
      $\gamma=\gamma'\uplus\gamma''$, $L=L'\uplus L''$, $R=R'\uplus R''$
      and

      $\prec$ is a (guessed) strict partial order on $P$ satisfying
      \eqref{item:I4b} and
      
      $\bullet$ additional ordering is guessed only between colored points of
      the left and right subterms:
      ${\prec}'={\prec}\cap(P'\times P')$ and 
      ${\prec}''={\prec}\cap(P''\times P'')$,

      $\bullet$ the new sequence of factors on some process $p$ is obtained by
      shuffling the sequences of factors of $p$ of the subterms:
      for all $i,j\in P$ and $p\in\Procs$, 
      if $i\in\alpha$ or $j\in\beta$ then ($i\precdot_p j$ iff
      $i\precdot'_p j$ or $i\precdot''_p j$),
      
      $\bullet$ if for some process $p$, the minimal event of the global \CBM
      occurs in the right subterm and its color has been forgotten ($p\in L''$),
      then we cannot insert a $p$-factor of the left subterm before the first
      $p$-factor of the right subterm (and similarly for the other cases):
      
      for all $i\in P'$, if $\pi(i)\in L''$ then $j\prec i$ for some $j\in P''$ 
      with $\pi(j)=\pi(i)$,
      \par
      for all $i\in P''$, if $\pi(i)\in L'$ then $j\prec i$ for some $j\in P'$ 
      with $\pi(j)=\pi(i)$,
      \par
      for all $i\in P'$, if $\pi(i)\in R''$ then $i\prec j$ for some $j\in P''$ 
      with $\pi(j)=\pi(i)$,
      \par
      for all $i\in P''$, if $\pi(i)\in R'$ then $i\prec j$ for some $j\in P'$ 
      with $\pi(j)=\pi(i)$.
      \\ \hline
    \end{tabular}
    
    \caption{Transitions of \Akstwcbm where $a\in\Sigma$, $p\in\Procs$, 
    $d\in\DS$, $i,j\in[k]$, and states $s=(P,\pi,\alpha,\beta,\gamma,\prec,L,R)$,
    $s'=(P',\pi',\alpha',\beta',\gamma',\prec',L',R')$ and 
    $s''=(P'',\pi'',\alpha'',\beta'',\gamma'',\prec'',L'',R'')$.
    }
    \protect\label{tab:Akstwcbm}
  \end{table}

  The \emph{transitions} of \Akstwcbm are defined in Table~\ref{tab:Akstwcbm}. 
  We check inductively that the invariants (\ref{item:I1}--\ref{item:I5}) are 
  preserved by the transitions of \Akstwcbm. This is clear at the leaves. Let 
  us go through the other cases.
  \begin{itemize}
    \item $\add{i,j}\procrel$. Clearly, (\ref{item:I1}--\ref{item:I3}) are 
    preserved. Next, \eqref{item:I4a} is also preserved since the 
    edge $\chi(i)\procrel\chi(j)$ is only added when $i\leadsto j$, which 
    implies $i\prec j$. Items  (\ref{item:I4b}--\ref{item:I4d}) are 
    trivially preserved. Finally, \eqref{item:I5} is also preserved since the 
    effect of $\add{i,j}\procrel$ is to turn the $\elastic$-edge from
    $\chi(i)$ to $\chi(j)$ into a $\procrel$-edge.
  
    \item $\add{i,j}d$.  As above, (\ref{item:I1}--\ref{item:I4}) are trivially
    preserved.  For \eqref{item:I5}, notice that the effect of $\add{i,j}d$ is
    to add a $\matchrel^d$-edge from $\chi(i)$ to $\chi(j)$.  The resulting
    graph is still a split-\CBM since the transition is only allowed when
    $i\prec j$, $\pi(i)=\writer(d)$, $\pi(j)=\reader(d)$ and
    $i,j\notin\gamma$.
  
    \item $\forget{i}$.  We can also easily check that
    (\ref{item:I1}--\ref{item:I5}) are preserved.  The only non-trivial cases
    are \eqref{item:I4c} and \eqref{item:I4d}.  Assume that $j\precdot'_p \ell$.
    Either $j\precdot_p \ell$ and we can conclude easily.  Or $j\precdot_p
    i\precdot_p \ell$ and the conditions of the $\forget{i}$-transition imply
    that $i\in\alpha$ and $i\in\beta$.  By induction we obtain
    $\ell\in\beta$, $j\in\alpha$ and
    $\chi(j)\procrel^+\chi(i)\procrel^+\chi(\ell)$.  Hence, \eqref{item:I4c} and
    \eqref{item:I4d} hold.
    
    \item $\oplus$ is the most difficult case. Assume that the transition 
    $s',s''\xrightarrow{\oplus}s$ is applied at the root of a term 
    $\tau=\tau'\oplus\tau''$ and that the invariants 
    (\ref{item:I1}--\ref{item:I5}) hold for $(s',{\tau'})$ and 
    $(s'',{\tau''})$. Since $P'\cap P''=\emptyset$, 
    \eqref{item:I1} implies that $\dom(\chi_{\tau'})$ and $\dom(\chi_{\tau''})$ 
    are disjoint and $\tau$ is a legal \kSTT. It is easy to check that  
    (\ref{item:I1}--\ref{item:I3}) hold for $(s,\tau)$.
     
    We turn now to \eqref{item:I4}.  By definition of the $\oplus$-transition,
    $\prec$ is a strict partial order on $P$ satisfying \eqref{item:I4b}.  Now,
    $\chi(i)\mathrel{({\procrel}\cup{\matchrel})}^+\chi(j)$ in
    $G_\tau$ iff
    $\chi(i)\mathrel{({\procrel}\cup{\matchrel})}^+\chi(j)$ in
    $G_{\tau'}$ or in $G_{\tau''}$, which implies $i\prec' j$ or $i\prec'' j$,
    and finally $i\prec j$. Hence, \eqref{item:I4a} holds.
    Assume that $i\precdot_p j$.  If $i\in\alpha$ or $j\in\beta$ then
    $i\precdot'_p j$ or $i\precdot''_p j$ by definition of the
    $\oplus$-transition.  We deduce that $j\in\alpha\cap\beta$ and
    $\chi(i)\procrel^+\chi(j)$.  The other case is $i\notin\alpha$ and
    $j\notin\beta$.  Hence, \eqref{item:I4c} and \eqref{item:I4d} hold.
    Hence, \eqref{item:I4} holds for $(s,\tau)$.
    
    Next, we check that \eqref{item:I5} holds for the pair $(s,\tau)$.  Let
    $\sem{\tau}=(G,\chi)$ with
    $G=(\Events,\procrel,(\matchrel^d)_{d\in\DS},\pid,\lambda)$.  We have to
    prove that $\sem{\tau}$ is a split-\CBM with additional process edges
    defined by ${\elastic}=\{(\chi(i),\chi(j))\mid i,j\in P \wedge i\leadsto j\}$.
    
    Let $p\in\Procs$.  The set of $p$-events is  
    $\Events_p=\Events'_p\uplus\Events''_p$. If
    $\Events''_p=\emptyset$ then $\elastic$ coincide with $\elastic'$ on
    $\Events_p=\Events'_p$ and $(\Events_p,{\procrel}\cup{\elastic},\lambda)$ is
    indeed a word.  The same holds when $\Events'_p=\emptyset$.
    We assume now that $\Events'_p\neq\emptyset\neq\Events''_p$.
    {We can prove that each $p$-factor, i.e., each
    $\procrel$-connected component of $\Events_p$, has at least one endpoint
    colored.}
    Since $\prec$ is a total order on $\pi^{-1}(p)$, it
    induces a total order on the factors: $w_1,\ldots,w_n$.  Let 
    $e_1,f_1,\ldots,e_n,f_n$
    be the left and right endpoints of $w_1,\ldots,w_n$.  Let
    $i_1,j_1,\ldots,i_n,j_n$ be the colors of $e_1,f_1,\ldots,e_n,f_n$.  We have
    $j_\ell\precdot_p i_{\ell+1}$ and $j_\ell\notin\alpha$ for $1\leq\ell<n$.
    Therefore, $j_\ell\leadsto i_{\ell+1}$ and $f_\ell\elastic e_{\ell+1}$.
    We deduce that $(\Events_p,{\procrel}\cup{\elastic},\lambda)$ is the word 
    $w_1w_2\cdots w_n$.
    Also, $p\in L'$ iff $e_1\in\Events'_p$ and $e_1$ is not colored. Similarly,
    $p\in L''$ iff $e_1\in\Events''_p$ and $e_1$ is not colored. Therefore,
    $p\in L=L'\cup L''$ iff $e_1$ is not colored. Similarly,
    $p\in P=R'\cup R''$ iff $f_n$ is not colored. 
        
    We prove now that the relation $R={\procrel}\cup{\elastic}\cup{\matchrel}$
    is acyclic.  Notice that $e\procrel f$ in $G_\tau$ iff $e\procrel f$ in
    either $G_{\tau'}$ or $G_{\tau''}$.  The same holds for $\matchrel$.
    Moreover, the relation ${\procrel}\cup{\matchrel}$ is acyclic both in
    $G_{\tau'}$ and in $G_{\tau''}$.  Hence, if the relation $R$ admits a cycle
    in $G_\tau$, it must use some $\elastic$-edges:
    $$
    e_1\elastic f_1({\procrel}\cup{\matchrel})^*e_2\elastic f_2 \cdots 
    e_n\elastic f_n({\procrel}\cup{\matchrel})^*e_1 \,.
    $$
    By definition of $\elastic$ we have $e_1,f_1,\ldots,e_n,f_n\in\chi(P)$.
    Let $i_1,j_1,\ldots,i_n,j_n\in P$ such that $\chi(i_\ell)=e_\ell$ and 
    $\chi(j_\ell)=f_\ell$ for $1\leq\ell\leq n$. By definition of $\elastic$ 
    and using \eqref{item:I4a} we deduce that
    $i_1\prec j_1\preceq i_2\prec j_2 \cdots i_n\prec j_n\preceq i_1$,
    a contradiction with $\prec$ acyclic.
    
    This concludes the proof that \eqref{item:I5} holds for $(s,\tau)$.
  \end{itemize}
  
  A state $s=(P,\pi,\alpha,\beta,\gamma,\prec,L,R)$ of \Akstwcbm is accepting if
  $P=\emptyset$.  It follows that if a binary tree $\tau$ is accepted by
  \Akstwcbm then $\tau$ is a \kSTT and $\sem{\tau}=(G,\chi)$ is a
  split-\CBM with $\elastic$ defined as in \eqref{item:I5}.  From the definition
  of accepting states, we deduce that $P=\emptyset$, and ${\elastic}=\emptyset$.
  Therefore, $G$ is a \CBM and $\dom(\chi)=P=\emptyset$.
  
  Remark: There are some legal \kSTTs denoting \CBMs that are not accepted by
  the above automaton.  For instance, the term
  $$
  \tau = \add{i,j}\procrel \add{i,j}\procrel ((i,a,p)\oplus(j,b,p))
  $$
  is not accepted because the automaton prevents adding twice the same edge to 
  the graph. To circumvent this restriction, the automaton may additionally 
  store a relation ${\procrel}\subseteq P^2$ such that $i\procrel j$ iff 
  $\chi(i)\procrel\chi(j)$. Then, a transition $\add{i,j}\procrel$ is possible 
  if either $i\procrel j$ or $i\leadsto j$.
  
  Similarly, by keeping for each data-structure $d\in\DS$ a relation
  ${\matchrel}^d\subseteq P^2$ such that $i\matchrel^d j$ iff
  $\chi(i)\matchrel^d\chi(j)$, the tree automaton may allow adding several times
  a same data-structure edge.
\end{proof}

\begin{exercise}
  Prove that, if $\tau$ is a \kSTT and $\sem{\tau}=(G_\tau,\chi_\tau)$ is such
  that $G_\tau$ is a \CBM and $\dom(\chi_\tau)=\emptyset$ then $\tau$ is
  accepted by the tree automaton \Akstwcbm constructed in the proof of
  Proposition~\ref{prop:Akstwcbm}.
  
  Hint: The only non-determinism in the tree automaton \Akstwcbm occurs during 
  $\oplus$-transitions, when a strict partial order $\prec$ is guessed.
  When $\tau$ is a \kSTT such that $G_\tau$ is a \CBM, we can resolve the
  non-determinism of \Akstwcbm by choosing the order induced by $G_\tau$: if
  $\tau'$ is a subterm of $\tau$ then $G_{\tau'}$ is a subgraph of $G_\tau$ and
  the ordering $\prec$ guessed by \Akstwcbm at $\tau'$ should be $i\prec j$ iff
  $\chi_{\tau'}(i)<\chi_{\tau'}(j)$ in $G_\tau$.
\end{exercise}

\begin{exercise}
  Modify the automaton constructed in the proof of 
  Proposition~\ref{prop:Akstwcbm} in order to check that data-structures 
  $d\in\Stacks$ follow the LIFO policy and data-structures $d\in\Queues$ follow 
  the FIFO policy.
  
  Hint: For a data-structure $d\in\DS$, store a relation $R_d\subseteq
  (P\cup\Procs)^2$ with the invariant defined below.  
  For each event $e\in\Events$, define $\zeta(e)=\pid(e)$ if there is no active
  color $i$ such that $\chi(i)\procrel^* e$ and let $\zeta(e)$ be the maximal such
  $i$ otherwise.
  The invariant is $R_d=\{(\zeta(e),\zeta(f))\mid e\matchrel^d f\}$.
  When taking a $\oplus$-transition, make sure that the 
  policy of $d$ is respected.
\end{exercise}

\begin{proposition}\label{prop:Akstwsys}
  Given $\Sys\in\CPDS(\Arch,\Sigma)$, we can construct a tree automaton 
  \Akstwsys of size $|\Sys|^{\mathcal{O}(k+|\Procs|)}$ such that for all \kSTT 
  $\tau\in\Lang(\Akstwcbm)$, we have \par
  $\tau\in\Lang(\Akstwsys)$ iff $G_\tau$ is 
  accepted by the \CPDS \Sys.
\end{proposition}

\begin{proof}
  The automaton \Akstwsys follows the definitions and notations of 
  Section~\ref{sec:graph-semantics}.
  
  A state of \Akstwsys is a tuple $s=(P,\pi,\alpha,\beta,\delta,\sigma)$ where
  $P,\pi,\alpha,\beta$ are as in the proof of Proposition~\ref{prop:Akstwcbm},
  $\delta$ stores the transitions that are guessed for the events associated
  with active colors, and $\sigma$ stores the (partial) global final state.
  More precisely, when reading bottom-up a term $\tau\in\Lang(\Akstwcbm)$, the 
  tree automaton \Akstwsys reaches a state $s$  satisfying the
  following conditions with $\sem{\tau}=(G,\chi)$:
  \begin{enumerate}[label=($\mathsf{J}_{\arabic*}$),ref=$\mathsf{J}_{\arabic*}$]

    \item\label{item:J1} $P=\dom(\chi)\subseteq[k]$ is the set of active 
    colors, 
  
    \item\label{item:J2} $\pi\colon P\to\Procs$ gives the associated processes:
    $\pi(i)=\pid(\chi(i))$ for all $i\in P$,
    
    \item\label{item:J3} $\alpha,\beta\subseteq P$ are such that 
    \begin{itemize}[nosep]
      \item $i\in\alpha$ iff $\chi(i)$ is the source of a
      $\procrel$-edge in $G$.
    
      \item $i\in\beta$ iff $\chi(i)$ is the target of a
      $\procrel$-edge in $G$.
    \end{itemize}
  
    \item\label{item:J4} $\delta\colon P\to\Trans$ defines for each active color
    $i\in P$ the transition $\delta(i)$ guessed for the event $\chi(i)$.

    \item\label{item:J5} $\sigma\colon\Procs\to\Locs$ is a partial map which
    collects the global final state: when the color $i$ of the maximal event of
    some process $p$ is forgotten, we store the target state of $\delta(i)$ in 
    $\sigma(p)$.
  \end{enumerate}
  
  The number of states of \Akstwsys is $|\Sys|^{\mathcal{O}(k+|\Procs|)}$.

  \begin{table}[!t]
    \noindent\hspace{-3mm}
    \begin{tabular}{|c|p{120mm}|}
      \hline
      $\bot \xrightarrow{(i,a,p)}s$ 
      &
      if $P=\{i\}$, $\pi(i)=p$, $\alpha=\beta=\emptyset$,
      $\delta(i)\in\Delta_p$ with $\mathsf{lbl}(\delta(i))=a$ and $\dom(\sigma)=\emptyset$.
      \\ \hline
      $s\xrightarrow{\add{i,j}{\procrel}}s'$ 
      &
      if $i,j\in P$ and 
      $\mathsf{tgt}(\delta(i))=\mathsf{src}(\delta(j))$.
      \par
      Then, $P'=P$, $\pi'=\pi$, $\delta'=\delta$, $\sigma'=\sigma$,
      $\alpha'=\alpha\cup\{i\}$, $\beta'=\beta\cup\{j\}$.
      \\ \hline
      $s\xrightarrow{\add{i,j}{d}}s'$ 
      &
      if $i,j\in P$ and $\delta(i)\in\Delta^{!}$, $\delta(j)\in\Delta^{?}$ and
      $\mathsf{val}(\delta(i))=\mathsf{val}(\delta(j))$. \par
      Then, $s'=s$.
      \\ \hline
      $s\xrightarrow{\forget{i}}s'$ 
      &
      if $i\in P$ and 
      $i\in\beta \vee \mathsf{src}(\delta(i))=\iota_{\pi(i)}$
      and $i\in\alpha \vee i\notin\dom(\sigma)$.
      \par
      Then, $P'=P\setminus\{i\}$ and 
      $\pi',\alpha',\beta',\delta'$ are the
      restrictions of $\pi,\alpha,\beta,\delta$ to $P'$, and 
      $\sigma'=\sigma$ if $i\in\alpha$ and
      $\sigma'=\sigma[\pi(i)\mapsto\mathsf{tgt}(\delta(i))]$ otherwise.
      \\ \hline
      $s',s''\xrightarrow{\oplus}s$
      &
      if $P'\cap P''=\emptyset$  and
      $\dom(\sigma')\cap\dom(\sigma'')=\emptyset$ (this condition is actually 
      always satisfied for terms accepted by \Akstwcbm).
      
      Then, $s$ is the disjoint union of
      $s'$ and $s''$: $P=P'\uplus P''$, $\pi=\pi'\uplus\pi''$,
      $\alpha=\alpha'\uplus\alpha''$, $\beta=\beta'\uplus\beta''$,
      $\delta=\delta'\uplus\delta''$
      and $\sigma=\sigma'\uplus\sigma''$.
      \\ \hline
    \end{tabular}
    
    \caption{Transitions of \Akstwsys where $a\in\Sigma$, $p\in\Procs$, 
    $d\in\DS$, $i,j\in[k]$, $s=(P,\pi,\alpha,\beta,\delta,\sigma)$,
    $s'=(P',\pi',\alpha',\beta',\delta',\sigma')$ and 
    $s''=(P'',\pi'',\alpha'',\beta'',\delta'',\sigma'')$.
    }
    \protect\label{tab:Akstwsys}
  \end{table}

  The transitions of \Akstwsys are given in Table~\ref{tab:Akstwsys}.

  A state $s$ of \Akstwsys is accepting if $P=\emptyset$ and $\bar{\sigma}\in 
  F$ is an accepting global state of \Sys, where $\bar{\sigma}$ completes the 
  partial global state $\sigma$ with the initial location for 
  processes having no events in the \CBM:
  $\bar{\sigma}(p)=\sigma(p)$ if $p\in\dom(\sigma)$ and 
  $\bar{\sigma}(p)=\iota_p$ otherwise.
\end{proof}

\begin{corollary}\label{cor:stwnonemptiness}
  The problem {\sc{stw-Nonemptiness}}$(\Arch,\Sigma)$ is decidable in 
  \textsc{ExpTime}. The procedure is only polynomial in the size of the \CPDS.
\end{corollary}

\begin{proof}
  The problem reduces to checking nonemptiness of a tree automaton.  Given $k>0$
  and $\Sys\in\CPDS(\Arch,\Sigma)$, we have
  $\Lang(\Sys)\cap\kstwCBM\neq\emptyset$ iff
  $\Lang(\Akstwcbm\cap\Akstwsys)\neq\emptyset$.
  
  Indeed, given $\mscn\in\Lang(\Sys)\cap\kstwCBM$, we find 
  $\tau\in\Lang(\Akstwcbm)$ with $G_\tau=\mscn$ by 
  Proposition~\ref{prop:Akstwcbm}. We obtain $\tau\in\Lang(\Akstwsys)$ by 
  Proposition~\ref{prop:Akstwsys}.
  Conversely, if $\tau\in\Lang(\Akstwcbm\cap\Akstwsys)$ then 
  $G_\tau\in\kstwCBM$ by Proposition~\ref{prop:Akstwcbm} and 
  $G_\tau\in\Lang(\Sys)$ by Proposition~\ref{prop:Akstwsys}.
\end{proof}

\begin{proposition}\label{prop:Akstwphi}
  Given $k>0$ and an \MSO formula $\varphi\in\MSO(\Arch,\Sigma)$, we can
  construct a tree automaton \Akstwphi such that for all \kSTT $\tau$ we have
  $$
  \tau\in\Lang(\Akstwphi)
  \qquad\text{iff}\qquad
  G_\tau\models\varphi \,.
  $$
\end{proposition}

\begin{proof}
  From $\varphi\in\MSO(\Arch,\Sigma)$, we construct 
  $\tilde{\varphi}^{k}\in\MSO(\Lambda^{k},\da_0,\da_1)$ 
  using Proposition~\ref{prop:stt-mso-interpretation}. Using  \cite{ThaWri68} 
  we obtain an equivalent tree automaton \Akstwphi.
\end{proof}

\begin{corollary}\label{cor:stwmodelchecking}
  The problem {\sc{stw-ModelChecking}}$(\Arch,\Sigma)$ is decidable.  The
  procedure is only polynomial in the size of the \CPDS.
\end{corollary}

\begin{proof}
  The problem reduces to checking emptiness of a tree automaton.  Given
  $k>0$, a \CPDS \Sys, and a formula $\varphi\in\MSO(\Arch,\Sigma)$ we have
  $\Lang(\Sys)\cap\kstwCBM\subseteq\Lang(\varphi)$ iff
  $\Lang(\Akstwcbm\cap\Akstwsys\cap
  \A_{\neg\varphi}^{k\text{-}\mathsf{stw}})=\emptyset$.
\end{proof}

\ifPDL

\newcommand{\back}[1]{\ensuremath{#1^{-1}}}
\newcommand{\existspath}[1]{\ensuremath{\langle#1\rangle}}
\newcommand{\existsloop}[1]{\ensuremath{\mathsf{Loop}\langle#1\rangle}}
\newcommand{\forallpath}[1]{\ensuremath{[#1]}}

\newcommand{\Existsev}{\mathsf{E}}
\newcommand{\Forallev}{\mathsf{A}}
\newcommand{\AP}{\ensuremath{\mathsf{AP}}}
\newcommand{\U}{\mathbin{\mathsf{U}}}
\newcommand{\X}{\mathop{\mathsf{X}\vphantom{a}}\nolimits}

\chapter{Propositional Dynamic Logic}

In the previous sections, we showed that model checking \CPDSs against \MSO formulas is
undecidable in general but decidability is recovered when we restrict to behaviors of
bounded (special) tree-width.  However, the transformation of an \MSO formula into a tree
automaton is inherently non-elementary, and this non-elementary lower bound is in fact
inherited by the model-checking problem against specifications in \MSO. The goal in this
chapter is to study a more tractable specification language in the spirit of linear
temporal logic (\LTL) for sequential systems.  We will focus on (variants of) the
Propositional Dynamic Logic (\PDL) introduced by Fischer and Ladner for reasonning about
program schemas \cite{Fischer_1979}.  Over finite words, $\PDL$ is as expressive as \MSO
(hence more expressive than $\LTL=\FO$) since it includes regular expressions.  But
verification problems for \PDL can be solved in \textsc{PSpace}, i.e., as efficiently as
for \LTL. We will see that for concurrent systems with data structures (\CPDS), the logic
\PDL, though undecidable in general, is decidable in \textsc{ExpTime} (or
2-\textsc{ExpTime}) when we restrict to behaviours with bounded (special) tree-width.

\section{Propositional Dynamic Logic}

\underline{{\bf Syntax:}}

Let $\AP\neq\emptyset$ be a nonempty set of atomic propositions (node labels) and
$\Gamma\neq\emptyset$ be a nonempty set of atomic programs (edge labels).
The syntax of $\ICPDL(\AP,\Gamma)$ is given by
\begin{align*}
  \Phi &::= \Existsev \sigma \mid \Phi \vee \Phi \mid \neg \Phi
  \\
  \sigma &::= p \mid \sigma \vee \sigma \mid \neg \sigma 
  \mid \existspath{\pi}{\sigma} 
  && && \mid \existsloop{\pi}
  \\
  \pi &::= \gamma 
  \mid \test{\sigma} \mid
  \pi + \pi \mid \pi \cdot \pi \mid \pi^{\ast}
  &&\mid \pi^{-1} && &&\mid \pi\cap\pi
\end{align*}
where $p\in\AP$, $\gamma\in\Gamma$.
We call
\begin{itemize}
  \item $\Phi$ a sentence,
  \item $\sigma$ a \emph{state formula} or \emph{node formula}, it has one implicit state 
  variable,
  \item $\pi$ a \emph{program}, or \emph{path formula}, it has two implicit state 
  variables.
\end{itemize}
If intersection $\pi\cap\pi$ is not allowed, the fragment is \PDL with loop and
converse (\LCPDL).  If neither intersection nor loop are allowed, the fragment
is \PDL with converse (\CPDL).  If backward paths $\pi^{-1}$ are not allowed
the fragment is called \PDL with intersection (\IPDL).  In simple \PDL neither
backward paths nor intersections nor loops are allowed.

\medskip
\underline{{\bf Semantics of \ICPDL formulas:}}

Let $G=(V,(E_\gamma)_{\gamma\in\Gamma},\lambda)$ be a $(2^\AP,\Gamma)$-labeled
graph: $\lambda\colon V\to2^\AP$ and $E_\gamma \subseteq V^2$.

\smallskip
\underline{{\bf State formulas:}}
The semantics of a state formula $\sigma$ wrt.\ $G$ is a set
$\Sem{\sigma}{G} \subseteq V$, inductively defined below. We also write 
$G,v\models\sigma$ for $v\in\Sem{\sigma}{G}$.
\[
\begin{array}{rcl}
\Sem{p}{G} & := & \{v\in V\mid p\in\lambda(v)\} \\[1ex]

\Sem{\sigma_1 \vee \sigma_2}{G} & := & \Sem{\sigma_1}{G} \cup \Sem{\sigma_2}{G}\\[1ex]

\Sem{\neg\sigma}{G} & :=&  V \setminus \Sem{\sigma}{G}\\[1ex]

\Sem{\existspath{\pi}\sigma}{G} & := & \{v \in V \mid \textup{there is } w \in \Sem{\sigma}{G} 
\textup{ such that } (v,w) \in \Sem{\pi}{G}\}\\[1ex]

\Sem{\existsloop{\pi}}{G} & := & \{v \in V \mid (v,v) \in \Sem{\pi}{G}\}
\end{array}
\]

\underline{{\bf Path formulas:}}
The semantics of a path formula $\pi$ wrt.\ $G$ is a set $\Sem{\pi}{G} \subseteq
V \times V$, inductively defined below. We also write 
$G,u,v\models\pi$ for $(u,v)\in\Sem{\pi}{G}$.
\[
\begin{array}{rcl}

\Sem{\gamma}{G} & := & E_\gamma \\[1ex]

\Sem{\test{\sigma}}{G} & := & \{(v,v) \mid v \in \Sem{\sigma}{G}\}\\[1ex]

\Sem{\back{\pi}}{G} & := & \Sem{\pi}{G}^{-1} = \{(u,v) \mid (v,u) \in \Sem{\pi}{G}\}\\[1ex]

\Sem{\pi_1 + \pi_2}{G} & := & \Sem{\pi_1}{G} \cup \Sem{\pi_2}{G}
\qquad\qquad
\Sem{\pi_1 \cap \pi_2}{G} ~:=~ \Sem{\pi_1}{G} \cap \Sem{\pi_2}{G}\\[1ex]

\Sem{\pi_1 \cdot \pi_2}{G} & := & \Sem{\pi_1}{G} \circ \Sem{\pi_2}{G}\\[1ex]
& & = \{(u,w) \in V \times V \mid \exists v \in V: (u,v) \in \Sem{\pi_1}{G} 
\textup{ and } (v,w) \in \Sem{\pi_2}{G}\}\\[1ex]

\Sem{\pi^\ast}{G} & := & \Sem{\pi}{G}^\ast = \bigcup_{n \in \mathds{N}} \Sem{\pi}{G}^n
\end{array}
\]

\underline{{\bf Sentences:}}
We write $G \models \Existsev\sigma$ if $\Sem{\sigma}{G} \neq \emptyset$.  

For $\Phi \in \ICPDL(\AP,\Gamma)$, we let $L(\Phi) := \{G\mid G \models \Phi\}$.

\medskip
One can show that \ICPDL is no more expressive than \MSO:

\begin{exercise}\label{exc:ICPDLtoMSO}
  Show that, for every $\ICPDL(\AP,\Gamma)$ sentence $\Phi$, state formula
  $\sigma$ and path formula $\pi$, we can construct \emph{equivalent}
  $\MSO(\AP,\Gamma)$ formulas $\overline{\Phi}$, $\overline{\sigma}(x)$ and
  $\overline{\pi}(x,y)$ where $\overline{\Phi}$ is a sentence,
  $\mathsf{Free}(\overline{\sigma})=\{x\}$ and
  $\mathsf{Free}(\overline{\pi})=\{x,y\}$.
\end{exercise}

Notice that \MSO is strictly more powerful than \ICPDL. Indeed, we cannot 
express that a graph is connected in \ICPDL. Also, the modality $\U_2$ defined 
below checks a path of even length and therefore cannot be expressed in \FO.

Over message passing automata ($\DS=\Queues$) it was shown very recently
that $\FO(<,\procrel,\matchrel)$ has the same expressive power as the
starfree fragment of \LCPDL \cite{BFG-concur18}.

\medskip
For \CBMs over $(\Arch,\Act)$, we use the set $\AP=\Procs\cup\Act$ of atomic propositions
and the set $\Gamma=\{\procrel\}\cup\DS$ of atomic programs and we often write
$\matchrel^d$ instead of $d$ for atomic steps of path formulas.
We then write $\mathsf{(ILC)PDL}(\Arch,\Sigma)$.

\begin{example}
Consider the following abbreviation/examples:
\begin{itemize}
\item $\Forallev \sigma ~\fequiv~ \neg\Existsev\neg\sigma$ ~~~~~($\ICPDL$ formula)
\item $\true ~\fequiv~ p \vee \neg p$ ~~~~~(state formula)
\item $\forallpath{\pi}\sigma ~\fequiv~ \neg\existspath{\pi}\neg\sigma$ ~~~~~(state formula)
\item $\existspath{\pi} ~\fequiv~ \existspath{\pi}\true$ ~~~~~(state formula)
\item $\X\sigma ~\fequiv~ \existspath{\procrel}\sigma$ ~~~~~(state formula)
\item $\sigma_1\U\sigma_2 ~\fequiv~ \existspath{(\test{\sigma_1}\cdot{\procrel})^*}\sigma_2$ ~~~~~(state formula)
\item $\sigma_1\U_2\sigma_2 ~\fequiv~ \existspath{(\test{\sigma_1}\cdot{\procrel}\cdot\test{\sigma_1}\cdot{\procrel})^*}\sigma_2$ ~~~~~(state formula)
\item ${\matchrel} \fequiv \sum_{d \in \DS} {\matchrel^d}$ ~~~~~(path formula)
\item $\mathsf{write} \fequiv \existspath{\matchrel}\true$ and 
$\mathsf{read} \fequiv \existspath{\matchrel^{-1}}\true$~~~~~(state formulas)
\item $\Phi_1 \fequiv \Forallev \bigl(a \Rightarrow \existspath{({\procrel} + {\matchrel})^\ast}b\bigr) \in \PDL(\Arch,\Sigma)$
\item $\mathsf{req\textup{-}ack} \fequiv (\test{p_1} \cdot {\matchrel^{c_1}} \cdot {\procrel} \cdot {\matchrel^{c_2}})
~~+~~ (\test{p_1} \cdot {\matchrel^{c_1}} \cdot {\procrel} \cdot {\matchrel^s} \cdot {\procrel} \cdot {\matchrel^{c_2}})$\\
(path formula; cf.\ client-server system from previous chapter)
\item $\Phi_2 =
\Forallev (a \Rightarrow \forallpath{\mathsf{req\textup{-}ack}}a) ~\wedge~
\Forallev (b \Rightarrow \forallpath{\mathsf{req\textup{-}ack}}b) ~~~\in \PDL(\Arch,\Sigma)$\\[1ex]
$\equiv \Forallev (\forallpath{\test{a} \cdot \mathsf{req\textup{-}ack}}a) ~\wedge~
\Forallev (\forallpath{\test{b} \cdot \mathsf{req\textup{-}ack}}b) ~~~\in \PDL(\Arch,\Sigma)$\\[1ex]
$\equiv \forall x,y \Bigl(\mathit{req\textup{-}ack}(x,y) \Rightarrow \bigl((a(x) \wedge a(y)) \vee (b(x) \wedge b(y))\bigr)\Bigr)~\in \MSO(\Arch,\Sigma)$
\end{itemize}
We have $\Lang(\csSys) \not\subseteq L(\Phi_1)$ and $\Lang(\csSys) \subseteq L(\Phi_2)$.
\end{example}

\begin{example}\label{ex:pdl-cbm}
  Let $\Arch$ be an architecture, $\AP=\Sigma\uplus\Procs$ and $\Gamma=\{\procrel\}\cup\DS$.
  Most conditions ensuring that a $(\AP,\Gamma)$-labeled graph
  $G=(\Events,\procrel,(\matchrel^d)_{d\in\DS},\pid,\lambda)$ is in fact a \CBM
  (Definition~\ref{def:cbm}) can be expressed in \LCPDL. Missing is a formula stating that
  on each $\Events_p$, $\procrel$ is the successor relation of a \emph{total} order.
  \begin{itemize}
    \item $\mathsf{LABELS}=\Forallev \Big( 
    \bigvee_{p\in\Procs} \big(p \wedge \bigwedge_{p'\neq p} \neg p' \big)
    \wedge \bigvee_{a\in\Act} \big(a \wedge \bigwedge_{a'\neq a} \neg a' \big)
    \Big)$
  
    \item $\mathsf{PROCESS}=
    \Forallev\big({\existspath{\procrel}}\implies
    \bigvee_{p\in\Procs}(p\wedge\existspath{{\procrel}}p)\big)$
    
    \item $\mathsf{ORDER}=\Forallev\neg\existsloop{({\procrel}+{\matchrel})^+}$

    \item $\mathsf{SUCCESSOR}=
    \Forallev\neg\existsloop{{\procrel}^{+}\cdot{\procrel}^{+}\cdot{\procrel}^{-1}}$

    \item $\mathsf{WRITER}=\Forallev\bigwedge_{d\in\DS}
    {\existspath{\matchrel^d}}\implies\writer(d)$
  
    \item $\mathsf{READER}=\Forallev\bigwedge_{d\in\DS}
    {\existspath{(\matchrel^d)^{-1}}}\implies\reader(d)$

    \item $\mathsf{DISJOINT}=\neg\Existsev\Big(
    \existspath{{\matchrel}\cdot{\matchrel}} \vee
    \bigvee_{d\neq d'}
    \big(\existspath{\matchrel^d}\wedge\existspath{\matchrel^{d'}}\big)
    \vee\big(\existspath{(\matchrel^d)^{-1}}\wedge\existspath{(\matchrel^{d'})^{-1}}\big)
    \Big)$
    
    \hspace{18mm}${}\wedge\neg\Existsev\Big(\bigvee_{d\in\DS}
    \existsloop{({\matchrel}^d\cdot{\procrel}^+\cdot({\matchrel}^d)^{-1})
    +(({\matchrel}^d)^{-1}\cdot{\procrel}^+\cdot{\matchrel}^d})
    \Big)$

    \item $\mathsf{FIFO}=\Forallev\bigwedge_{c\in\Queues}
    \neg\existsloop{{\matchrel}^c\cdot{\procrel}^+\cdot({\matchrel}^c)^{-1}\cdot{\procrel}^+}$

    \item $\mathsf{LIFO}=\Forallev\bigwedge_{s\in\Stacks}
    \Big( {\existspath{\matchrel^s}}\implies{}$
    
    \hspace{25mm}$\existsloop{{\procrel}\cdot
    \big( {\matchrel}^s\cdot{\procrel} + 
    \test{\neg\existspath{{\matchrel}^s+({\matchrel}^s)^{-1}}}\cdot{\procrel} 
    \big)^* \cdot({\matchrel}^s)^{-1}} \Big)$
  \end{itemize}
\end{example}

\underline{{\bf Quantified propositions:}}

It is useful to extend \ICPDL with existentially quantified propositions and cardinality 
constraints. A sentence in \EQICPDL is a formula of the form
$$
\Psi=\exists p_1,\ldots,p_n~\chi\wedge\Phi
$$
where $\Phi$ is an \ICPDL sentence which may use atomic propositions in
$\AP'=\AP\uplus\{p_1,\ldots,p_n\}$, and $\chi$ expresses cardinality constraints, i.e., is
a boolean combination of formulas of the form $\#p_i\leq c$ with $1\leq i\leq n$ and
$c\in\mathbb{N}$ (constants are usually encoded in binary).

If $\Phi\in\LCPDL$ then $\Psi\in\EQLCPDL$ and similarly for the other fragments.

\underline{{\bf Semantics of \EQICPDL:}}
Let $G=(V,(E_\gamma)_{\gamma\in\Gamma},\lambda)$ be a $(2^\AP,\Gamma)$-labeled
graph.
Then, $G\models\Psi$ if we can extend the vertex labeling to $\lambda'\colon V\to
2^{\AP'}$ with $\lambda'(v)\cap\AP=\lambda(v)$ for all $v\in V$, such that
\begin{itemize}
  \item  $G'=(V,(E_\gamma)_{\gamma\in\Gamma},\lambda')\models\Phi$ and 

  \item  $\chi$ evaluates to true when $\#p_i$ is replaced with $m_i=|\{v\in V\mid 
  p_i\in\lambda'(v)\}|$ for each $1\leq i\leq n$.
\end{itemize}

This extension allows to express connectivity in $\EQCPDL$:
$$
\mathsf{CONNECTED} = \exists p~(\#p\leq1)\wedge 
\Forallev\,\existspath{\big(\textstyle{\sum_{\gamma\in\Gamma}}~\gamma+\gamma^{-1}\big)^{*}}p
$$

On \CBMs, we can express in \EQPDL that the process relation has a maximal element on 
each process:
$$
\mathsf{MAXPROCESS} = \exists (Z_p)_{p\in\Procs}~
\Big( \textstyle{\bigwedge_{p\in\Procs}}~\#Z_p\leq1 \Big)
\wedge 
\Forallev\Big( p \implies \existspath{\procrel^{*}} (p\wedge Z_p) \Big)
$$
To get a \emph{linear} order on each process, it remains to state that we have no join, 
which is first-order definable but not \EQICPDL definable.
\begin{lemma}\label{lem:cbm-pdl-nojoin}
  Let $G$ be a $(\AP,\Gamma)$-labeled graph. Then, $G\in\CBM(\Arch,\Sigma)$ iff $G$ 
  satisfies $\mathsf{ALMOSTCBM}\wedge\mathsf{NOJOIN}$ where $\mathsf{ALMOSTCBM}$ is in 
  \EQLCPDL and $\mathsf{NOJOIN}$ is a first-order sentence:
  \begin{align*}
    \mathsf{ALMOSTCBM} & = \mathsf{LABELS} \wedge \mathsf{PROCESS} 
    \wedge \mathsf{ORDER} \wedge \mathsf{MAXPROCESS}
    \\
    &\phantom{{}= \mathsf{LABELS}}\wedge \mathsf{WRITER} \wedge \mathsf{READER} 
     \wedge \mathsf{DISJOINT} \wedge \mathsf{FIFO} \wedge \mathsf{LIFO}
    \\
    \mathsf{NOJOIN} & = \forall x,y,z~(x\procrel z \wedge y\procrel z \implies x=y)
  \end{align*}  
\end{lemma}

\begin{proof}
  It is easy to see that a \CBM satisfies $\mathsf{ALMOSTCBM}\wedge\mathsf{NOJOIN}$.
  Conversely, assuming that a graph $G\models\mathsf{ALMOSTCBM}\wedge\mathsf{NOJOIN}$, we
  show that $\procrel$ is the successor relation of a total order on each process.  First,
  $\mathsf{LABELS} \wedge \mathsf{PROCESS} \wedge \mathsf{ORDER}$ implies that
  $\procrel^{*}$ consists of disjoint partial orders on each process.  In addition,
  $\mathsf{MAXPROCESS}$ implies that $\procrel$ defines on each process a connected
  directed acyclic graph (DAG) with a unique maximal element.  Finally, $\mathsf{NOJOIN}$
  implies that the DAG is actually a linear graph.
\end{proof}

One of the main advantages of the extension is the following:
\newcommand{\Psisys}{\ensuremath{\Psi_\Sys}\xspace}

\begin{mytheorem}\label{thm:CPDStoPDL}
  For every $\Sys \in \CPDS(\Arch,\Sigma)$, there is a sentence
  $\Psisys\in\EQPDL(\Arch,\Sigma)$ of size $\mathcal{O}(|\Sys|^{2})$ such that 
  $\Lang(\Sys)=\{\mscn\in\CBM(\Arch,\Sigma)\mid\mscn\models\Psisys\}$.
\end{mytheorem}

\begin{proof}
  Fix $\Sys=((\Sys_p)_{p\in\Procs},\Val,\FinLocs)\in\CPDS(\Arch,\Sigma)$.  Recall the
  notations of Definition~\ref{def:cpds}.  We use a quantified atomic proposition $X_t$
  for each transition $t$ and $Y$ for the minimal events of processes.  We define
  \[
  \Psisys =
  \begin{array}[t]{rl}
    \multicolumn{2}{l}{\exists (X_t)_{t\in\mathsf{Trans}} \, \exists Y\, \Bigl[ 
    \Forallev \, \bigvee_{t\in\mathsf{Trans}}~
    \Big( X_t \wedge \bigwedge_{t'\neq t}~\neg X_{t'} \Big)
    }
    \\[2ex]
    \wedge & \Forallev \, \bigwedge_{p\in\Procs} \bigl(p \implies 
    \bigvee_{t\in\Delta_p} X_t \bigr)
    \\[2ex]
    \wedge & \Forallev \, \bigwedge_{a\in\Sigma} \bigl(a \implies 
    \bigvee_{t\in\mathsf{Trans}\mid\mathsf{lbl}(t)=a} X_t \bigr)
    \\[2ex]
    \wedge & \Forallev \, \bigl( \existspath{\procrel} \implies 
    \bigvee_{\substack{t,t'\in\mathsf{Trans} \mid \mathsf{tgt}(t)=\mathsf{src}(t')}} 
    X_t \wedge \existspath{\procrel}X_{t'} \bigr)
    \\[2ex]
    \wedge & \Forallev \, \bigl( \existspath{\matchrel} \implies 
    \bigvee_{\substack{t\in\Delta^{!},t'\in\Delta^{?} \mid \mathsf{val}(t)=\mathsf{val}(t') 
    \wedge \mathsf{ds}(t)=\mathsf{ds}(t')}} X_t \wedge 
    \existspath{\matchrel^{\mathsf{ds}(t)}}X_{t'} \bigr)
    \\[2ex]
    \wedge & \Big( \Forallev \, \neg\existspath{\procrel}Y \Big)
    \wedge \Big( \bigwedge_{p\in\Procs} \Existsev \, p \implies \Existsev \,(p\wedge Y) \Big)
    \\[2ex]
    \wedge & \Forallev \, \bigl( Y \implies
    \bigvee_{p\in\Procs,t\in\Delta_p\mid\mathsf{src(t)=\iota_p}} X_t \bigr)
    \\[2ex]
    \wedge & 
    \bigvee_{(\ell_p)_{p\in\Procs}\in\FinLocs}
    \Big( \bigwedge_{\substack{p\in\Procs \\ \mid \ell_p\neq\iota_p}}~ \Existsev \, p \Big)
    \wedge \Forallev \, \Big( \neg\existspath{\procrel} 
    \implies \bigvee_{\substack{p\in\Procs,t\in\Delta_p \\ 
    \mid\mathsf{tgt}(t)=\ell_p}} X_{t} \Big)
    \,\Bigr]
  \end{array}
  \]
  This completes the construction of the desired formula $\Psisys$. 
\end{proof}

\clearpage
\section{Satisfiability and Model Checking}

For an architecture $\Arch$ and an alphabet $\Sigma$, consider the following problems:

\begin{center}
\begin{tabular}{ll}
\toprule
\textsf{(EQ-)(ILC)}\PDL-\textsc{Satisfiability}$(\Arch,\Sigma)$:\\
\midrule
{Instance:} & \hspace{-10em}$\Phi \in \textsf{(EQ-)(ILC)PDL}(\Arch,\Sigma)$\\[0.5ex]
{Question:} & \hspace{-10em}$L(\Phi) \neq \emptyset$\,?\\
\bottomrule
\end{tabular}
\end{center}

\begin{center}
\begin{tabular}{ll}
\toprule
\textsf{(EQ-)(ILC)}\PDL-\textsc{ModelChecking}$(\Arch,\Sigma)$:\\
\midrule
{Instance:} & \hspace{-10em}$\Sys \in \CPDS(\Arch,\Sigma)$\,; $\Phi \in 
\textsf{(EQ-)(ILC)PDL}(\Arch,\Sigma)$\\[0.5ex]
{Question:} & \hspace{-10em}$L(\Sys) \subseteq L(\Phi)$\,?\\
\bottomrule
\end{tabular}
\end{center}

\bigskip

\begin{mytheorem}
Let $\Arch$ be given as follows (and $\Sigma$ be arbitrary):
\begin{center}
{
\includegraphics[scale=0.5]{simplearch.pdf}
}
\end{center}
Then, all the abovementioned problems are undecidable.
\end{mytheorem}

This is proved using a reduction from the (non)emptiness problem of \CPDSs.  
The reduction is trivial for the \textsc{ModelChecking} problem, we simply use the
specification $\Phi=\mathsf{False}=\Existsev\,(p\wedge\neg p)$.  
For the \textsc{Satisfiability} problem of \EQPDL, the reduction follows from
Theorem~\ref{thm:CPDStoPDL}: given $\Sys\in\CPDS$ we construct $\Psisys$.  

For \PDL-\textsc{Satisfiability}, we start with $\Sys\in\CPDS(\Arch,\Sigma)$ and we
construct $\Phi_\Sys\in\PDL(\Arch,\Sigma')$ where
$\Sigma'=\Sigma\times\Trans\times\{0,1\}$.  The second component of $\Sigma'$ allows to
label events with transitions (corresponding to $(X_t)_{t\in\Trans}$) and the last 
component allows to mark the minimal event of each process (corresponding to $Y$). Then 
$\Phi_\Sys$ is similar to the kernel of $\Psisys$.

By Theorem~\ref{PDLtoCPDS}, we obtain the following positive result:

\begin{mytheorem}
Suppose $\DS=\Bags$. Then, the problems
\begin{itemize}
\item[] \PDL-\textsc{Satisfiability}$(\Arch,\Sigma)$ and
\item[] \PDL-\textsc{ModelChecking}$(\Arch,\Sigma)$
\end{itemize}
are both decidable.
\end{mytheorem}

\newpage
\section{\PDL and special tree-width}\label{sec:stwpdl}

We have seen in Section~\ref{sec:stw-tree-interpretation} that a $(2^\AP,\Gamma)$-labeled
graph $G$ with special tree-width at most $k$ can be interpreted in any \kSTT $\tau$ such
that $\sem{\tau}=(G,\chi)$.  We prove now a backward translation for \PDL, which is an
analog of Proposition~\ref{prop:stt-mso-interpretation}.

We define $\AP^k=\{\sttunion, i, \add{i}{p}, \add{i,j}\gamma, \forget{i} \mid
i,j\in[k], p\in\AP, \gamma\in\Gamma\}$.

\begin{proposition}[PDL interpretation]\label{prop:stt-pdl-interpretation}
  For all sentences $\Phi\in\EQICPDL(\AP,\Gamma)$ and all $k>0$, we can
  construct a sentence
  $\widetilde{\Phi}^k\in\EQICPDL(\AP^k,\da_0,\da_1)$ of size 
  $\mathcal{O}(k^2|\Phi|)$ such that,
  for every valid \kSTT $\tau$ with $\sem{\tau}=(G,\chi)$, we have
  $$
  G\models\Phi \textup{~~~iff~~~} \tau\models \widetilde{\Phi}^k \,.
  $$
  Moreover, if $\Phi$ is in some fragment $\textsf{(EQ-)(IL)CPDL}(\Arch,\Sigma)$
  then we can construct $\widetilde{\Phi}^k$ in the corresponding fragment
  $\textsf{(EQ-)(IL)CPDL}(\AP^k,\da_0,\da_1)$.
\end{proposition}

\begin{proof}
  We prove by induction that for all $\ICPDL(\AP,\Gamma)$ sentences $\Phi$, state formulas
  $\sigma$ and path formulas $\pi$, and all $k>0$, there are $\ICPDL(\AP^k,\da_0,\da_1)$
  formulas $\widetilde{\Phi}^k$, $\widetilde{\sigma}^k$ and $\widetilde{\pi}^k$ such that,
  for all valid \kSTT $\tau$ with $\sem{\tau}=(G,\chi)$, and all vertices $e,f$ of $G$
  (identified with leaves of $\tau$) we have
  \begin{align*}
    G\models\Phi &\textup{~~~iff~~~} \tau\models \widetilde{\Phi}^k \\
    G,e\models\sigma &\textup{~~~iff~~~} \tau,e\models \widetilde{\sigma}^k \\
    G,e,f\models\pi &\textup{~~~iff~~~} \tau,e,f\models \widetilde{\pi}^k 
  \end{align*}

  The difficult cases are for edge relations $\gamma\in\Gamma$ and vertex labels
  $p\in\AP$.  We define
  \begin{align*}
    \widetilde{p}^{k} &= \bigvee_{0\leq i\leq k} 
    i \wedge \existspath{(\test{\neg\forget{i}}\cdot{\uparrow})^+}{\add{i}p}
    \\
    \widetilde{\gamma}^k &= \sum_{0\leq i,j\leq k} 
    \test{i}\cdot(\test{\neg\forget{i}}\cdot{\uparrow})^+
    \cdot\test{\add{i,j}\gamma}\cdot(\test{\neg\forget{j}}\cdot\da)^+
    \cdot\test{j} \,.
  \end{align*}
  where $\da=\da_0+\da_1$ and $\ua=\da^{-1}$.
  The formula $\widetilde{p}^k$ is of size $\mathcal{O}(k)$ and
  the formula $\widetilde{\gamma}^k$ is of size $\mathcal{O}(k^2)$.
  The other cases are trivial:
  \begin{align*}
    \widetilde{\Existsev\sigma}^k &= \Existsev(\neg\existspath{\da}{}\wedge\widetilde{\sigma}^k)
    &
    \widetilde{\neg\Phi}^k &= \neg\,\widetilde{\Phi}^k
    &
    \widetilde{\Phi_1\vee\Phi_2}^k &= \widetilde{\Phi_1}^k \vee \widetilde{\Phi_2}^k
    \\
    &&
    \widetilde{\neg\sigma}^k &= \neg\,\widetilde{\sigma}^k
    &
    \widetilde{\sigma_1 \vee \sigma_2}^k &= \widetilde{\sigma_1}^k \vee \widetilde{\sigma_2}^k
    \\
    &&
    \widetilde{\existspath{\pi}\sigma}^k &= \existspath{\widetilde{\pi}^k}{\widetilde{\sigma}^k}
    &
    \widetilde{\existsloop{\pi}}^k &= \existsloop{\widetilde{\pi}^k}
    \\
    \widetilde{\test{\sigma}}^k &= \test{\widetilde{\sigma}^k}
    &
    \widetilde{\pi^{-1}}^k &= (\widetilde{\pi}^k)^{-1}
    &
    \widetilde{\pi_1+\pi_2}^k &= \widetilde{\pi_1}^k + \widetilde{\pi_2}^k
    \\
    \widetilde{\pi_1\cdot\pi_2}^k &= \widetilde{\pi_1}^k \cdot \widetilde{\pi_2}^k
    &
    \widetilde{\pi^{\ast}}^k &= (\widetilde{\pi}^k)^{\ast}
    &
    \widetilde{\pi_1\cap\pi_2}^k &= \widetilde{\pi_1}^k \cap \widetilde{\pi_2}^k
  \end{align*}
  Finally, for an existentially quantified formula 
  $\Psi=\exists\,p_1,\ldots,p_n~\chi\wedge\Phi$ we define
  $$
  \widetilde{\Psi}^{k}=\exists\,p_1,\ldots,p_n~~\chi\wedge\widetilde{\Phi}^{k}
  \wedge\Forallev\!\! \bigwedge_{1\leq i\leq n} p_i \implies \neg\existspath{\da}{}
  $$
  Notice that $\Phi$ may use atomic propositions in $\AP'=\AP\uplus\{p_1,\ldots,p_n\}$. 
  In addition to $\widetilde{p}^{k}$ for $p\in\AP$ defined above, we simply let 
  $\widetilde{p_i}^{k}=p_i$ for $1\leq i\leq n$. Indeed, in $\widetilde{\Psi}^{k}$, the 
  existential propositional quantification directly labels nodes of the tree with atomic 
  propositions from $\{p_1,\ldots,p_n\}$ and the last conjunct in $\widetilde{\Psi}^{k}$ 
  ensures that only leaves (i.e., vertices of $G$) are labeled with these atomic 
  propositions.
\end{proof}

\begin{exercise}
  Write a sentence $\Phi^{k}_{\mathsf{valid}}\in\PDL(\AP^{k},\da_0,\da_1)$ 
  stating that the binary tree is a valid \kSTT. What is the size of 
  $\Phi^{k}_{\mathsf{valid}}$?
\end{exercise}

\begin{theorem}[G{\"o}ller, Lohrey, and Lutz \cite{Goeller2009}]\label{thm:pdl-tree}
  For a given sentence $\Phi\in\LCPDL(\AP,\da_0,\da_1)$ over binary trees, we can
  construct a tree automaton $\mathcal{B}_\Phi$ of size
  $2^{\poly(|\Phi|)}$ such that, for every binary tree $\tau$, we have
  $$
  \tau\models\Phi \textup{~~~iff~~~} \tau\in\Lang(\mathcal{B}_\Phi) \,.
  $$
  Moreover, if $\Phi\in\ICPDL(\AP,\da_0,\da_1)$ then we can construct an 
  equivalent $\mathcal{B}_\Phi$ of double exponential size.
\end{theorem}

\begin{proof}
  We construct an \emph{alternating tree walking automaton}
  (\ATWA) which is equivalent to $\Phi$.  An \ATWA ``walks'' in a
  tree, similarly to a path formula from \CPDL. In addition, it can spawn
  several copies of an automaton, which \emph{all} have to accept the input.
  This spawning is dual to non-deterministic choice, hence the name
  \emph{alternating}.  If $\Phi$ is an \LCPDL formula, then the
  \ATWA is of size polynomial in $\Phi$.
  
  We sketch some steps of the construction. For each state subformula $\sigma$ there is a 
  state $q_\sigma$ of the \ATWA so that for each node $e$ of the tree $T$ we have 
  $T,e\models\sigma$ iff the \ATWA accepts starting from the configuration 
  $(q_\sigma,e)$. The interesting cases are when $\sigma=\existspath{\pi}{}$ or 
  $\sigma=\existsloop{\pi}$. We construct an automaton $\mathcal{C}_\pi$ associated with 
  $\pi$. For instance, consider 
  $\pi=\da^{+}\cdot\test{a}\cdot\ua^{+}\cdot\test{b}\cdot\da^{+}\cdot\test{c}\cdot\ua^{+}$.
  The automaton $\mathcal{C}_\pi$ is
  
  \begin{gpicture}[name = path-automaton]
    \gasset{Nw=6, Nh = 6, Nframe = y}\gasset{loopdiam=4}
    \node[Nmarks=i](0)(0,0){0}
    \node(1)(15,0){1}
    \node(2)(35,0){2}
    \node(3)(50,0){3}
    \node(4)(70,0){4}
    \node(5)(85,0){5}
    \node(6)(105,0){6}
    \node[Nmarks=r](7)(120,0){7}
    \drawedge(0,1){$\da$}\drawloop(1){$\da$}
    \drawedge(1,2){$\test{a}$}
    \drawedge(2,3){$\ua$}\drawloop(3){$\ua$}
    \drawedge(3,4){$\test{b}$}
    \drawedge(4,5){$\da$}\drawloop(5){$\da$}
    \drawedge(5,6){$\test{c}$}
    \drawedge(6,7){$\ua$}\drawloop(7){$\ua$}
  \end{gpicture}
  
  For formula $\sigma=\existspath{\pi}{}$ we start the walking automaton 
  $\mathcal{C}_\pi$ in state $0$ and it should accept in state $7$.
  If instead of a simple $\test{a}$ with $a\in\AP$ we have some $\test{\sigma'}$ then we 
  spawn a copy of the \ATWA in state $q_{\sigma'}$ from the current node of the tree.
  
  The case $\sigma=\existsloop{\pi}$ is more interesting since the walking automaton 
  cannot check that at the end of the path it is back in the same node of the tree. The 
  trick is to consider pairs $(p,q)$ of states of $\mathcal{C}_\pi$ and the \ATWA will 
  accept from state $(p,q)$ starting at node $e$ if there is a loop on $e$ following a 
  path going from $p$ to $q$ in $\mathcal{C}_\pi$. It uses the following transitions, 
  where $\varepsilon$ means that we stay at the same node:
  \begin{align*}
    (p,p)&\xrightarrow{\varepsilon}\mathsf{accept} 
    \\
    (p,q)&\xrightarrow{\varepsilon}q_{\sigma'} 
    &&\text{if } p\xrightarrow{\test{\sigma'}}q \text{ in } \mathcal{C}_\pi
    \\
    (p,q)&\xrightarrow{\varepsilon}(p,r)\wedge(r,q)
    &&\text{splits the loop}
    \\
    (p,q)&\xrightarrow{\da_i}(p',q')
    &&\text{if } p\xrightarrow{\da_i}p' \wedge q'\xrightarrow{\ua_i}q \text{ in } \mathcal{C}_\pi
    \\
    (p,q)&\xrightarrow{\ua_i}(p',q')
    &&\text{if } p\xrightarrow{\ua_i}p' \wedge q'\xrightarrow{\da_i}q \text{ in } \mathcal{C}_\pi
  \end{align*}
  With the example above, we have (at least) the following transitions:
  \begin{align*}
    (0,7)&\xrightarrow{\da}(1,6)\vee(1,7)
    & (1,7)&\xrightarrow{\da}(1,6)\vee(1,7)
    \\
    (0,3)&\xrightarrow{\da}(1,2)\vee(1,3)
    & (1,3)&\xrightarrow{\da}(1,2)\vee(1,3)
    \\
    (4,7)&\xrightarrow{\da}(5,6)\vee(5,7)
    & (5,7)&\xrightarrow{\da}(5,6)\vee(5,7)
    \\
    (2,5)&\xrightarrow{\ua}(3,4)\vee(3,5)
    & (3,5)&\xrightarrow{\ua}(3,4)\vee(3,5)
    \\
    (1,2)&\xrightarrow{a}\mathsf{accept} 
    & (3,4)&\xrightarrow{b}\mathsf{accept} 
    & (5,6)&\xrightarrow{c}\mathsf{accept} 
    \\
    (0,7)&\xrightarrow{\varepsilon}(0,3)\wedge(3,7)
    & (1,7)&\xrightarrow{\varepsilon}(1,3)\wedge(3,7)
    \\
    (3,7)&\xrightarrow{\varepsilon}(3,4)\wedge(4,7)
    & (3,7)&\xrightarrow{\varepsilon}(3,5)\wedge(5,7)
  \end{align*}

  Next, we construct a non-deterministic tree automaton equivalent to the \ATWA
  associated with $\Phi$.  This induces an exponential blow-up.
  Hence, the resulting automaton is of size $2^{\poly(|\Phi|)}$.  
  
  If $\Phi$ is in \ICPDL then the construction of the corresponding \ATWA is
  exponential, resulting in $\mathcal{B}_\Phi$ of double exponential size.
\end{proof}

Since tree automata are closed under projections and can be used to check the cardinality 
constraints, the above result can be generalized to \EQLCPDL.

\begin{corollary}\label{cor:eqpdl-tree}
  Given a sentence $\Psi\in\EQLCPDL(\AP,\da_0,\da_1)$ over binary trees, we can
  construct a tree automaton $\mathcal{B}_\Psi$ of size
  $2^{\poly(|\Psi|)}$ such that, for every binary tree $\tau$, we have
  $$
  \tau\models\Psi \textup{~~~iff~~~} \tau\in\Lang(\mathcal{B}_\Psi) \,.
  $$
  Moreover, if $\Psi\in\EQICPDL(\AP,\da_0,\da_1)$ then we can construct an 
  equivalent $\mathcal{B}_\Psi$ of double exponential size.
\end{corollary}

\begin{proof}
  Let $\Psi=\exists p_1,\ldots,p_n~\chi\wedge\Phi$ be the given formula with 
  $\Phi\in\LCPDL(\AP',\da_0,\da_1)$ and $\chi$ the cardinality constraint. 
  Recall that $\AP'=\AP\uplus\{p_1,\ldots,p_n\}$.
  
  For each $1\leq i\leq n$, let $M_i$ be the largest constant $c$ occuring in some 
  constraint $\#p_i\leq c$, with $M_i=-1$ if there are no cardinality constraints on $p_i$.
  We construct a deterministic bottom-up tree automaton $\mathcal{B}_\chi$:
  \begin{description}
    \item[States]  $Q=\prod_{1\leq i\leq n}[1+M_i]$ where $[N]=\{0,1,\ldots,N\}$.
    
    A state is a tuple $\overline{m}=(m_1,\ldots,m_n)$.  It is accepting if
    $\overline{m}\models\chi$, i.e., $\chi$ evaluates to true when, for each $1\leq i\leq
    n$, $\#p_i$ is replaced with $m_i$.
  
    \item[Transitions] Let $a\in 2^{\AP'}$ be a possible node label of the tree. The 
    transitions are described below, where $\oplus_i$ denotes addition with threshold 
    $1+M_i$: $m \oplus_i m' = \min(m+m',1+M_i)$.
    \begin{description}
      \item[Leaves]  $(a,\overline{m})$ is a transition 
      if for all $1\leq i\leq n$ we have 
      $$m_i=
      \begin{cases}
        0 \oplus_i 1 & \text{if } p_i\in a \\
        0 & \text{otherwise.}
      \end{cases}$$
    
      \item[Unary]  $(\overline{m}',a,\overline{m})$ is a transition 
      if for all $1\leq i\leq n$ we have
      $$m_i=
      \begin{cases}
        m'_i \oplus_i 1 & \text{if } p_i\in a \\
        m'_i & \text{otherwise.}
      \end{cases}$$
    
      \item[Binary]  $(\overline{m}',\overline{m}'',a,\overline{m})$ is a transition 
      if for all $1\leq i\leq n$ we have 
      $$m_i=
      \begin{cases}
        m'_i \oplus_i m''_i \oplus_i 1 & \text{if } p_i\in a \\
        m'_i \oplus_i m''_i & \text{otherwise.}
      \end{cases}$$
    \end{description}
  \end{description}
  It is easy to see that a tree $\tau$ over $\AP'$ satisfies $\chi$ iff it is accepted by 
  $\mathcal{B}_\chi$. If the constants in $\chi$ are encoded in binary, then the size of 
  $\mathcal{B}_\chi$ is $2^{\mathcal{O}(|\AP'|+|\chi|)}$.
  
  Now, consider the tree automaton $\mathcal{B}_\Phi$ given by 
  Theorem~\ref{thm:pdl-tree}. We define $\mathcal{B}_\Psi$ as the projection over $\AP$ 
  of the intersection $\mathcal{B}_\chi \cap \mathcal{B}_\Phi$. It is of size 
  $2^{\poly(|\Psi|)}$ (assuming that all atomic propositions in $\AP$ occur in $\Phi$).
\end{proof}

\begin{corollary}\label{cor:Akstwpdl}
  Given a sentence $\Psi\in\EQLCPDL(\AP,\Gamma)$ and $k>0$, we can
  construct a tree automaton \AkstwPsi of size $2^{\poly(k,|\Psi|)}$
  such that, for every valid \kSTT $\tau$ with $\sem{\tau}=(G,\chi)$, we have
  $$
  G\models\Psi \textup{~~~iff~~~} \tau\in\Lang(\AkstwPsi) \,.
  $$
  Moreover, if $\Psi\in\EQICPDL(\AP,\Gamma)$ then we can construct an 
  equivalent \AkstwPsi of double exponential size.
\end{corollary}

\begin{proof}
  Given $\Psi\in\EQLCPDL(\AP,\Gamma)$ and $k>0$, we construct using
  Proposition~\ref{prop:stt-pdl-interpretation} the corresponding formula
  $\widetilde{\Psi}^k$ of size $\mathcal{O}(k^2|\Psi|)$.  Then, we apply
  Corollary~\ref{cor:eqpdl-tree} to construct the automaton
  $\AkstwPsi=\mathcal{B}_{\widetilde{\Psi}^k}$ of size $2^{\poly(k,|\Psi|)}$.
\end{proof}

\newpage

\section{\STW-\EQICPDL model checking}
\newcommand{\AkswPhi}{\ensuremath{\A_\Phi^{k\text{-}\mathsf{sw}}}\xspace}
\newcommand{\AkstwnotPhi}{\ensuremath{\A_{\neg\Phi}^{k\text{-}\mathsf{stw}}}\xspace}

We are now in a position to give `efficient' decision procedures to solve the 
model-checking problem of \CPDSs against \PDL specifications when we restrict to 
behaviours of bounded special tree-width.
For a given architecture $\Arch$ and alphabet $\Sigma$, we consider the following problem:

\begin{center}
\begin{tabular}{ll}
\toprule
\textsc{stw}-\textsf{(EQ-)(ILC)}\PDL-\textsc{ModelChecking}$(\Arch,\Sigma)$:\\
\midrule
{Instance:} & \hspace{-14em}$\Sys \in \CPDS(\Arch,\Sigma)$\,; 
$\Phi\in\textsf{(EQ-)(ILC)PDL}(\Arch,\Sigma)$; $k>0$\\[0.5ex]
{Question:} & \hspace{-14em}$L(\Sys) \cap \kstwCBM \subseteq L(\Phi)$\,?\\
\bottomrule
\end{tabular}
\end{center}

Here, we suppose that $k$ is given in unary.

Decidability of this problem follows from decidability of \MSO and the statement
of Exercise~\ref{exc:ICPDLtoMSO}, saying that every \ICPDL formula can be
(effectively) translated into an \MSO sentence.  Unfortunately, this does not
give us an elementary upper bound.  Instead, we will use
Corollary~\ref{cor:Akstwpdl} in order to obtain the following result.

\begin{mytheorem}\label{thm:icpdl-mc}
  The bound $k$ on the special tree-width is given in unary.

  \medskip\centering
  \begin{tabular}{|l|l|}
    \hline
    Problem & can be solved in \\
    \hline
    \textsc{stw}-\LCPDL-\textsc{ModelChecking}$(\Arch,\Sigma)$ & \textsc{ExpTime}  \\
    \hline
    \textsc{stw}-\EQLCPDL-\textsc{ModelChecking}$(\Arch,\Sigma)$ & $2$-\textsc{ExpTime}  \\
    \hline
    \textsc{stw}-\ICPDL-\textsc{ModelChecking}$(\Arch,\Sigma)$ & $2$-\textsc{ExpTime}  \\
    \hline
    \textsc{stw}-\EQICPDL-\textsc{ModelChecking}$(\Arch,\Sigma)$ & $3$-\textsc{ExpTime}  \\
    \hline
  \end{tabular}
\end{mytheorem}

This result for \LCPDL and \ICPDL
was proved in \cite{CGN-atva14} using split-width instead of special tree-width.

\begin{proof}
  The first ingredient is a tree automaton \Akstwcbm which accepts all binary trees $\tau$
  such that $\tau$ is a valid \kSTT which defines a \CBM. This was already established by 
  Proposition~\ref{prop:Akstwcbm} with a rather involved direct construction. Here 
  instead, we give a much simpler proof by taking advantage of \EQLCPDL and the general 
  construction of Corollary~\ref{cor:Akstwpdl}.
  
  Let $\AP=\Sigma\uplus\Procs$ and $\Gamma=\{\procrel\}\cup\DS$.  Consider a
  $(2^\AP,\Gamma)$-labeled graph $G=(V,(E_\gamma)_{\gamma\in\Gamma},\lambda)$.  By
  Lemma~\ref{lem:cbm-pdl-nojoin}, $G$ is a \CBM iff it satisfies 
  $\mathsf{ALMOSTCBM}\wedge\mathsf{NOJOIN}$.
  Notice that the \EQLCPDL sentence is of size 
  $\mathcal{O}(|\Procs|^{2}+|\DS|^{2}+|\Sigma|^{2})$.
  We construct a deterministic bottom-up tree automaton $\Akstw_\mathsf{NOJOIN}$ with
  $2^{\mathcal{O}(k^{2})}$ states which accepts all binary trees that are valid \kSTTs and
  such that the associated graph satisfies $\mathsf{NOJOIN}$. This is the only tree 
  automaton that we need to construct directly.

  A state of $\Akstw_\mathsf{NOJOIN}$ is a partial map $\beta\colon[k]\to\{\false,\true\}$.  
  When reading a \kSTT $\tau$ with $\sem{\tau}=(G,\chi)$ the deterministic bottom-up
  automaton will reach the state $\beta$ satisfying the following two conditions:
  \begin{enumerate}[nosep,label=($\mathsf{C}_{\arabic*}$),ref=$\mathsf{C}_{\arabic*}$]
    \item\label{item:C1} $\dom(\beta)=\dom(\chi)\subseteq[k]$ is the set of active 
    colors, 

    \item\label{item:C2} $G\models\mathsf{NOJOIN}$ and for each active color $j$ we have 
    $\beta(j)=\true$ iff $\chi(j)$ is the target of a $\procrel$ edge in $G$.
  \end{enumerate}
  The transitions are defined below.
  
  \begin{tabular}{|c|p{110mm}|}
    \hline
    $\bot \xrightarrow{i}\beta$ 
    &
    where $\dom(\beta)=\{i\}$ and $\beta(i)=\false$.
    \\ \hline
    $\beta\xrightarrow{\add{i}{a}}\beta$ 
    &
    if $a\in\AP$ (and $i\in\dom(\beta)$).
    \\ \hline
    $\beta\xrightarrow{\add{i,j}{d}}\beta$ 
    &
    if $d\in\DS$ (and $i,j\in\dom(\beta)$, $i\neq j$).
    \\ \hline
    $\beta'\xrightarrow{\add{i,j}{\procrel}}\beta$ 
    &
    if $\beta'(j)\neq\true$ (and $i,j\in\dom(\beta')$, $i\neq j$),
    then, $\beta=\beta'[j\mapsto\true]$.
    \\ \hline
    $\beta'\xrightarrow{\forget{i}}\beta$ 
    &
    if ($i\in\dom(\beta')$ and) $\beta$ is the restriction of $\beta'$ to 
    $\dom(\beta')\setminus\{i\}$.
    \\ \hline
    $\beta',\beta''\xrightarrow{\oplus}\beta$
    &
    if $\dom(\beta')\cap\dom(\beta'')=\emptyset$ (active colors should be disjoint) and 
    $\beta=\beta'\uplus\beta''$.
    \\ \hline
  \end{tabular}
  
  We can easily check that $\Akstw_\mathsf{NOJOIN}$ has a run on a binary tree $\tau$
  iff $\tau$ is a valid \kSTT and $G_\tau\models\mathsf{NOJOIN}$.  
  The number of states of $\Akstw_\mathsf{NOJOIN}$ is at most $3^{k+1}$.
  
  Note that $\Akstw_\mathsf{NOJOIN}$ does not accept \kSTTs adding twice the same 
  $\procrel$ edge, such as for instance $\add{0,1}{\procrel}\add{0,1}{\procrel}(0\oplus1)$, 
  which is actually a valid \kSTT. This can be fixed using as states maps 
  $\beta\colon[k]\to\{\false,\true\}\cup[k]$ such that as before
  $\dom(\beta)=\dom(\chi)\subseteq[k]$ is the set of active colors, and for 
  $j\in\dom(\beta)$ we have $\beta(j)=\false$ if $\chi(j)$ is not the target of a 
  $\procrel$ edge in $G$,  $\beta(j)=\true$ if $e\procrel\chi(j)$ in $G$ but $e$ is no 
  more in the image of $\chi$, and $\beta(j)=i$ if $e\procrel\chi(j)$ in $G$ and 
  $e=\chi(i)$.

  Now, by Corollary~\ref{cor:Akstwpdl}, we obtain a tree automaton
  $\Akstw_\mathsf{ALMOSTCBM}$ of size $2^{\poly(k,|\Arch|,|\Sigma|)}$ which is equivalent
  to the \EQLCPDL sentence $\mathsf{ALMOSTCBM}$.
  Finally, the automaton \Akstwcbm is obtained as the intersection of
  $\Akstw_\mathsf{NOJOIN}$ and $\Akstw_\mathsf{ALMOSTCBM}$.
  
  From Corollary~\ref{cor:Akstwpdl}, we construct the tree automata \AkstwnotPhi of size
  $2^{\poly(k,|\Phi|)}$ when $\Phi\in\LCPDL$.
  Now, using Theorem~\ref{thm:CPDStoPDL} and Corollary~\ref{cor:Akstwpdl}, we construct
  the tree automaton \Akstwsys of size $2^{\poly(k,|\Sys|)}$ (we could use also
  Proposition~\ref{prop:Akstwsys} to construct \Akstwsys with a size only polynomial in
  $|\Sys|$).  We obtain
  $$
  \Lang(\Sys) \cap \kstwCBM \subseteq \Lang(\Phi) \text{~iff~} 
  \Lang(\Akstwcbm\cap\Akstwsys\cap\AkstwnotPhi) = \emptyset
  $$
  and we conclude since emptiness for non-deterministic tree automata can be checked in
  linear time~\cite{TATA}.

  Notice that, in order to model-check a system against a specification $\Psi$ in \EQLCPDL
  (instead of \LCPDL) with the above method, we need an automaton $\Akstw_{\neg\Psi}$.
  Since $\neg\Psi$ is not in \EQLCPDL (negation in front of the existentially quantified
  propositions), we cannot use Corollary~\ref{cor:Akstwpdl} directly.  Instead, we
  construct the automaton $\Akstw_\Psi$ using Corollary~\ref{cor:Akstwpdl} and check if
  $\Lang(\Akstwcbm\cap\Akstwsys)\subseteq\Lang(\Akstw_\Psi)$.  The inclusion problem for
  non-deterministic tree automata is \textsc{ExpTime}-complete~\cite{TATA}. The 
  2-\textsc{ExpTime} complexity for \EQLCPDL model-checking follows.
\end{proof}

\newpage

\section{Concrete Underapproximation Classes and Special Tree-Width}\label{sec:concrete}

So far, we considered the following (decidable) version of the model-checking
problem for \CPDSs: Given a \txtCPDS $\Sys$, a sentence $\Phi$ in \MSO or 
\EQICPDL, and $k>0$, do we have
\[L(\Sys) \cap \kstwCBM \subseteq L(\Phi)\,?\]

We will now study some other, more ``concrete'' families $\C=(\C_k)_{k \ge 0}$ that are
\begin{itemize}
\item monotone ($\C_k \subseteq \C_{k+1}$ for all $k \ge 0$),
\item complete ($\bigcup_{k \ge 0} \C_k = \CBM$), and
\item decidable (the usual decision problems are decidable when the domain of \CBMs is restricted to $\C_k$).
\end{itemize}

In particular, the following model-checking problems for $\C$ should be decidable:
\begin{center}
\begin{tabular}{ll}
\toprule
$\C$-\MSO-\textsc{ModelChecking}$(\Arch,\Sigma)$:\\
\midrule
{Instance:} & \hspace{-10em}$\Sys \in \CPDS(\Arch,\Sigma)$\,; $\Phi \in \MSO(\Arch,\Sigma)$; $k \ge 0$\\[0.5ex]
{Question:} & \hspace{-10em}$L(\Sys) \cap \C_k \subseteq L(\Phi)$\,?\\
\bottomrule
\end{tabular}
\end{center}

\begin{center}
\begin{tabular}{ll}
\toprule
$\C$-\textsf{(EQ-)(ILC)}\PDL-\textsc{ModelChecking}$(\Arch,\Sigma)$:\\
\midrule
{Instance:} & \hspace{-10em}$\Sys \in \CPDS(\Arch,\Sigma)$\,; 
$\Phi\in \textsf{(EQ-)(ILC)}\PDL(\Arch,\Sigma)$; $k \ge 0$\\[0.5ex]
{Question:} & \hspace{-10em}$L(\Sys) \cap \C_k \subseteq L(\Phi)$\,?\\
\bottomrule
\end{tabular}
\end{center}

Next, we argue that, to show decidability of the above problem, we can make use
of the previous results on special tree-width.


Consider a family $\C=(\C_k)_{k \ge 0}$ such that the following hold, for all $k \ge 0$:
\begin{itemize}
\item[(1)] there is $k' \ge 0$ such that $\C_k \subseteq \stwCBM{{k'}}$ (and $k'$ is ``easily'' computable),
\item[(2)] 
\begin{itemize}
  \item[(a)] there is $\Phi_k \in \MSO$ such that $L(\Phi_k) = \C_k$, or

  \item[(b)] there is $\Phi_k\in\textsf{(EQ-)(ILC)}\PDL$ such that $L(\Phi_k)\cap\CBM=\C_k$.
  
  By Theorem~\ref{thm:CPDStoPDL}, this is the case if 
  there is $\Sys_k \in \CPDS$ with $L(\Sys_k) = \C_k$.
\end{itemize}
\end{itemize}
Then, we have
$$
\{\tau \mid \tau\in k'\text{-}\STT \text{ and } G_\tau\in\C_k \}
=\Lang(\A_\cbm^{k'\text{-}\mathsf{stw}} \cap \A_{\Phi_k}^{k'\text{-}\mathsf{stw}}) \,.
$$
We deduce that the \MSO or \textsf{(EQ-)(ILC)}\PDL model-checking problems are decidable
due to the following equivalences:
\begin{align*}
  \Lang(\Sys) \cap \C_k \subseteq \Lang(\Phi) 
  &\textup{~~~iff~~~} 
  \Lang(\A_\cbm^{k'\text{-}\mathsf{stw}} \cap \A_{\Phi_k}^{k'\text{-}\mathsf{stw}}
  \cap \A_\Sys^{k'\text{-}\mathsf{stw}} ) \subseteq 
  \Lang(\A_{\Phi}^{k'\text{-}\mathsf{stw}}) 
  \\
  &\textup{~~~iff~~~} 
  \Lang(\A_\cbm^{k'\text{-}\mathsf{stw}} \cap \A_{\Phi_k}^{k'\text{-}\mathsf{stw}}
  \cap \A_\Sys^{k'\text{-}\mathsf{stw}} \cap 
  \A_{\neg\Phi}^{k'\text{-}\mathsf{stw}}) = \emptyset
\end{align*}

Depending on the size of $\Phi_k\in\textsf{(EQ-)(ILC)}\PDL$ we get an upper-bound of
\textsc{ExpTime}, 2\textsc{ExpTime} or 3\textsc{ExpTime} complexity for the
\textsf{(EQ-)(ILC)}\PDL-\textsc{ModelChecking}.

\subsection{Context-Bounded \MNWs}

In the following, we assume that the architecture $\Arch$ satisfies $|\Procs|=1$ and $\DS=\Stacks$. Actually, many underapproximation classes have been defined for this setting of \emph{multiply nested words} (\MNWs).

In the first class that we consider, we restrict the number of contexts.  In
each context, only one stack can be accessed.

\begin{definition}\label{def:context}
  Let $\mscn=(a_1 \ldots a_n,(\matchrel^d)_{d\in\DS}) \in \CBM$.  Note that
  $\Events = \{1,\ldots,n\}$.  A \emph{context} of $\mscn$ is a possibly empty
  interval $I = \{e,e+1,\ldots,f\}$, for some $e,f \in \Events$, such that, for
  all $d,d' \in \DS$, $(i,j) \in {\matchrel^d}$, and $(i',j') \in
  {\matchrel^{d'}}$, the following holds:
  \[
  \left.
  \begin{array}{rl}
    & I \cap \{i,j\} \neq \emptyset\\[1ex]
    \wedge & I \cap \{i',j'\} \neq \emptyset
  \end{array}
  \right\} ~~\Rightarrow~~ d = d'
  \]
\end{definition}

\begin{definition}[\cite{QadeerR05}]\label{def:cntbounded}
  For $k \ge 0$, we call $\mscn$ $k$-\emph{context-bounded} if there are
  contexts $I_1,\ldots,I_k$ of $\mscn$ such that
  $\Events = I_1 \cup \ldots \cup I_k$.
\end{definition}

\begin{theorem}[\cite{QadeerR05}]
  Non-emptiness (reachability) of multipushdown systems restricted to
  $k$-context-bounded is decidable in NP.
\end{theorem}

The set of $k$-context-bounded \MNWs (over the fixed architecture) is denoted by $\Context{k}$.
Moreover, we let $\AllContext = (\Context{k})_{k \ge 0}$.

\begin{example}
Consider the MNW below, over two stacks and a singleton set $\Sigma$ (so that we omit its letters).
\begin{center}
\fbox{
\includegraphics[scale=0.5]{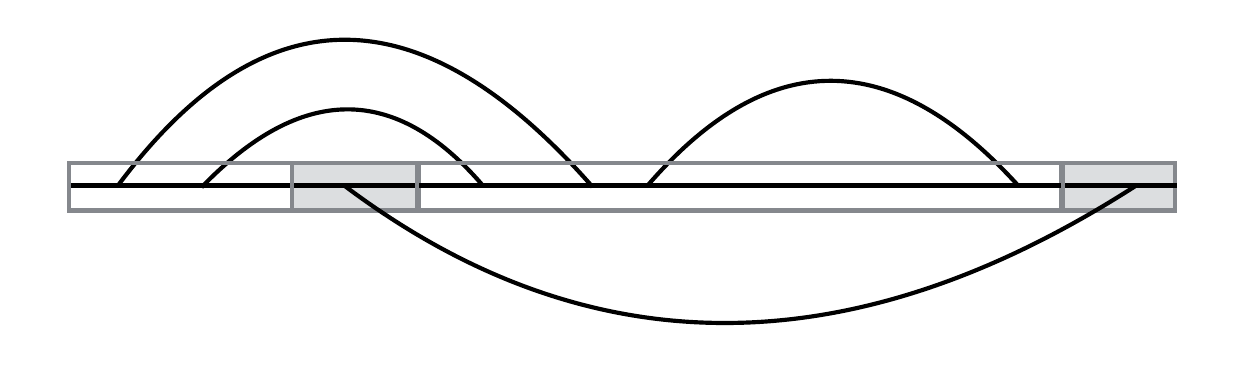}
}
\end{center}
The curved edges above the horizontal line stand for one of the stacks, the curved edge below it represents the other stack. The \MNW is $4$-context-bounded, but not $3$-context-bounded.
\end{example}

\begin{lemma}
  For $k\ge 2$, we have $\Context{k}\subseteq\stwCBM{(2k-1)}$.
\end{lemma}

\begin{proof}
  See Section~\ref{sec:game-stw}.
\end{proof}

\begin{lemma}\label{lem:context-pdl}
  For all $k>0$, there is $\Phi^k_{\mathsf{context}}\in\CPDL$ of size
  $\mathcal{O}(k|\DS|^2)$ such that $L(\Phi^k_{\mathsf{context}})=\Context{k}$.
\end{lemma}

\begin{proof}
  We introduce some macros.  For $d\in\DS$, let
  $\mathsf{RW}_d=\existspath{{\matchrel}^d+({\matchrel}^d)^{-1}}$ be the state
  formula characterizing events accessing the data structure $d$.  We define
  $$
  \Phi^k_{\mathsf{context}} = \neg\Existsev \Big\langle \Big( \sum_{d\neq d'}
  \test{\mathsf{RW}_d} \cdot {\procrel}^{+} \cdot \test{\mathsf{RW}_{d'}} \Big)^{k} \Big\rangle 
  $$
  Notice that the path formula 
  $\test{\mathsf{RW}_d}\cdot{\procrel}^+\cdot\test{\mathsf{RW}_{d'}}$ with 
  $d\neq d'$ ensures a change of context.
\end{proof}

\begin{corollary}\label{cor:context-MC}
  For any architecture $\Arch$ such that $|\Procs|=1$ and $\DS=\Stacks$, the 
  $\AllContext${\sc{-ModelChecking}}$(\Arch)$ problem is decidable for \MSO or 
  \EQICPDL specifications. The problem is in \textsc{ExpTime} for \LCPDL specifications,
  in 2\textsc{ExpTime} for \ICPDL or \EQLCPDL specifications, and 
  in 3\textsc{ExpTime} for \EQICPDL.
\end{corollary}

\begin{proof}
  Let $k'=\max(3,2k-1)$.  The $\AllContext${\sc{-ModelChecking}}$(\Arch)$
  problem for a \CPDS \Sys and a specification $\Phi$ in \textsf{(ILC)}\PDL reduces to the
  emptiness problem for the tree automaton
  $\A_\cbm^{k'\text{-}\mathsf{stw}} 
  \cap \A_{\Phi^k_{\mathsf{context}}}^{k'\text{-}\mathsf{stw}}
  \cap \A_\Sys^{k'\text{-}\mathsf{stw}} 
  \cap \A_{\neg\Phi}^{k'\text{-}\mathsf{stw}}$, 
  and for $\Phi$ in \textsf{EQ-(ILC)}\PDL to the inclusion problem
  $\Lang( \A_\cbm^{k'\text{-}\mathsf{stw}} 
  \cap \A_{\Phi^k_{\mathsf{context}}}^{k'\text{-}\mathsf{stw}}
  \cap \A_\Sys^{k'\text{-}\mathsf{stw}} ) \subseteq
  \Lang( \A_{\Phi}^{k'\text{-}\mathsf{stw}} )$.
\end{proof}

We can also directly construct a \txtCPDS for the language $\Context{k}$:

\begin{lemma}
  For all $k>0$, there is $\Sys_k \in \CPDS$ such that $L(\Sys_k)=\Context{k}$.
\end{lemma}

\begin{proof}
  The idea is simple: The set of locations being
  $\{\inLoc\}\cup(\DS\times\{1,\ldots,k\})$ and $\Val$ a singleton set, one
  keeps track of the current data structure and context number.  The \CPDS stays
  in state $\inLoc$ while reading a prefix of internal events.
\end{proof}

\begin{lemma}
  For all $k \ge 0$, there is $\phi_k \in \MSO$ such that $L(\phi_k) = \Context{k}$.
\end{lemma}

\begin{proof}
We define a formula $\mathit{cont}_k(x,y)$ that says that, assuming $x \le y$, the events $x$ and $y$ are in the scope of at most $k$ contexts. It says that there are no $k+1$ events between $x$ and $y$ that are in distinct contexts:
\[
\mathit{cont}_k(x,y) ~=~ \neg \exists x_1,\ldots,x_{k+1}
\left(
\begin{array}{rl}
& x \le x_1 < \ldots < x_{k+1} \le y\\[1ex]
\wedge & \displaystyle\bigwedge_{1 \le i \le k} ~\bigvee_{d \neq d'} \mathit{stack}_d(x_i) \wedge \mathit{stack}_{d'}(x_{i+1})
\end{array}
\right)
\]
where $\mathit{stack}_d(x_i) = \exists z\,(x_i \matchrel^d z ~\vee~ z \matchrel^d x_i)$. With this, we set
\[\phi_k ~=~ \forall x\forall y\,\mathit{cont}_k(x,y)\,.\]
\end{proof}

\newpage

\subsection{Phase-Bounded \MNWs}

There is another well established notion for \MNWs, which relaxes the notion of a context:
\begin{definition}[\cite{Madhusudan2007}]\label{def:phase}
Let $\mscn=(a_1 \ldots a_n,(\matchrel^d)_{d\in\DS}) \in \CBM(\Arch,\Sigma)$. A \emph{phase} of $\mscn$ is a set $I = \{e,e+1,\ldots,f\}$, for some $e,f \in \Events$, such that, for all $d,d' \in \DS$, $(i,j) \in {\matchrel^d}$, and $(i',j') \in {\matchrel^{d'}}$, the following holds:
\[
\left.
\begin{array}{rl}
& I \cap \{j\} \neq \emptyset\\[1ex]
\wedge & I \cap \{j'\} \neq \emptyset
\end{array}
\right\} ~~\Rightarrow~~ d = d'
\]
\end{definition}

\begin{theorem}[\cite{Madhusudan2007}]
  Non-emptiness (reachability) of multipushdown systems restricted to
  $k$-phase-bounded is decidable in 2\textsc{ExpTime}.
\end{theorem}

The special tree-width of $k$-\emph{phase-bounded} \MNWs (i.e., those \MNWs that
can be split into at most $k$ phases) is at most $k'=2^{2k}$ \cite{CGN-concur12}.

\begin{lemma}\label{lem:phase-pdl}
  For all $k>0$, there is $\Phi^k_{\mathsf{phase}}\in\CPDL$ of size
  $\mathcal{O}(k|\DS|^2)$ such that $L(\Phi^k_{\mathsf{phase}})=\Phase{k}$.
\end{lemma}

\begin{proof}
  The formula is obtained from $\Phi^k_{\mathsf{context}}$ by replacing $\mathsf{RW}_d$ 
  with $\mathsf{R_d}=\existspath{({\matchrel}^d)^{-1}}$:
  $$
  \Phi^k_{\mathsf{phase}} = \neg\Existsev \Big\langle \Big( \sum_{d\neq d'}
  \test{\mathsf{R}_d} \cdot {\procrel}^{+} \cdot \test{\mathsf{R}_{d'}} \Big)^{k} \Big\rangle 
  $$
  Notice that the path formula 
  $\test{\mathsf{R}_d}\cdot{\procrel}^+\cdot\test{\mathsf{R}_{d'}}$ with 
  $d\neq d'$ ensures a change of phase.
\end{proof}

\begin{corollary}\label{cor:phase-MC}
  For any architecture $\Arch$ such that $|\Procs|=1$ and $\DS=\Stacks$, the 
  $\AllPhase${\sc{-ModelChecking}}$(\Arch)$ problem is decidable for \MSO or 
  \EQICPDL specifications. The problem is in 2\textsc{ExpTime} for \ICPDL.
\end{corollary}

\newpage

\subsection{Scope-Bounded Nested Words}

Next, we define a restriction that captures more behaviors than pure contexts.
We continue to assume $|\Procs|=1$ and $\DS=\Stacks$.

\begin{definition}[\cite{LaTorreN11}]\label{def:scopebounded}
For $k \ge 0$, we call an MNW $\mscn$ $k$-\emph{scope-bounded} if, for all $(e,f) \in {\matchrel}$, there are contexts $I_1,\ldots,I_k$ of $\mscn$ such that
\[\{e,e+1,\ldots,f\} = I_1 \cup \ldots \cup I_k\,.\]
\end{definition}

\begin{theorem}[\cite{LaTorreN11}]
  Non-emptiness (reachability) of multipushdown systems restricted to
  $k$-scope-bounded is decidable in \textsc{PSpace}.
\end{theorem}

The set of $k$-scope-bounded \MNWs is denoted by $\Scope{k}$.
Moreover, we let $\AllScope = (\Scope{k})_{k \ge 0}$.

\begin{example}\label{ex:scope}
The figure below illustrates a \txtCBM $\mscn$ with $\mscn \in \Scope{3}$. Note that $\mscn \in \Context{5} \setminus \Context{4}$.
\begin{center}
\fbox{
\includegraphics[scale=0.5]{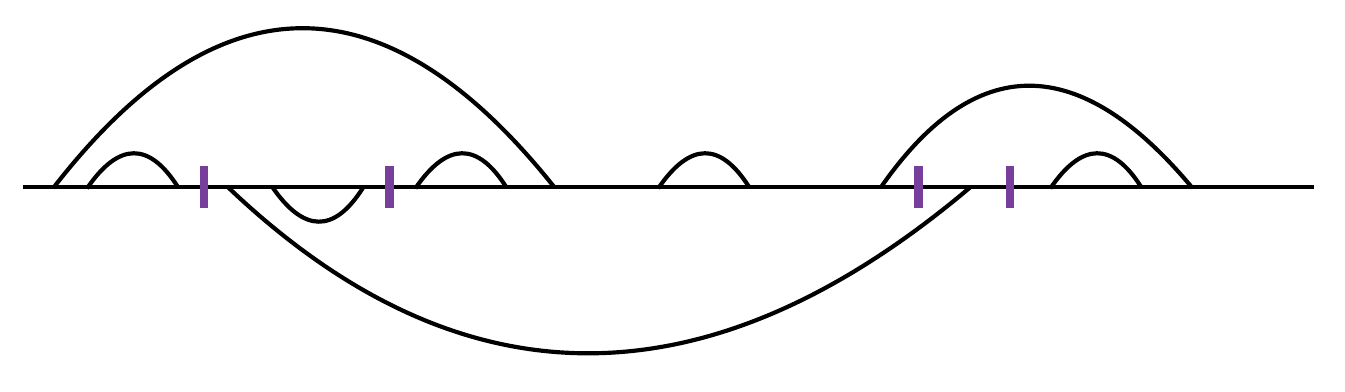}
}
\end{center}
Next, consider the set $L$ of \CBMs with an arbitrary number of alternating write-read edges, following the pattern below:
\begin{center}
\fbox{
\includegraphics[scale=0.5]{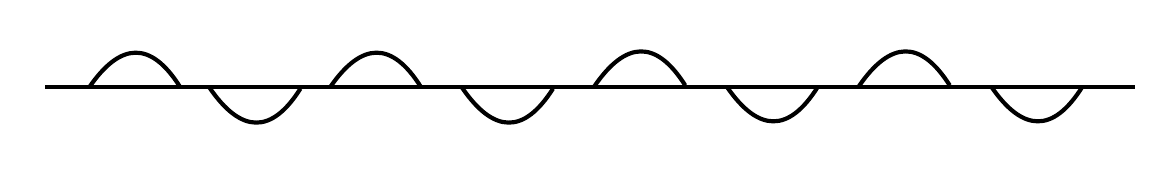}
}
\end{center}
Then, $L \subseteq \Scope{1}$ but $L \not\subseteq \Context{k}$ for all $k \ge 0$.
\end{example}

\begin{lemma}
  We have $\Scope{1} \subseteq \stwCBM{3}$ and $\Scope{k} \subseteq
  \stwCBM{(2k-1)}$ for $k \ge 2$.
\end{lemma}

\begin{proof}
  If $k=1$, then a CBM is the concatenation of (singly) nested words.  So,
  suppose $k \ge 2$.  Again, we show that Eve has a winning strategy in the
  split-game, using at most $2k$ colors.  To illustrate her strategy, we
  consider Figure~\ref{fig:scope-splits}, depicting the \txtCBM from
  Example~\ref{ex:scope}.  We omit internal events, which are easy to handle.
  \begin{figure}[ht]
    \begin{center}
      \fbox{\includegraphics[scale=0.4]{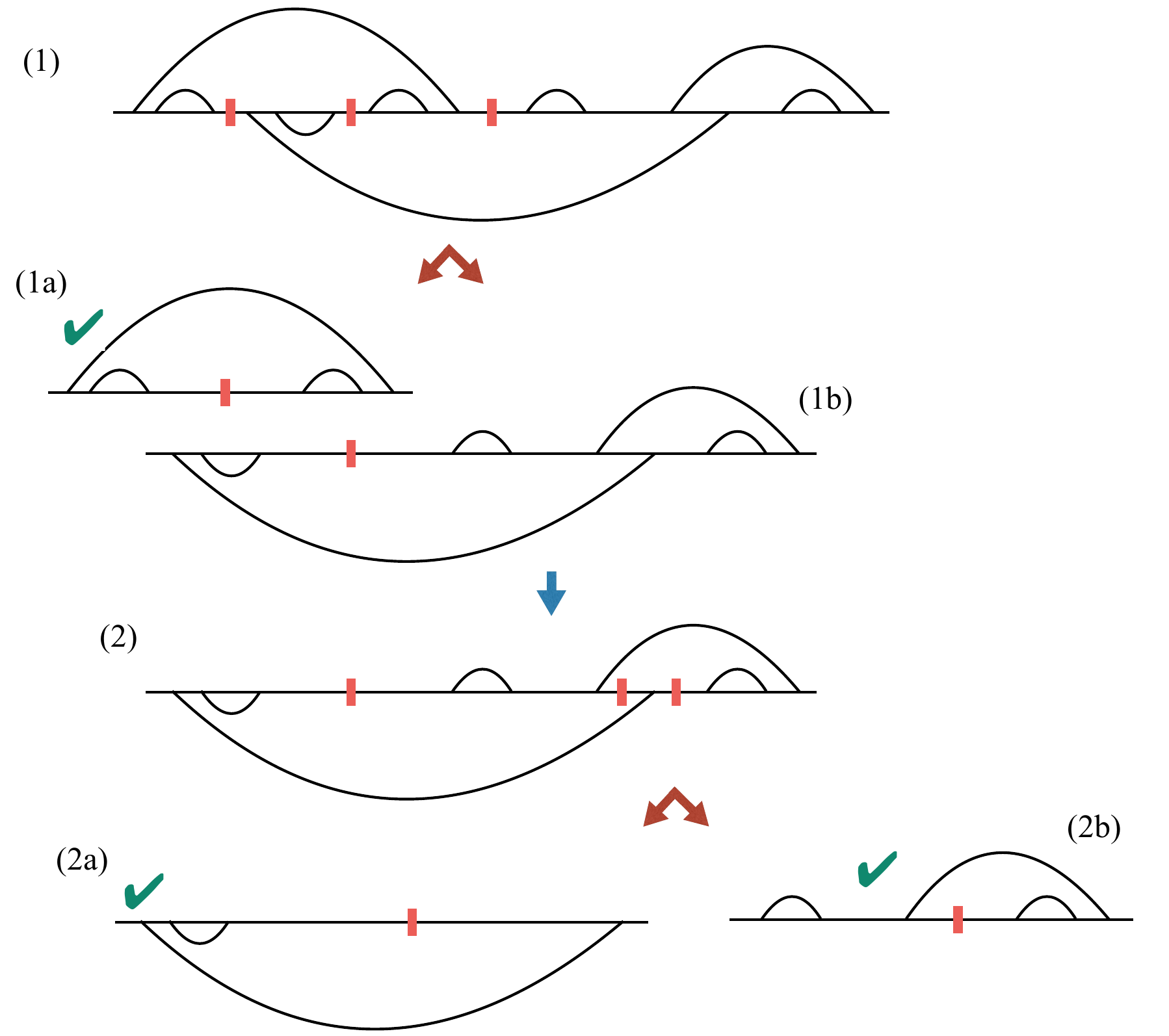}}
    \end{center}
    \caption{Split-game for scope-bounded \MNWs\label{fig:scope-splits}}
  \end{figure}
  \begin{itemize}
    \item Consider the \emph{leftmost} write (i.e., a push).  We split the
    process edge just behind the corresponding read (i.e., a pop).  Moreover, we
    divide the induced prefix into its contexts.  Since the \txtCBM is
    $k$-scope-bounded, this process requires at most $k$ split-edges (1), i.e., 
    at most $2k$ colors.

    \item One resulting component is a nested word (1a) with at most $2k-3$
    colors and where the last point is colored.  The nested word can be
    decomposed using three more colors, i.e., using a total of at most $2k$
    colors.

    \item We proceed with the other remaining component (1b), which possibly has
    already some split-edges.  Again, we look at the first call and its receive
    $f$, place a split-edge behind $f$, and divide the corresponding prefix into
    its contexts.  Though the prefix may already contain split-edges, we do not
    have more than $k$ split-edges after the division phase (2).  We proceed
    like in the 2nd item.
  \end{itemize}
  Altogether, Eve wins the split-game with at most $2k$ colors. 
\end{proof}

\begin{lemma}
  For all $k>0$, there is $\Phi^k_{\mathsf{scope}}\in\LCPDL$ of size
  $\mathcal{O}(k|\DS|^2)$ such that $L(\Phi^k_{\mathsf{scope}})=\Scope{k}$.
\end{lemma}

\begin{proof}
  Recall that the state formula
  $\mathsf{RW}_d=\existspath{{\matchrel}^d+({\matchrel}^d)^{-1}}$ characterizes
  events accessing the data structure $d\in\DS$.
  We define
  $$
  \Phi^k_{\mathsf{scope}} = \neg\Existsev  \mathsf{Loop}
  \Big\langle \Big(
  \sum_{d\neq d'} 
  \test{\mathsf{RW}_d}\cdot{\procrel}^+\cdot\test{\mathsf{RW}_{d'}}
  \Big)^{k}
  \cdot {\matchrel}^{-1} \Big\rangle 
  $$
  Notice that the path formula 
  $\test{\mathsf{RW}_d}\cdot{\procrel}^+\cdot\test{\mathsf{RW}_{d'}}$ with 
  $d\neq d'$ ensures a change of context.
\end{proof}

\begin{corollary}\label{cor:scope-MC}
  For any architecture $\Arch$ such that $|\Procs|=1$ and $\DS=\Stacks$, the 
  $\AllScope${\sc{-ModelChecking}}$(\Arch)$ problem is decidable for \MSO or 
  \ICPDL specifications. The problem is in \textsc{ExpTime} for \LCPDL and in 
  2\textsc{ExpTime} for \ICPDL.
\end{corollary}

\begin{lemma}
For all $k \ge 0$, there is $\phi_k \in \MSO$ such that $L(\phi_k) = \Scope{k}$.
\end{lemma}

\begin{proof}
According to the definition of $k$-scope-bounded words, it is enough to set
\[\phi_k ~=~ \forall x\forall y\,\displaystyle(x \matchrel y ~\Rightarrow~ \mathit{cont}_k(x,y))\,.\]
\end{proof}

Again, one can also directly construct a \txtCPDS:

\begin{lemma}
For all $k \ge 0$, there is $\Sys_k \in \CPDS$ such that $L(\Sys_k) = \Scope{k}$.
\end{lemma}

\begin{proof}
The idea is to employ a counter from $1$ to $k$ for \emph{each} stack $d \in \DS$, and to count the number of contexts in the scope of every \emph{outermost} write-read edge of $d$. Thus, we can do with set of locations $\{\inLoc,\ell_\textup{pref}\} \cup (\DS \times \{0,1,\ldots,k\}^\DS)$ and $\Val = \{\textup{o},\textup{i}\}$, where a pushed value $\textup{o}$ signals an \emph{outermost} nesting edge, and $\textup{i}$ an inner nesting edge.
The following figure shows how the two counters, one for each stack, work:
\begin{center}
\fbox{
\includegraphics[scale=0.5]{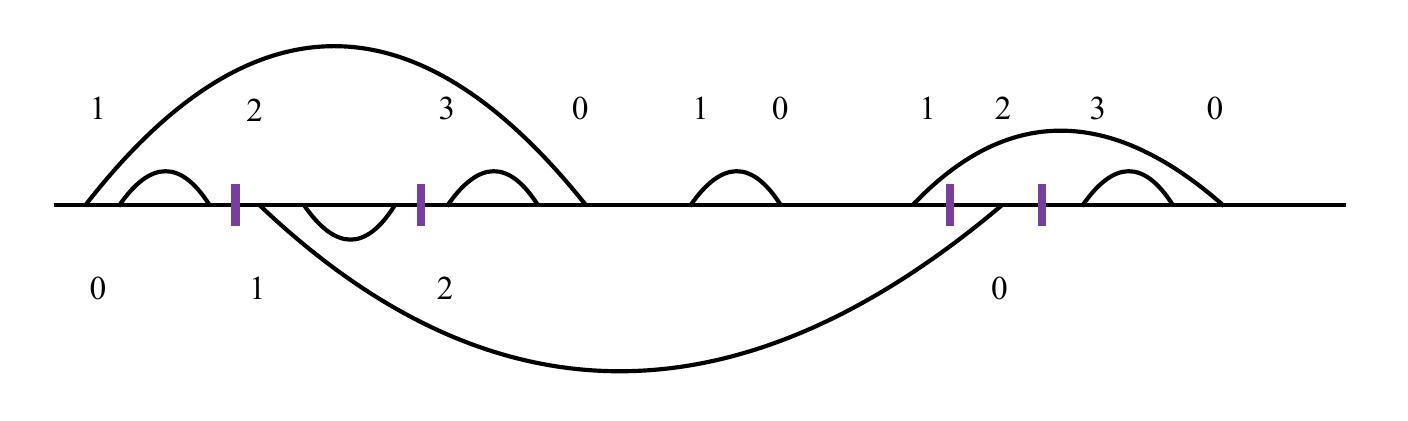}
}
\end{center}
Note that the number of locations is exponential in $|\DS|$.
\end{proof}
 
\newpage

\subsection{Existentially Bounded MSCs} 

We consider existentially bounded \CBMs.
The following definition is slightly different from the notion that we have already seen.
In fact, we now define a \emph{local} variant of existential bounds.

We assume that there are no local events and that $\Sigma$ is a singleton set
(and can therefore be omitted).  

\begin{definition}
  A \CBM $\mscn$ is $\exists k$\emph{-bounded} if it admits some linearization
  $<_\mathsf{lin}$ such that, at any time, there are no more than $k$ messages
  in \emph{each} data structure: for all $g\in\Events$ and all $d\in\DS$,
  $|\{(e,f)\in{\matchrel}^{d} \mid e \leq_\mathsf{lin} g <_\mathsf{lin} f\}| \leq k$.
\end{definition}

The set of $\exists k$-bounded \CBMs (over the given architecture) is denoted by $\ebMSCs{k}$.
Moreover, we let $\allebMSCs = (\ebMSCs{k})_{k \ge 0}$.

\begin{example}
Consider the MSC $\mscn$ below:
\begin{center}
\fbox{
\includegraphics[scale=0.5]{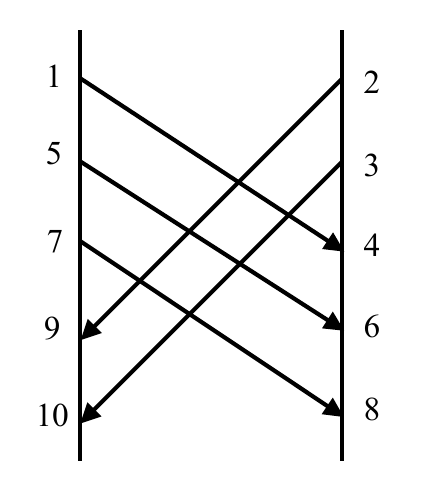}
}
\end{center}
The suggested linearization shows that $\mscn$ is $\exists 2$-bounded.
In fact, $\mscn \in \ebMSCs{2} \setminus \ebMSCs{1}$.
\end{example}

We have already seen (Exercise~\ref{ex:stw-exist-bounded} that $\exists
k$-bounded \CBMs have bounded special tree-width, though the bound on special
tree-width has to take into account that we defined a local variant of
existential bounds:

\begin{lemma}
  For all $k \ge 0$, we have $\ebMSCs{k} \subseteq \stwCBM{(k|\DS|+|\Procs|)}$.
\end{lemma}

\begin{proof}
  Let us quickly recall the proof by means of Figure~\ref{fig:ebmsc-splits}.
  \begin{figure}[ht]
    \begin{center}
      \fbox{
      \includegraphics[scale=0.4]{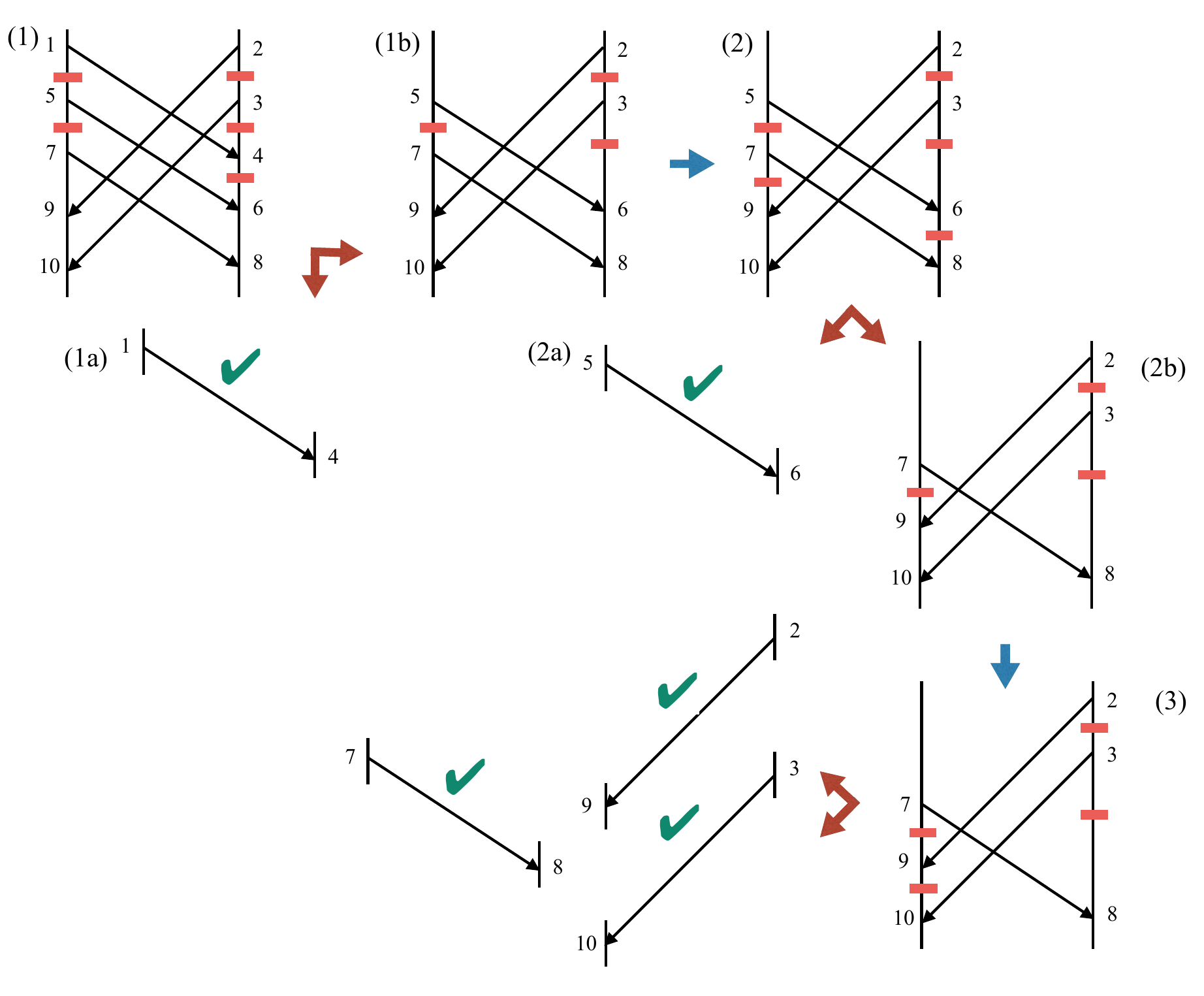}
      }
    \end{center}
    \caption{Split-game for existentially bounded \MSCs\label{fig:ebmsc-splits}}
  \end{figure}
  
  \begin{itemize}
    \item Eve's strategy is to choose a linearization and to cut the first
    $k|\DS|+1= 5$ events according to that linearization (1).  For this, she
    uses $k|\DS|+|\Procs|+1$ colors and she removes the $\procrel$-edges 
    touching the first $k|\DS|+1$ events.
    
    \item Then, there is at least one isolated message edge, which in this case
    is $(1,4)$.
    
    \item In the remaining component, we again cut some more events until we
    isolated $k|\DS|+1$ of them (2), and so on.
  \end{itemize}
  Thus, Eve wins the decomposition game with at most $k|\DS|+|\Procs|+1$ colors. 
\end{proof}

\begin{lemma}\label{lem:exist-bounded-pdl}
  For all $k>0$, there is $\Phi^k_{\exists\mathsf{B}}\in\LCPDL$ of size
  $\mathcal{O}(k|\DS|)$ such that
  $L(\Phi^k_{\exists\mathsf{B}}) \cap \CBM=\ebMSCs{k}$.
\end{lemma}

\begin{proof}
  For an \CBM $\mscn$, consider the binary relation ${\rightsquigarrow}_k
  \subseteq \Events \times \Events$ that connects events $(f,e)$ where, for some
  data structure $d\in\DS$ and some $i\geq1$:
  \begin{itemize}[nosep]
    \item $e$ is the $(i+k)$-th write on $d$
    \item $f$ is the $i$-th read from $d$
  \end{itemize}
  The relation ${\rightsquigarrow}_k$ is illustrated below for the cases $k=1$ (cyclic $\Rightarrow$ not $\exists 1$-bounded) and $k=2$ (acyclic $\Rightarrow$ $\exists 2$-bounded).
  \begin{center}
    \fbox{
    \includegraphics[scale=0.5]{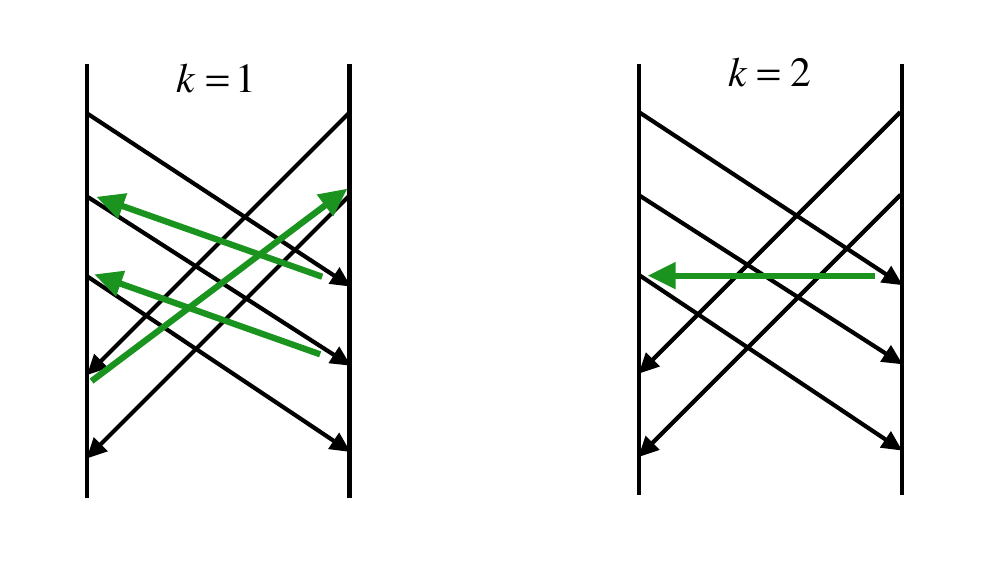}
    }
  \end{center}

  \begin{claim}[\cite{LohreyM04} for the special case $\DS=\Queues$] \mbox{ }
    \begin{center}
      $\mscn \in \ebMSCs{k}$ iff ${<} \cup {\rightsquigarrow}_k$ is acyclic.
    \end{center}
  \end{claim}
  
  $\Leftarrow$ Take $<_\mathsf{lin}$ any linearization of ${<} \cup 
  {\rightsquigarrow}_k$. We check that $<_\mathsf{lin}$ is $k$-bounded.
  
  $\Rightarrow$ Let $<_\mathsf{lin}$ be a $k$-bounded linearization of $<$. 
  We check that $f\rightsquigarrow e$ implies $f<_\mathsf{lin}e$.

  The binary relation ${\rightsquigarrow}_k$ can be formalized by a path 
  formula from \CPDL:
  $$
  {\rightsquigarrow}_k = \sum_{d\in\DS}({\matchrel}^d)^{-1}\cdot
  \Big( {\procrel}\cdot(\test{\neg\existspath{\matchrel^d}}
  \cdot{\procrel})^{*}\cdot\test{\existspath{\matchrel^d}} \Big)^k
  $$
  Acyclicity relies on the loop predicate:
  $\Phi^k_{\exists\mathsf{B}} =
  \neg\Existsev\existsloop{({\procrel}+{\matchrel}+{\rightsquigarrow}_k)^+}$.
\end{proof}

\begin{corollary}\label{cor:context-MC}
  For any architecture $\Arch$,
  the $\allebMSCs${\sc{-ModelChecking}}$(\Arch)$ problem is decidable for \MSO
  or \ICPDL specifications.  The problem is in \textsc{ExpTime} for \LCPDL and
  in 2\textsc{ExpTime} for \ICPDL.
\end{corollary}

\begin{lemma}
  For all $k \ge 0$, there is $\phi_k \in \MSO$ such that $L(\phi_k) = \ebMSCs{k}$.
\end{lemma}

For message passing systems, there is a \txtCPDS defining $\ebMSCs{k}$, too, but
the proof is much harder (and omitted here):

\begin{theorem}[\cite{GKM06}]
  For any architecture $\Arch$ such that $\DS=\Queues$, and for all $k \ge 0$,
  there is $\Sys_k \in \CPDS$ such that $L(\Sys_k) = \ebMSCs{k}$.
\end{theorem}

\newpage
\section{Synthesis from \ICPDL to \CPDS}\label{sec:expricpdl}

The synthesis problem is to construct an implementation $\Sys$ from a given 
specification $\Phi$. Here we are interested in specifications given in 
\ICPDL and distributed implementations, i.e., \CPDSs.

\newcommand{\cbmL}{L_{\textup{\textsf{cbm}}}}

In the following, given $\Phi \in \ICPDL(\Arch,\Sigma)$, we let $\cbmL(\Phi) := L(\Phi) \cap \CBM(\Arch,\Sigma)$.

Unfortunately,  \IPDL and \LCPDL are too expressive to be translatable into \CPDSs:

\begin{mytheorem}\label{thm:ICPDLtoCPDS}
Suppose $\Sigma=\{a,b,c\}$ and let $\Arch$ be given as follows:
\begin{center}
{
\includegraphics[scale=0.5]{simplearch.pdf}
}
\end{center}
There is $\Phi\in\IPDL(\Arch,\Sigma)$ such that $\Lang(\Sys) \neq \cbmL(\Phi)$,
for all $\Sys \in \CPDS(\Arch,\Sigma)$.

There is $\Phi\in\LCPDL(\Arch,\Sigma)$ such that $\Lang(\Sys) \neq \cbmL(\Phi)$,
for all $\Sys \in \CPDS(\Arch,\Sigma)$.
\end{mytheorem}

\begin{exercise}
Prove Theorem~\ref{thm:ICPDLtoCPDS} using the idea in the proof of Theorem~\ref{thm:grids}.
\end{exercise}

The exact relation between $\CPDL$ and $\CPDS$ is unknown. However, every \PDL formula can be translated into a \txtCPDS (the special case $\DS=\Queues$ was considered in \cite{BKM-lmcs10}):

\begin{mytheorem}[\cite{cyriac:hal-00943690}]\label{PDLtoCPDS}
For every $\Phi \in \PDL(\Arch,\Sigma)$, there is $\Sys \in \CPDS(\Arch,\Sigma)$ such that $\Lang(\Sys) = \cbmL(\Phi)$.
\end{mytheorem}

\begin{proof}
  For a state formula $\sigma$, we construct, inductively,
  \[\Sys_\sigma=(\Locs,\Val,(\to_p)_{p\in\Procs},\inLoc,\FinLocs) \in
  \CPDS(\Arch,\Sigma)\] 
  together with a mapping $\holds_\sigma\colon \Trans \to \{0,1\}$ such that
  \begin{itemize}[nosep]
    \item $\Lang(\Sys_\sigma) = \CBM(\Arch,\Sigma)$ and,

    \item for all $\mscn \in \CBM(\Arch,\Sigma)$, all accepting runs $\rho$ of $\Sys_\sigma$ on $\mscn$, and all events $e$ of $\mscn$, we have
    \[e \in \Sem{\sigma}{\mscn} ~~\textup{iff}~~\holds_\sigma(\rho(e))=1\,.\]
  \end{itemize}
  Here $\Trans=\biguplus_{p\in\Procs}{\to}_p$ and a run is a map 
  $\rho\colon\Events\to\Trans$.

  The pair $(\Sys_\sigma,\gamma_\sigma)$ is a \emph{transducer} from the input 
  alphabet $\Sigma$ to the output alphabet $\{0,1\}$.
  
  \newpage

  \underline{\txtCPDS $\Sys_a$ for $a \in \Sigma$:}
  
  We let $\Sys_a=(\Locs,\Val,(\to_p)_{p\in\Procs},\inLoc,\FinLocs)$ where
  \begin{itemize}
    \item $\Locs = \{1\}$~~~
    $\Val = \{v\}$~~~
    $\inLoc=1$~~~
    $\FinLocs = \{1\}^\Procs$,
    
    \item transitions and output function ($b \neq a$, $p\in\Procs$ and $d\in\DS$):
    \begin{itemize}
      \item $\gamma_a(1 \xrightarrow{{a}}_p 1) =
      \gamma_a(1 \xrightarrow{a,d!v}_{\writer(d)} 1) =
      \gamma_a(1 \xrightarrow{a,d?v}_{\reader(d)} 1) = 1$; 

      \item $\gamma_a(1 \xrightarrow{{b}}_p 1) =
      \gamma_a(1 \xrightarrow{b,d!v}_{\writer(d)} 1) =
      \gamma_a(1 \xrightarrow{b,d?v}_{\reader(d)} 1) = 0$.
    \end{itemize}
  \end{itemize}
  
  \bigskip

  \underline{\txtCPDS $\Sys_p$ for $p \in \Procs$:}

  As above, $\Sys_p=(\Locs,\Val,(\to_q)_{q\in\Procs},\inLoc,\FinLocs)$ is the
  universal \CPDS with one state. The output function is modified:
  $$
  \gamma_p(t)=
  \begin{cases}
    1 & \text{if } t\in{\to}_p \\
    0 & \text{otherwise.}
  \end{cases}
  $$
  
  \bigskip

  \underline{\txtCPDS $\Sys_{\neg\sigma}$:}

  Suppose $\Sys_{\sigma}=(\Locs,\Val,(\to_p)_{p\in\Procs},\inLoc,\FinLocs)$ with
  associated mapping $\holds_\sigma\colon\Trans\to \{0,1\}$.  We set
  $\Sys_{\neg\sigma}=\Sys_\sigma$ and let $\holds_{\neg\sigma}(t) = 1 -
  \holds_{\sigma}(t)$ for all $t\in\Trans$.

  \bigskip
  
  \underline{\txtCPDS $\Sys_{\sigma_1 \vee \sigma_2}$:}
  
  Suppose $\Sys_{\sigma_i}=(\Locs_i,\Val_i,(\to_p^i)_{p\in\Procs},\inLoc^i,\FinLocs_i)$, for $i \in \{1,2\}$.
  
  We construct the product
  $\Sys_{\sigma_1}\times\Sys_{\sigma_2}=(\Locs,\Val,(\to_p)_{p\in\Procs},\inLoc,\FinLocs)$
  as usual:
  \begin{itemize}[nosep]
    \item $\Locs = \Locs_1 \times \Locs_2$
    \item $\Val = \Val_1 \times \Val_2$
    \item $\inLoc=(\inLoc^1,\inLoc^2)$
    \item $\FinLocs = \{((\ell_1^p,\ell_2^p))_{p \in \Procs} \in \Locs^\Procs \mid (\ell_i^p)_{p \in \Procs} \in \FinLocs_i$ for all $i \in \{1,2\}\}$
    \item transitions:
    \begin{itemize}
      \item $t_1\times t_2 = (\ell_1,\ell_2) \xrightarrow{{a}}_p (\ell_1',\ell_2')$ ~~if~~ 
      $t_i=\ell_i \mathrel{\smash{\xrightarrow{{a}}}_p^i} \ell_i'$ for all $i \in \{1,2\}$
      
      \item $t_1\times t_2 = (\ell_1,\ell_2) \xrightarrow{{a},d!(v_1,v_2)}_p (\ell_1',\ell_2')$ 
      ~~if~~ $t_i=\ell_i \mathrel{\smash{\xrightarrow{{a},d!v_i}}_p^i} \ell_i'$ for all $i \in \{1,2\}$
      
      \item $t_1\times t_2 =(\ell_1,\ell_2) \xrightarrow{{a},d?(v_1,v_2)}_p (\ell_1',\ell_2')$ 
      ~~if~~ $t_i=\ell_i \mathrel{\smash{\xrightarrow{{a},d?v_i}}_p^i} \ell_i'$ for $i \in \{1,2\}$
    \end{itemize}
  \end{itemize}
  Finally, we let $\holds_{\sigma_1\vee\sigma_2}(t_1\times 
  t_2)=\max\{\holds_{\sigma_1}(t_1),\holds_{\sigma_2}(t_2)\}$.
  
  More generally, transducers are closed under product: given 
  $(\Sys_1,\gamma_1)$ and $(\Sys_2,\gamma_2)$, we construct 
  $\Sys=\Sys_1\times\Sys_2$ as above and we let $\gamma(t_1\times 
  t_2)=(\gamma_1(t_1),\gamma_2(t_2))$.
  
  Transducers are also closed under composition.

  \newpage
  
  \underline{Let us turn to the case of formulas $\existspath{\pi}\sigma$.}

  $\existspath{\pi}\sigma \equiv \existspath{\pi \Conc \test{\sigma}} \true$.
  Hence, we may assume that if $\existspath{\pi}\sigma$ appears as a subformula then $\sigma$ is $\true$.
  Furthermore, we simply denote it by $\existspath{\pi}$ (which means $\existspath{\pi}\true$). 
  
  \begin{example}
    Let us illustrate the idea by means of an example. Consider the \PDL path formula
    \[\pi ~=~ \bigl(\test{a} \cdot ({\procrel} + {\matchrel})\bigr)^\ast \cdot \test{b}\,.\]
    We translate $\pi$ into a finite automaton $\B$ over the alphabet
    $\{\matchrel,\procrel,\test{a},\test{b}\}$ as follows:
    \begin{center}
      \fbox{
      \includegraphics[scale=0.45]{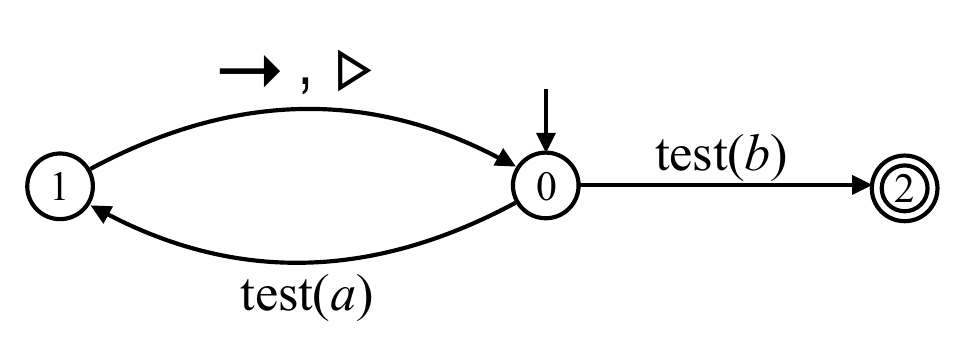}
      }
    \end{center}
    The \txtCPDS $\Sys_{\existspath{\pi}}$ will now label each event $e$ of a
    \txtCBM with the set of states from which one ``can reach a final state of 
    $\B$'', starting from event $e$. This computation proceeds backward
    starting from the maximal events wrt.\ $<$ (in the example, the only
    $b$-labeled one):
    \begin{center}
      \fbox{
      \includegraphics[scale=0.45]{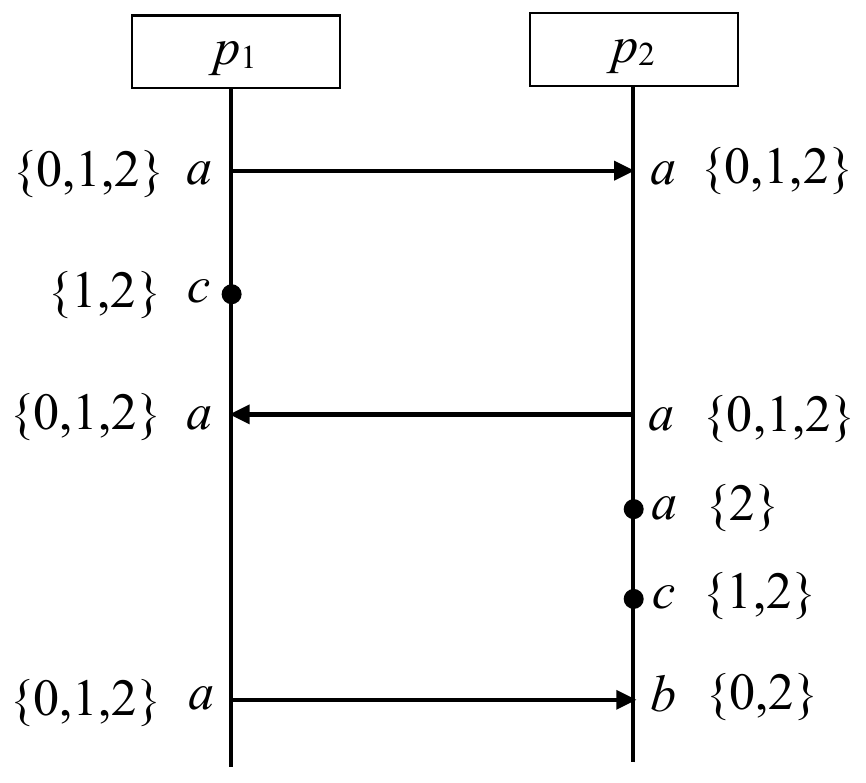}
      }
    \end{center}
    This can indeed be achieved by a \txtCPDS. To do so, the \txtCPDS has to
    ``inspect'', at each event $e$, the states at the immediate $\procrel$- and
    $\matchrel^d$-successors of $e$.  In particular, at a write event, it will
    have to guess what will be the state at the corresponding read event.
    Finally, an event satisfies $\existspath{\pi}$ iff the initial state $0$ is
    contained in the labeling.
  \end{example}
  
  \newpage
  
  \underline{\txtCPDS $\Sys_{\existspath{\pi}}$:}

Let $\text{Tests}(\pi) = \{\test{\sigma_1}, \ldots, \test{\sigma_n} \}$ be the set of tests appearing in $\pi$. 

Now, $\pi$ can be seen as a regular expression over the alphabet \[\Omega=\text{Tests}(\pi) \cup \{\matchrel^d \mid d \in \DS\} \cup \{\procrel\}\,.\]

Let $\B=(S,\delta,\iota,F)$ be a finite automaton over $\Omega$ for $\pi$, i.e., such that $L(\B) = L(\pi) \subseteq \Omega^\ast$. Note that we can assume $|S| = |\pi|$. Given $s \in S$, we set $\B_s=(S,\delta,s,F)$, i.e., $\B_s$ is essentially $\B$, but with new initial state $s$.

Let $\pi_s$ be a rational expression over $\Omega$ that is equivalent to $\B_s$ (in particular, $\pi_\iota = \pi$).

For a \txtCBM $\mscn=((w_p)_{p\in\Procs},(\matchrel^d)_{d\in\DS})$, we want to 
compute with a transducer the map $\nu: \Events \to 2^S$ such that
\[\nu(e) = \{s \in S \mid e \in \Sem{\existspath{\pi_s}}{\mscn}\}\,.\]

Let $\nu^+(e) =
\begin{cases}
\nu(f)  & \textup{if~} e \to f\\
\emptyset  & \textup{if~} e \text{ is maximal on its process.}
\end{cases}$

Let $H(e) = \{\test{\sigma_i} \mid e \in \Sem{\sigma_i}{\mscn}\} \subseteq \Omega$.

For $K \subseteq \Omega^\ast$ and $T \subseteq S$, let \[\delta^{-1}(K,T) = \{s \in S \mid \delta(s,w) \in T \text{ for some } w \in K\}\,.\]

\begin{lemma}\label{lem:nu}
\begin{itemize}
  \item[(i)] If $e$ is not a write event, then 
  \[\nu(e) = \delta^{-1}(H(e)^\ast, F \cup \delta^{-1}(\mathord{\to},\nu^+(e)))\,.\]

  \item[(ii)] If $e \matchrel^d f$, then 
  \[\nu(e) = \delta^{-1}(H(e)^\ast, F \cup \delta^{-1}(\mathord{\to},\nu^+(e)) 
  \cup \delta^{-1}(\mathord{\matchrel^d},\nu(f)))\,.\]
\end{itemize}
\end{lemma}

\begin{exercise}
Prove Lemma~\ref{lem:nu}.
\end{exercise}

\begin{remark}
If $e$ is maximal wrt.\ $<$, then $\nu(e) = \delta^{-1}(H(e)^\ast,F)$ can be computed directly.
\end{remark}

We are looking for a transducer $\Sys_\nu$, i.e., a \txtCPDS with output, which computes $\nu$.

\underline{Problem:} The computation of $\nu$ goes \emph{backward}, whereas a \txtCPDS run goes \emph{forward}.

\underline{Solution:} Guess $\nu$ nondeterministically and check afterwards whether the guess was correct.

We define $\Sys_\nu = \underbrace{(\Sys_{\sigma_1} \times \cdots
\times \Sys_{\sigma_n})}_{\substack{\text{computes $H$}\\\text{as
input for }\B}} {}\cdot \B$.  By induction, the
$(\Sys_{\sigma_i})_{1 \le i \le n}$ are given.

\newpage

The product transducer $\Sys=\Sys_{\sigma_1} \times \cdots\times \Sys_{\sigma_n}$ 
outputs letters in $\mathbb{B}^{n}=\{0,1\}^{n}$.

Let $\rho=\rho_1\times\cdots\times\rho_n$ be an accepting run of $\Sys$ on a \CBM $\mscn$. 

For all events $e\in\Events$, we have $\gamma(\rho(e))=(\gamma_1(\rho_1(e)),\ldots,\gamma_n(\rho_n(e)))$.

For $G=(G_1,\ldots,G_n)\in\mathbb{B}^{n}$, define $\overline{G}=\{\test{\sigma_i}\mid G_i=1\}$.

By induction hypothesis, we have $\gamma_i(\rho_i(e))=1$ iff $\mscn,e\models\sigma_i$.

Therefore, $\overline{\gamma(\rho(e))}=H(e)$.

The transducer $\B$ has input alphabet $\mathbb{B}^{n}$ and output alphabet 
$2^{S}$. Since we use future modalities, the transducer $\B$ must be 
non-determinitic. But we construct a transducer which is backward deterministic 
and backward complete. Since it is backward deterministic, it has a unique 
global final state, but it uses a set of global initial states.
This is a generalization compared with single local initial states, but can be
simulated with nondeterminism.

We define $\B=(\Locs,\Val,(\to_p)_{p\in\Procs},\mathsf{Init},\FinLocs)$ as follows:
\begin{itemize}
  \item $\Locs = 2^S$.
  A set $U\in\Locs$ represents $\nu(e)$.

  \item $\Val = 2^S$.
  Here, $W\in\Val$ represents $\nu(f)$ when $e \matchrel^d f$.

  \item $\mathsf{Init} = (2^S)^{\Procs}$.

  \item $\FinLocs = \{\emptyset\}^{\Procs} = \{(\emptyset,\ldots,\emptyset)\}$.
\end{itemize}

Now, we turn to the transitions.
\begin{align*}
  & U \xrightarrow{G}_p V 
  && \text{if } U = \delta^{-1}(G^\ast, F \cup \delta^{-1}(\mathord{\to},V)) 
  \\
  & U \xrightarrow{G,d!W}_p V
  && \text{if } U = \delta^{-1}(G^\ast, F \cup \delta^{-1}(\mathord{\to},V) \cup 
  \delta^{-1}(\mathord{\matchrel}^{d},W))
  \tag{$p=\writer(d)$}
  \\
  & U \xrightarrow{G,d?U}_p V
  && \text{if } U = \delta^{-1}(G^\ast, F \cup \delta^{-1}(\mathord{\to},V))
  \tag{$p=\reader(d)$}
\end{align*}

Finally, we set $\gamma_\B(t)=\mathsf{src}(t)$.

Define $\Sys_{\nu}$ as the composition
$(\Sys_1\times\cdots\times\Sys_n)\cdot\B$.

\begin{lemma}
  Let $\mscn$ be a \CBM and let 
  $\rho=\rho_1\times\cdots\times\rho_n\times\rho_\B$ be an accepting run of
  $\Sys_{\nu}$ on $\mscn$. Then, for all events 
  $e\in\Events$, we have 
  $$
  \gamma_{\nu}(\rho(e)) = 
  \gamma_\B(\rho_\B(e)) = \nu(e) \,.
  $$
\end{lemma}

\begin{proof}
  The proof is by induction on the partial order $<$ of $\mscn$, starting
  from the maximal events. It is based on Lemma~\ref{lem:nu}.
\end{proof}

Finally, we obtain $\Sys_{\existspath{\pi}}$ as the composition
$\Sys_\nu\cdot\C$ where $\C$ is the universal transducer with one state which 
transforms an input letter $T\subseteq S$ into the boolean value
$$
\begin{cases}
1  & \textup{if~} \iota \in T\\
0  & \text{otherwise.}
\end{cases}
$$

\newpage

It remains to define automata for PDL formulas $\Phi \in \PDL(\Arch,\Sigma)$.
Without loss of generality, we assume that $\Phi$ is a positive boolean combination of formulas of the form
$\Existsev \sigma$ or $\Forallev \sigma$. Disjunction and conjunction are easy to handle, since \CPDSs are closed under union and intersection (exercise).

For the case $\Forallev \sigma$, suppose we already have
$\Sys_{\sigma}=(\Locs,\Val,(\to_p)_{p\in\Procs},\inLoc,\FinLocs)$ with
associated mapping $\holds_\sigma\colon\Trans \to \{0,1\}$.
Then, the \txtCPDS for $\Forallev \sigma$ is simply the restriction of $\Sys$ 
to those transitions $t$ such that $\holds_\sigma(t) = 1$.

The case $\Existsev \sigma$ is left as an exercise.
\end{proof}

\begin{exercise}
Consider the following extension of $\PDL(\Arch,\Act)$, which we call $\PDL^{-1}(\Arch,\Act)$:
\begin{align*}
  \Phi &::= \Existsev \sigma \mid \Phi \vee \Phi \mid \neg \Phi
  \\
  \sigma &::= p \mid a \mid \sigma \vee \sigma \mid \neg \sigma \mid \existspath{\pi}{\sigma} \mid \textcolor{red}{\existspath{\pi^{-1}}\sigma}
  \\
  \pi &::= {\matchrel^d} \mid {\procrel} \mid \test{\sigma} \mid 
  \pi + \pi \mid \pi \cdot \pi \mid \pi^{\ast} 
\end{align*}
where $p \in\Procs$, $d \in \DS$ and  $a\in\Act$.
Show that Theorem~\ref{PDLtoCPDS} even holds for the logic $\PDL^{-1}(\Arch,\Act)$:\\
For every $\Phi \in \PDL^{-1}(\Arch,\Sigma)$, there is $\Sys \in \CPDS(\Arch,\Sigma)$ such that $\Lang(\Sys) = \cbmL(\Phi)$.
\end{exercise}

\begin{remark}
  Since $\PDL$ is closed under complementation (negation), while \CPDSs are not
  (for certain architectures), we obtain, as a corollary, that \CPDSs are
  strictly more expressive than \PDL.
\end{remark}


\bibliographystyle{alpha}
\bibliography{mylit}


\end{document}